%% file: manuscript.tex
\documentclass[final,smallextended,nospthms]{svjour3}
\smartqed
\usepackage{algorithm}
\usepackage{subfigure}
\usepackage{graphicx}
\usepackage{hyperref}
\usepackage{dsfont}
\usepackage{pgf}
\usepackage{amssymb}
\usepackage{hyperref}
\usepackage{xcolor}
\usepackage{enumitem}
\usepackage{amsmath}
\usepackage{amsopn}
\usepackage{mathtools}
\usepackage{mathptmx}
\usepackage{amsfonts}
\usepackage{mathrsfs}
\usepackage{amsthm,bm}

\newtheorem{theorem}{Theorem}[section]
\newtheorem{proposition}[theorem]{Proposition}
 
\newtheorem{remark}[theorem]{Remark}
\newtheorem{lemma}[theorem]{Lemma}

\newtheorem{conjecture}[theorem]{Conjecture}

\newcommand{\var}{\operatorname{var}}
\newcommand{\cov}{\operatorname{cov}}
\newcommand{\floor}[1]{\left\lfloor #1 \right\rfloor}

\newcommand{\blog}{\log_2}

\newcommand{\symdif}{\triangle}
\newcommand{\sA}{\mathcal{A}} 
\newcommand{\sB}{\mathcal{B}} 
\newcommand{\sD}{\mathcal{D}}
\newcommand{\sE}{\mathcal{E}} 
\newcommand{\sF}{\mathcal{F}} 
\newcommand{\sG}{\mathcal{G}}
 
\newcommand{\sP}{\mathcal{P}}
\newcommand{\sM}{\mathcal{M}}
\newcommand{\sN}{\mathcal{N}}
 
\newcommand{\sT}{\mathcal{T}} 
\newcommand{\sV}{\mathcal{V}} 
\newcommand{\sW}{\mathcal{W}}
\newcommand{\RR}{\mathbb{R}} 
\newcommand{\ZZ}{\mathbb{Z}} 
\newcommand{\NN}{\mathbb{N}} 
\newcommand{\PP}{\mathbb{P}} 
\newcommand{\EE}{\mathbb{E}}
\newcommand{\naturals}{\NN}
\newcommand{\posint}{\NN^+}
\renewcommand{\d}{\mathrm{d}\,}
\newcommand{\pc}{p_{\mathrm{c}}}
\newcommand{\qc}{q_{\ast}}
\newcommand{\betac}{\beta_{\mathrm{c}}}
\newcommand{\backwardCouplingTime}{\mathfrak{T}}
\newcommand{\indicator}{\mathbf{1}}
\newcommand{\<}{\langle}
\newcommand{\mix}{t_\mathrm{mix}}
\newcommand{\rel}{t_\mathrm{rel}}

\newcommand{\expauto}{t_{\exp}}
\newcommand{\coupling}{T}
\newcommand{\coupon}{W}
\newcommand{\standard}{S}
\newcommand{\secondeigenvalue}{\lambda_2}
\renewcommand{\top}{\sT}
\newcommand{\bottom}{\sB}
\newcommand{\TV}{\mathrm{TV}}
\newcommand{\lastvisit}{H}

\newcommand{\fk}{\phi}
\newcommand{\ising}{\pi}
\renewcommand{\>}{\rangle}
\numberwithin{equation}{section}
\begin{document}
\title{On the coupling time of the heat-bath process for the Fortuin-Kasteleyn random-cluster model}
\author{Andrea Collevecchio \and Eren Metin El\c{c}i \and Timothy M. Garoni \and Martin Weigel}
\authorrunning{Collevecchio et al.} 
\institute{Andrea Collevecchio \and Eren Metin El\c{c}i \at
  School of Mathematical Sciences, Monash University, Clayton, VIC, 3800, Australia\\
  \and 
  Timothy M. Garoni \at
ARC Centre of Excellence for Mathematical and Statistical Frontiers (ACEMS), School of Mathematical Sciences, Monash University, Clayton, VIC 3800, Australia
\\
  \and 
  Martin Weigel \at Applied Mathematics Research Centre, Coventry University, Coventry, CV1 5FB, United Kingdom\\
  \email{andrea.collevecchio@monash.edu}\\
  \email{elci@posteo.de}\\
  \email{tim.garoni@monash.edu}\\
  \email{martin.weigel@coventry.ac.uk}
}

\date{\today}

\maketitle
\begin{abstract}
We consider the coupling from the past implementation of the random-cluster heat-bath process, and study its random running time, or
\emph{coupling time}.  We focus on hypercubic lattices embedded on tori, in dimensions one to three, with cluster fugacity at least one. We
make a number of conjectures regarding the asymptotic behaviour of the coupling time, motivated by rigorous results in one dimension and
Monte Carlo simulations in dimensions two and three.  Amongst our findings, we observe that, for generic parameter values, the distribution
of the appropriately standardized coupling time converges to a Gumbel distribution, and that the standard deviation of the coupling time is
asymptotic to an explicit universal constant multiple of the relaxation time.  Perhaps surprisingly, we observe these results to hold both
off criticality, where the coupling time closely mimics the coupon collector's problem, and also \emph{at} the critical point, provided the
cluster fugacity is below the value at which the transition becomes discontinuous.  Finally, we consider analogous questions for the
single-spin Ising heat-bath process.

\end{abstract}
\section{Introduction}
\label{sec:Introduction}
Since nontrivial models in statistical mechanics are rarely exactly solvable, Monte Carlo simulations provide an important tool for
obtaining information on phase diagrams and critical exponents.  The standard Markov-chain Monte Carlo procedure involves constructing a
Markov chain with the desired stationary distribution, and then running the chain long enough that the resulting samples
are close to stationarity.  The central obstacle to practical applications of MCMC is that it is typically not known \emph{a priori} how
many steps are required in order to reach (approximate) stationarity. In principle, the answer to this question can be quantified by
quantities such as the \emph{relaxation time} or \emph{mixing time} of the Markov chain (see below).  However, rigorously proving
practically useful upper bounds on such quantities is a very challenging task, as is their empirical estimation from simulations.

Coupling From The Past (CFTP), introduced by Propp and Wilson~\cite{ProppWilson96}, is a refinement of the MCMC method, which automatically
determines the required running time of the Markov chain, and then outputs exact samples, rather than approximate ones. The price that must be paid
for these two significant benefits is that, unlike naive MCMC, the running time of CFTP is random. The key question in determining the
efficiency of the CFTP method for a given application therefore becomes to understand the distribution of its random running time, or
\emph{coupling time}.  The name ``coupling from the past'' derives from two key features of the method. Firstly, rather than running a
single Markov chain, CFTP requires multiple Markov chains be run simultaneously (coupling). Secondly, the chains are not run forward from
time 0, but are instead run from the past to time 0.

In this article, we present a detailed study of the coupling time for the heat-bath dynamics of the Fortuin-Kasteleyn (FK) random-cluster model.
This process is one of the examples originally considered in~\cite{ProppWilson96}, and has been the subject of several recent
studies~\cite{BlancaSinclair16,GuoJerrum16,ElciWeigel13,ElciWeigel14,DengGaroniSokal07_sweeny}.  As discussed in more detail below, when the
cluster fugacity $q\ge1$, this process possesses an important monotonicity property, which makes it an ideal candidate for an efficient
implementation of CFTP. 

We consider the FK process on $d$-dimensional tori, $\ZZ_L^d$, for $d=1,2,3$.  Our methods are a combination of rigorous proof for $d=1$,
and systematic Monte Carlo experiments for $d=2,3$.  Based on our studies, we conjecture a number of results for the coupling time, which we
state precisely in Section~\ref{subsec:results}. Among them, we conjecture that, for generic choices of parameters $(p,q)$, the distribution
of the coupling time (appropriately standardized) tends to a Gumbel distribution as $L\to\infty$. For the special case of $q=1$, the
coupling time corresponds precisely to the coupon collector's problem, for which the Gumbel limit is a classical result~\cite{ErdosRenyi61}.
The surprising observation is that such a limit appears not only to remain universally valid for the FK heat-bath process at any
off-critical choice of $(p,q)\in(0,1)\times[1,\infty)$, but also \emph{at} the critical point, provided $q$ is below the value at which the
transition becomes discontinuous.  In particular, we conjecture that this limit law holds for all $p\in(0,1)$ when $q\in[1,4)$ and $d=2$.

In addition, we find strong evidence that the standard deviation of the coupling time is asymptotic, as $L\to\infty$, to a universal
constant times the relaxation time. Again, this is conjectured to hold not only off criticality for arbitrary $p\neq\pc$ and $q\ge1$, but
also \emph{at} $p=\pc$, provided $q$ is below the value at which the transition becomes discontinuous.  If true, this result suggests an
efficient empirical method for estimating the relaxation time of the FK heat-bath process: simply generate a number of independent
realizations of the coupling time and compute the sample variance.  We emphasize that this result would imply that consideration of the
coupling time can provide non-trivial information about the original Markov chain, and so its significance extends beyond possible
applications of the CFTP method, to standard MCMC simulations of the FK heat-bath chain.

For comparison, we also briefly study the single-spin-update heat-bath process for the Ising model.  Due to the slow mixing in the low
temperature phase~\cite{CesiGuadagniMartinelliSchonmann96}, our numerical results focus on the critical and high temperature regimes.  In
the high temperature regime, we find identical behaviour to that described above for the FK heat-bath process; in particular we find the
same Gumbel limit law for the coupling time, and the same relationship between the relaxation time and the standard deviation of the
coupling time. At criticality, however, the situation changes somewhat.  The relaxation time and coupling time standard deviation do still
appear to be asymptotically proportional, but now with a different proportionality constant.  Moreover, while the standardized coupling time
again appears to have a non-degenerate limit at criticality, the limit appears not to be of Gumbel type in this case.

\subsection{Outline}
Let us outline the remainder of this article. In Section~\ref{FK heat-bath process definitions} we define the FK heat-bath process, and
discuss some relevant recent literature.  We also define the coupling time, and explain its connection to CFTP.
Section~\ref{subsec:results} summarizes our theorems and conjectures for the FK coupling time. Sections~\ref{sec:fk moments} and~\ref{sec:fk
  limiting distribution} respectively consider the moments and limiting distributions, and present numerical evidence to support the
conjectures outlined in Section~\ref{subsec:results}. Sections~\ref{sec:FK arbitrary graphs} and~\ref{sec:FK one dimension} provide proofs
of Theorems~\ref{thm:FK bounds on arbitrary graphs} and~\ref{thm: d=1}, respectively. Section~\ref{sec:Ising} summarizes the analogous
results for the single-spin-update Ising heat-bath process.  Finally, Appendix~\ref{appendix:autocorrelation functions of increasing
  observables} establishes some relevant properties of autocorrelation functions of the FK heat-bath process, which we make use of in
Section~\ref{sec:fk moments}, and Appendix~\ref{appendix:coupon collecting} discusses some technical lemmas concerning the coupon
collector's problem.

\section{Fortuin-Kasteleyn heat-bath process}
\label{FK heat-bath process definitions}
\subsection{Definitions}
\label{subsec:definitions}
The Fortuin-Kasteleyn random-cluster model is a correlated bond percolation model, which can be defined on an arbitrary finite graph $G=(V,E)$ 
with parameters $p\in[0,1]$ and $q>0$ via the measure
\begin{equation}
\fk(A) = \dfrac{1}{Z_G(p,q)}\,q^{k(A)}\,p^{|A|}\,(1-p)^{|A^c|}, \qquad A\subseteq E
\label{FK distribution}
\end{equation}
where $k(A)$ is the number of connected components (\emph{clusters}) in the spanning subgraph $(V,A)$.  The partition function, $Z_G(p,q)$
is closely related to the Tutte polynomial, and its computation is known to be a \#P-hard problem~\cite{Welsh93,JaegerVertiganWelsh90}.  For
$q=1$, the FK model coincides with standard bond percolation, while for integer $q>1$ it is intimately related to the $q$-state Potts
model. Appropriate limits as $q\to0$ also coincide with spanning forests and uniform spanning trees.

While our focus in the current article is on finite graphs, standard arguments (see e.g.~\cite{Grimmett06}) allow random-cluster measures to
be defined\footnote{For concreteness, in the present discussion we refer to the measure corresponding to wired boundary
  conditions~\cite[Section 4.2]{Grimmett06}.} on the infinite lattice $\ZZ^d$. In this setting, it is well known~\cite{Grimmett06} that for
given $q\ge1$ and $d\ge2$, there exists a critical probability $\pc\in(0,1)$, such that the origin belongs to an infinite cluster with zero
probability when $p<\pc$, and with strictly positive probability when $p>\pc$. The exact value of $\pc$ when $d=2$ was recently
proved~\cite{BeffaraDuminilCopin12} to be $\pc=\sqrt{q}/(1+\sqrt{q})$.  The corresponding phase transition is said to be continuous if there
is zero probability that the origin belongs to an infinite cluster at $p=\pc$, and is discontinuous otherwise. It is
known~\cite{LaanaitMessagerMiracleSoleRuizShlosman91} that the transition is discontinuous for sufficiently large $q$.  It is
conjectured~\cite[Conjecture 6.32]{Grimmett06} that for every $d\ge2$ there exists $\qc$ such that the transition is continuous for $q<\qc$
and discontinuous for $q>\qc$. This has recently been proved when $d=2$, and moreover the exact value $\qc=4$ was established, confirming a
longstanding conjecture of Baxter~\cite{Baxter78}.  More precisely, in the specific case of $d=2$, the phase transition is
continuous~\cite{DuminilCopinSidoraviciusTassion15} for $1\le q\le 4$, and discontinuous~\cite{DuminilCopinGagnebinHarelManolescuTassion16} for
$q>4$.  Although $\pc=\pc(q,d)$ depends on $d$ and $q$, and $\qc$ depends on $d$, for brevity, we shall not explicitly write this dependence
when the values of $q,d$ are clear from the context.

To ease notation, for $A\subseteq E$ and $e\in E$, let $A_e:=A\setminus e$ and $A^e:=A\cup e$. Note that $A=A^e$ iff $e\in A$, and $A=A_e$
iff $e\not\in A$. An edge $e\in A$ is said to be \emph{occupied} in $A$.  An edge $e\in E$ is said to be \emph{pivotal} to the configuration
$A$ if $k(A_e)\neq k(A^e)$.

The FK heat-bath process has transition matrix $P=\dfrac{1}{m}\sum_{e\in E} P_e$ where
\begin{equation}
\begin{split}
P_e(A,B) &:= 
\begin{cases}
p(A,e), & B=A^e,\\
1-p(A,e), & B=A_e,\\
0, & \text{otherwise},
\end{cases}
\label{transition matrix}
\\
p(A,e) &:=
\frac{\fk(A^e)}{\fk(A^e)+\fk(A_e)}
=
\begin{cases}
\tilde{p}, & e \text{ is pivotal to } A, \\
p, & \text{otherwise},
\end{cases}
\\
\tilde{p}&:=\dfrac{p}{1+(q-1)(1-p)}.
\end{split}
\end{equation}
Note that, if $q\ge1$ and ${p\in(0,1)}$, we have $\tilde{p}\le p$, with equality iff $q=1$. 

We now proceed to define the central quantity of interest in this article, the coupling time of the FK heat-bath process.
It should be emphasized that the coupling time, and the corresponding CFTP algorithm, are not uniquely determined by the transition
probabilities of the process, but rather by the particular \emph{random mapping representation} that is chosen. Random mapping representations for
Markov chains provide convenient methods for constructing useful couplings, and also for constructing practical computational
implementations~\cite{LevinPeresWilmer09}. 

We focus attention on the following random mapping representation for $P$. Define $f: 2^E\times E\times [0,1]\to 2^E$ via
\begin{equation}
  f(A,e,u) := 
  \begin{cases}
    A^e, & u \le p(A,e),\\
    A_e, & u > p(A,e).
    \end{cases}
\label{random mapping definition}
\end{equation}
Let $\sE$ and $U$ be independent, with $\sE$ uniform on $E$ and $U$ uniform on $[0,1]$. By construction, $\PP(f(A,\sE,U)=B)=P(A,B)$, and so
$(f,\sE,U)$ defines a random mapping representation for $P$~\cite{LevinPeresWilmer09}.  It is straightforward to verify that $f$ is
\emph{monotonic}: for any fixed $e\in E$ and $u\in[0,1]$, if $A\subseteq B$, then $f(A,e,u)\subseteq f(B,e,u)$.  This random mapping
representation corresponds precisely to the manner in which a computational physicist would implement the transition matrix $P$ in practice.

Let $(\sE_t,U_t)_{t\in\posint}$ be an iid sequence\footnote{We adopt the convention that $\naturals:=\{0,1,2,\ldots\}$ and
  $\posint:=\{1,2,\ldots\}$.}  of copies of $(\sE,U)$.  Define $\top_t$ by $\top_0=E$ and $\top_{t+1}=f(\top_t,\sE_{t+1},U_{t+1})$. We refer
to $\top_t$ as the \emph{top chain}.  Likewise, the \emph{bottom chain} is defined by $\bottom_0=\emptyset$ and
$\bottom_{t+1}=f(\bottom_t,\sE_{t+1},U_{t+1})$. By construction, both $(\top_t)_{t\in\naturals}$ and $(\bottom_t)_{t\in\naturals}$ are
Markov chains with transition matrix $P$.  The coupled process $(\sB_t,\sF_t)_{t\in\naturals}$ is the fundamental object of consideration in
this article. For brevity, in what follows, we will refer to the coupled process $(\sB_t,\sF_t)_{t\in\naturals}$ as ``the FK heat-bath coupling''. 

We define the \emph{coupling time} of the FK heat-bath process to be 
\begin{equation}
\coupling:=\min\{t\in\naturals : \top_t = \bottom_t\}.
\label{coupling time definition}
\end{equation}
Note that, strictly speaking, the coupling time is a property of the FK heat-bath coupling, rather than of a single FK heat-bath process. 
Also note that, by monotonicity, a Markov chain started at time 0 in any state $A\subseteq E$ will have coalesced with $\top_t$ and $\bottom_t$ by time
$t=\coupling$, so $\top_{\coupling}$ can be viewed as the state of the Markov chain at the first time in which the initial state has been
\emph{forgotten} by the above coupling. As discussed further in Section~\ref{subsec:cftp}, the coupling time has the same distribution as
the running time of the CFTP algorithm.

\subsection{Previous studies of FK Glauber processes}
\label{subsec:previous studies}
A reversible Markov chain with stationary distribution~\eqref{FK distribution}, which is \emph{local} in the sense that at most one edge is
updated per time step, is typically referred to as a \emph{Glauber process} for the FK model. The two most commonly studied Glauber
processes for the FK model are the heat-bath process, as studied here, and the Metropolis process, as first studied numerically
in~\cite{Sweeny83}.

As a consequence of general results concerning heat-bath chains~\cite{DyerGreenhillUllrich14}, the transition matrix of the FK heat-bath
process, $P$, has non-negative eigenvalues. If $\secondeigenvalue$ denotes the second-largest eigenvalue of $P$, the \emph{relaxation
  time}~\cite{LevinPeresWilmer09} of $P$ is
 \begin{equation}
 \rel := \frac{1}{1-\secondeigenvalue}.
 \label{relaxation time definition}
 \end{equation}
A closely related quantity is the \emph{exponential autocorrelation time}~\cite{MadrasSlade96,Sokal97}, defined by
\begin{equation}
\expauto:=\frac{-1}{\ln(\secondeigenvalue)} = \frac{-1}{\ln(1-1/\rel)}.
\label{exponential autocorrelation time definition}
\end{equation}
It is easily verified that
\begin{equation}
\rel - 1 \le \expauto \le \rel.
\label{bound relating expauto and rel}
\end{equation}

Another quantity of importance is the \emph{mixing time}~\cite{LevinPeresWilmer09}, defined by
\begin{equation}
\mix(\epsilon) := \max_{A\subseteq E}\, \min_{t\in\naturals}\,\{\|P^t(A,\cdot)-\fk\|_{\TV} \le \epsilon\},
\label{mixing time definition}
\end{equation}
where $\|\cdot\|_{\TV}$ denotes total variation distance. Since $\mix(\epsilon) \le \lceil\blog\epsilon^{-1}\rceil\,\mix(1/4)$, one also
defines $\mix:=\mix(1/4)$~\cite{LevinPeresWilmer09}. Combining~\cite[Theorem 12.3]{LevinPeresWilmer09} 
and~\cite[Theorem 12.4]{LevinPeresWilmer09} with Lemma~\ref{lem: mu min} implies that for the FK heat-bath process
\begin{equation}
\frac{\rel-1}{2} \le \mix \le \ln\left(\frac{4q^2}{p(1-p)}\right)\,m\,\rel.
\end{equation}
The quantities $\rel$, $\expauto$ and $\mix$ all quantify the rate at which a Markov chain approaches stationarity, of
\emph{mixes}~\cite{LevinPeresWilmer09}.

Numerical studies~\cite{Gliozzi02,WangKozanSwendsen02,DengGaroniSokal07_sweeny} of FK Glauber processes suggest that their mixing in the
neighbourhood of continuous phase transitions can be surprisingly efficient; comparable to, and possibly faster than, non-local cluster
algorithms such as the Swendsen-Wang and Chayes-Machta processes~\cite{SwendsenWang87,ChayesMachta98}. In addition, it was observed
numerically in~\cite{DengGaroniSokal07_sweeny} that, for the FK Metropolis-Glauber process at criticality on the square and simple-cubic
lattices, certain observables apparently decorrelate asymptotically faster than a single sweep (i.e. in time $o(|E|)$), suggesting FK
Glauber processes could have significant advantages over cluster algorithms.

Significant progress has recently been made in rigorously bounding the mixing time of FK Glauber processes.  In~\cite{GuoJerrum16}, the
mixing time of the $q=2$ FK Metropolis-Glauber process on a graph with $m$ edges and $n$ vertices was shown to be $O(n^4 m^3)$. In addition,
precise asymptotics were given in~\cite{BlancaSinclair16} for the case of $q\ge1$ on $L\times L$ boxes in $\ZZ^2$, showing\footnote{The
  notation $a_L\asymp b_L$ means that there exist constants $c,C>0$ such that $c b_L \le a_L \le C b_L$ for all sufficiently large $L$.}
that $\mix \asymp L^2 \ln L$, provided $p\neq\pc$.  Even more recently, it has been shown in~\cite{GheissariLubetzky16} that
$\mix=O(L^{\ln L})$ on two-dimensional tori $\ZZ_L^2$ at $p=\pc$.

An important practical issue when simulating FK Glauber processes is the need to identify whether the edge to be updated is pivotal to the
current edge configuration. Sweeny~\cite{Sweeny83} proposed an algorithm for performing the necessary connectivity checks, which was
applicable to planar graphs.  In \cite{Elci15_thesis,ElciWeigel13,ElciWeigel14}, it was demonstrated that this algorithmic problem can be
efficiently solved by utilizing, and adapting, dynamic connectivity algorithms and appropriate data structures introduced
in~\cite{HolmDeLichtenbergThorup01}. These latter methods are applicable to arbitrary graphs, and can perform the required pivotality tests in time which
is poly-logarithmic in the graph size.

\subsection{Coupling from the past}
\label{subsec:cftp}
For completeness, in this section we present a brief review of the CFTP method applied to the FK heat-bath process. We note however that the
material in this section, which follows the discussion in~\cite{ProppWilson96}, serves only as motivation for studying the coupling
time~\eqref{coupling time definition}, and none of the concepts introduced in this section will be required outside of this section.

Let $(\sE_t,U_t)_{t\ge 0}$ be an iid sequence of copies of
$(\sE,U)$, define random maps $f_{-t} := f(\cdot,\sE_t,U_t)$, and for $t\in\posint$ form the compositions
\begin{equation}
F_{-t}^0 := f_0 \circ f_{-1} \circ \ldots \circ f_{-(t-1)}.
\label{backward composition}
\end{equation}
We can then define the \emph{backward coupling time} to be 
\begin{equation}
\backwardCouplingTime := \min\{t\in\posint : F_{-t}^0(E) = F_{-t}^0(\emptyset)\}.
\label{backward coupling time definition}
\end{equation}
As first shown in~\cite{ProppWilson96}, the random state $F_{-\backwardCouplingTime}^0(E)=F_{-\backwardCouplingTime}^0(\emptyset)$ is an
exact sample from the FK distribution~\eqref{FK distribution}. Algorithmically, a single step of the above procedure corresponds to starting
chains in states $E$ and $\emptyset$ at some point in the past, and running them until time 0. This procedure is then applied iteratively,
starting the chains at ever more distant times in the past, and terminating the iteration at the first time that the chains started at $E$ and
$\emptyset$ agree at time 0.

To appreciate why the resulting state $F_{-\backwardCouplingTime}^0(E)$ is distributed according to~\eqref{FK distribution}, we can make the
following observations. Firstly, by monotonicity, if $F_{-t}^0(E)=F_{-t}^0(\emptyset)$ then $F_{-t}^0(A)=F_{-t}^0(E)$ for every $A\subseteq
E$. Secondly, if $F_{-t}^0(E)=F_{-t}^0(\emptyset)$ then $F_{-s}^0(E)=F_{-s}^0(\emptyset)$ for every $s\ge t$. Therefore, the state
$F_{-\backwardCouplingTime}^0(E)$ coincides with $F_{-s}^0(A)$ for any $s\ge\backwardCouplingTime$ and $A\subseteq E$.  In this sense, we can picture 
$F_{-\backwardCouplingTime}^0(E)$ as the state,
at time $t=0$, of a Markov chain that started at an arbitrary state in the infinite past.

For comparison, note that performing a standard Markov-chain Monte Carlo simulation simply corresponds to composing the sequence of random maps in the
opposite order to~\eqref{backward composition}. Specifically, to defining random maps $f_{t} := f(\cdot,\sE_t,U_t)$ and forming the compositions
$$ 
F_{0}^t := f_{t} \circ \ldots \circ f_{1}.
$$ 
Even though, by monotonicity, we have $F_{0}^{\coupling}(E)=F_0^{\coupling}(A)=F_{0}^{\coupling}(\emptyset)$ for all $A\subseteq E$, there
is no reason to suspect $F_0^{\coupling}(E)$ should have distribution~\eqref{FK distribution}.

Despite the significant differences between the forward and backward couplings, it can be shown, quite generally, that forward and backward
coupling times are identically distributed~\cite{ProppWilson96}. As a consequence, to study the behaviour of the random running time
$\backwardCouplingTime$ of CFTP, it suffices to consider only the forward coupling time $\coupling$, defined in~\eqref{coupling time
  definition}.

The CFTP algorithm described above is the simplest version, however a number of algorithmic improvements have been devised. In particular,
rather than choosing the restart times to be $-1,-2,-3,\ldots$, the restart times can be chosen to be $-a_1,-a_2,\ldots$ for any monotonic
natural sequence $a_1,a_2,\ldots$. See the pedagogical discussions in~\cite{Jerrum98,Haggstrom03,LevinPeresWilmer09} for more details on
CFTP algorithms.

\subsection{Behaviour of the coupling time}
\label{subsec:results}
We now summarize our main results for the coupling time.  We begin with some general results, holding on arbitrary finite connected graphs,
which relate the coupling time~\eqref{coupling time definition} to $\mix$ and $\expauto$. Theorem~\ref{thm:FK bounds on arbitrary graphs} is
a slight refinement, in the specific setting of the FK heat-bath coupling, of the results presented in~\cite[Section 5]{ProppWilson96}.  Its
proof is deferred until Section~\ref{sec:FK arbitrary graphs}.
\begin{theorem} Consider the FK heat-bath coupling with parameters $p\in(0,1)$ and $q\ge1$, on a finite connected graph with $m\ge1$ edges, and let
$\psi:=\psi(p,q):=\dfrac{q^2}{p(1-p)}$. Then
 \begin{align}
 \frac{e^{-t/\expauto}}{2} \le \PP(\coupling>t) &\le e^{(\ln(\psi)+2)\,m - t/\expauto},\label{tail distribution is exp with scale expauto}\\
 \frac{\mix-1}{4} \le \EE(\coupling) &\le \min\left\{12\,\blog(4m)\,\mix,\,4(\blog(\psi)+3)\,m\,\expauto\right\},\\
 \sqrt{\var(\coupling)} &\le \min\left\{15\,\blog(4m)\,\mix,\,5(\blog(\psi)+3)\,m\,\expauto\right\}
 \end{align}
\label{thm:FK bounds on arbitrary graphs}
\end{theorem}
\begin{remark}
\label{rem:d=2 bounds}
In the special case of $L\times L$ boxes in $\ZZ^2$, with $p\neq \pc$, we can combine the mixing time bound presented
in~\cite{BlancaSinclair16} with Theorem~\ref{thm:FK bounds on arbitrary graphs} to conclude that both $\EE(\coupling)$ and
$\sqrt{\var(\coupling)}$ are $O(L^2\,\ln^2L)$, and that $\EE(\coupling)$ is $\Omega(L^2\ln L)$. 
Likewise, the results in~\cite{GheissariLubetzky16} imply that, at $p=\pc$, both $\EE(\coupling)$ and $\var(\coupling)$ are $L^{O(\ln L)}$ 
on $\ZZ_L^2$. 
\end{remark}

As mentioned briefly in Section~\ref{sec:Introduction}, the coupling time is related to the coupon collector's problem. We now make this
connection more precise. Consider a finite connected graph $G=(V,E)$ with $|E|=m$ and let
\begin{equation}
\coupon:=\min\{t\in\posint : \{\sE_1,\ldots,\sE_t\}=E\}.
\label{coupon definition}
\end{equation}
The random variable $\coupon$ is the \emph{coupon collector's time}, for the edge process $(\sE_t)_{t\in\posint}$, and its behaviour is
well-understood~\cite{ErdosRenyi61,LevinPeresWilmer09}.  It is elementary to show (see e.g.~\cite{Posfai10}) that
\begin{align}
\EE(\coupon) & = m\, H_m\,\sim m\ln(m), 
\label{coupon collector mean}\\
\var(\coupon) &= m^2 H_m^{(2)} - m\, H_m \sim \frac{\pi^2}{6} m^2, 
\label{coupon collector variance}
\end{align}
as $m\to\infty$, where $H_m^{(k)}:=\sum_{i=1}^m i^{-k}$ is the generalized Harmonic number~\cite{GrahamKnuthPatashnik94} of order $k$, and
$H_m:=H_m^{(1)}$.  Moreover, as first shown in~\cite{ErdosRenyi61}, for any $x\in\RR$ we have
\begin{equation}
\lim_{m\to\infty} \PP[\coupon\le \EE(\coupon)+x\sqrt{\var(\coupon)}] = G(x),
\label{coupon collector CLT}
\end{equation}
where
\begin{equation}
G(x):=\exp\left(-\exp\left(- \dfrac{\pi}{\sqrt{6}}\,x - \gamma\right)\right), \quad x\in \RR,
\label{gumbel distribution function}
\end{equation}
is the distribution function of the Gumbel distribution with zero mean and unit variance, and $\gamma$ is the Euler-Mascheroni constant.

Since the top and bottom chains cannot coalesce until every edge has been updated at least once, we clearly have
\begin{equation}
\coupling\ge \coupon.
\label{coupon lower bounds coupling}
\end{equation}
Moreover, if $t\in\posint$, then by monotonicity, $\bottom_t$ and $\top_t$ will disagree on the edge $\sE_t$ iff $\sE_t\in\top_t$ and
$\sE_t\not\in\bottom_t$. In turn, this will occur iff: $\sE_t$ is pivotal to $\bottom_{t-1}$ but not pivotal to $\top_{t-1}$; and $\tilde{p}
< U_t \le p$. If $G$ is a tree, every edge is pivotal to every $A\subseteq E$, and the first condition cannot occur. If $q=1$, then
$\tilde{p}=p$ and the second condition cannot occur. It follows that if $q=1$, or if $G$ is a tree, then $\coupling=\coupon$ identically.

Our main interest in this article is the case that $G$ is $\ZZ_L^d$ for some choice of $L$ and $d$. In this case, $\coupling$ is certainly
not identically equal to $\coupon$. For $d=1$ however, Theorem~\ref{thm: d=1} shows that, for large $L$, the behaviour of $\coupling$
closely mimics that of $\coupon$. To emphasize the dependence of $\coupling$ and $\coupon$ on $L$ we append subscripts in the remainder of
this section.
\begin{theorem}
\label{thm: d=1}
Consider the FK heat-bath coupling on $\ZZ_L$ with parameters $p\in(0,1)$ and $q\geq 1$. 
Then, as $L\to\infty$, we have:
\begin{enumerate}[label=(\roman*)]
\item\label{thm:d=1 coupling time mean} $\EE(\coupling_L) \sim \EE(\coupon_L)$.
\item\label{thm:d=1 coupling time variance}   $\var(\coupling_L) \sim \var(\coupon_L)$
\item\label{thm:d=1 coupling time distribution} $\PP[\coupling_L \leq \EE(\coupling_L) + x \sqrt{\var(\coupling_L)}]\to G(x)$ for each $x\in\RR$.
\item\label{thm:d=1 relaxation time} $\rel \asymp L$.
\end{enumerate}
\end{theorem}

Intuitively, one expects the behaviour of the model on $\ZZ_L$ to be representative of the sub-critical behaviour on $\ZZ_L^d$ for any
$d\ge1$.  This suggests that the sub-critical behaviour on $\ZZ_L^d$ should again be governed by the coupon collector time.
Conjectures~\ref{conj:off-critical mean} and~\ref{conj:off-critical variance} formalize this intuition in the case of the mean and variance.
These conjectures are consistent with the rigorous bounds known in two dimensions, discussed in Remark~\ref{rem:d=2 bounds}.

To ease notation in what follows, we define $\mu_\coupling(L):=\EE(\coupling_L)$ and $\sigma_\coupling(L):=\sqrt{\var(\coupling_L)}$, and
likewise set $\mu_\coupon(L):=\EE(\coupon_L)$ and $\sigma_\coupon(L):=\sqrt{\var(\coupon_L)}$.  For brevity, we omit explicit mention of the
dependence of $\mu_\coupling$ and $\sigma_\coupling$ on $p,q$. In later sections, we shall also often omit explicit mention of their $L$ dependence.
\begin{conjecture}[Off-critical mean]
\label{conj:off-critical mean}
Consider the FK heat-bath coupling on $\ZZ_L^d$ with $d\geq 2$, $q \geq 1$ and $p\in(0,1)$ such that $p\neq\pc$. 
There exists $C(p,q,d)\ge1$ such that as $L\to\infty$
$$
\mu_\coupling(L) \sim C(p,q,d)\,\mu_\coupon(L).
$$
\end{conjecture}
\noindent Numerical evidence in support of Conjecture~\ref{conj:off-critical mean} is presented in Section~\ref{subsec:off-critical moments}.

We note that, if correct, Conjecture~\ref{conj:off-critical mean} combined with~\eqref{coupon collector mean} and the recent mixing time
bound~\cite{BlancaSinclair16} implies that for $d=2$ we have $\mu_\coupling(L)\asymp \mix(L)$ whenever $p\neq\pc$.  It seems natural to expect
that this in fact holds in all dimensions. Given the difficulty of estimating $\mix$ numerically, however, we have no empirical evidence
to directly support the claim $\mu_\coupling(L)\asymp \mix(L)$, and we therefore do not state it formally as a conjecture.

\begin{conjecture}[Off-critical variance]
\label{conj:off-critical variance}
Consider the FK heat-bath coupling on $\ZZ_L^d$ with $d\geq 2$, $q \geq 1$ and $p\in(0,1)$ such that $p\neq\pc$. 
There exists $C(p,q,d)\ge1$ such that as $L\to\infty$
$$
\sigma_\coupling(L) \sim C(p,q,d)\,\sigma_\coupon(L).
$$
\end{conjecture}
\noindent Numerical evidence in support of Conjecture~\ref{conj:off-critical variance} is presented in Section~\ref{subsec:off-critical moments}.

One consequence of Theorem~\ref{thm: d=1} is that $\sigma_\coupling(L)\asymp\rel(L)$ when $d=1$. While no precise asymptotics appear to be
known for $\rel$ when $d>1$, from a physical standpoint one expects that $\rel(L)\asymp L^d$ for $p\neq\pc$, in any dimension $d$. Under this
additional hypothesis, Conjecture~\ref{conj:off-critical variance} is equivalent to the conjecture that $\sigma_\coupling(L)\asymp \rel(L)$.
We shall return to this observation shortly.

Combining Conjectures~\ref{conj:off-critical mean} and~\ref{conj:off-critical variance} with~\eqref{coupon collector mean} and~\eqref{coupon
  collector variance} implies $\sigma_\coupling(L)/\mu_\coupling(L)$ goes to zero as $L\to\infty$. It then follows from Chebyshev's
inequality that for any $\epsilon>0$
$$
\PP[(1-\epsilon)\mu_\coupling(L) < \coupling_L < (1+\epsilon)\mu_\coupling(L)] \ge 1 - o(1), \qquad L\to\infty.
$$
While the moments of $\coupling_L$ do not behave like the corresponding moments of $\coupon_L$ at $\pc$, our numerical results
do suggest that $\mu_\coupling(L)$ remains the dominant time scale at criticality when $q<\qc$.
\begin{conjecture}
\label{conj:mean dominates standard deviation}
Consider the FK heat-bath coupling on $\ZZ_L^d$ with $d\geq 2$, $q \geq 1$ and $p\in(0,1)$ such that if $q\ge\qc$ then $p\neq\pc$. 
Then $\sigma_\coupling(L)/\mu_\coupling(L)\to0$ as $L\to\infty$.
\end{conjecture}
\noindent Numerical evidence in support of Conjecture~\ref{conj:mean dominates standard deviation} is presented in Section~\ref{subsec:critical moments}. 
Our numerical results suggest that Conjecture~\ref{conj:mean dominates standard deviation} does not hold at $p=\pc$ when $q\ge\qc$.

In light of Conjectures~\ref{conj:off-critical mean} and~\ref{conj:off-critical variance}, one is tempted to conjecture further that
Part~\ref{thm:d=1 coupling time distribution} of Theorem~\ref{thm: d=1}, the Gumbel limit law, also extends to the case $d>1$ in the
off-critical regime. Section~\ref{subsec:off-critical limiting distribution} provides strong numerical evidence to support this claim. What
is perhaps more surprising, however, is that the numerical results of Section~\ref{subsec:critical limiting distribution} strongly suggest
that the Gumbel limit law holds even \emph{at} the critical point, provided $q<\qc$. This is despite the fact that $\mu_\coupling(L)$ and
$\sigma_\coupling(L)$ certainly \emph{do not} behave like the analogous moments of $\coupon_L$ at $p=\pc$. In this sense, it seems
$\coupling_L$ displays a ``superuniversal'' central limit theorem, independent of $q$, for all $q<\qc$. Conjecture~\ref{conj:limiting
  distribution} formalizes this claim.
\begin{conjecture}[Limiting Distribution]
\label{conj:limiting distribution}
Consider the FK heat-bath coupling on $\ZZ_L^d$ with $d\geq 2$, $q \geq 1$ and $p\in(0,1)$ such that if $q\ge\qc$ then $p\neq\pc$. 
Then
$$
\lim_{L\rightarrow \infty} \PP[\coupling_L \leq \mu_\coupling(L) + x \sigma_\coupling(L)]= G(x), \qquad \text{ for each } x\in \RR.
$$
\end{conjecture}
\noindent Numerical evidence in support of Conjecture~\ref{conj:limiting distribution} is presented in Section~\ref{sec:fk limiting distribution}.
Our numerical results suggest the Gumbel limit law does not hold at $p=\pc$ when $q > \qc$. The special case $(p,q)=(\pc,\qc)$
appears to be rather subtle, and we are hesitant to make any predictions concerning it.

If we assume that Conjecture~\ref{conj:limiting distribution} is correct, then combining it with Theorem~\ref{thm:FK bounds on arbitrary
  graphs} suggests that $\sigma_\coupling(L)$ is asymptotic to $\expauto(L)$ as $L\to\infty$. Indeed, setting $t=\mu_L + x\,\sigma_L$
in~\eqref{tail distribution is exp with scale expauto} implies
$$
-\ln(2) - \frac{\mu_\coupling}{\expauto} - \frac{\sigma_\coupling}{\expauto}\,x \le 
\ln \PP(\coupling_L > \mu_\coupling + x \, \sigma_\coupling)
\le
d(\ln(\psi)+2)\,L^d - \frac{\mu_\coupling}{\expauto} - \frac{\sigma_\coupling}{\expauto}\,x.
$$
However, it is easily obtained from~\eqref{gumbel distribution function} that
$$
\ln[1-G(x)] \sim -\gamma -\frac{\pi}{\sqrt{6}}\,x, \qquad x\to\infty.
$$
Combining these facts with Conjecture~\ref{conj:limiting distribution} then motivates the following conjecture.
\begin{conjecture}[Variance]
\label{conj:variance vs relaxation time}
Consider the FK heat-bath coupling on $\ZZ_L^d$ with $d\geq 2$, $q \geq 1$ and $p\in(0,1)$ such that if $q\ge\qc$ then $p\neq\pc$. 
Then
$$
\sigma_\coupling(L) \sim \frac{\pi}{\sqrt{6}} \expauto(L), \qquad L\to\infty.
$$
\end{conjecture}
\noindent Numerical evidence in support of Conjecture~\ref{conj:variance vs relaxation time} is presented in Section~\ref{subsec:fk relaxation time}.

\begin{remark}
\label{rem:cutoff}
Recall that a sequence of chains has a \emph{cutoff}~\cite{LevinPeresWilmer09} if, for all $\epsilon>0$, we have
$\mix(L,\epsilon)/\mix(L,1-\epsilon)\to1$ as $L\to\infty$.  A necessary condition~\cite[Proposition 18.4]{LevinPeresWilmer09} for cut-off is
that $\mix(L)/\rel(L)\to\infty$ as $L\to\infty$.  If one assumes the validity of Conjectures~\ref{conj:mean dominates standard deviation}
and~\ref{conj:variance vs relaxation time}, and also assumes that $\mix\asymp\mu_\coupling(L)$, then this necessary condition will be
satisfied for the FK heat-bath process on $\ZZ_L^d$ with $d\geq 1$, $q \geq 1$ and any $p\in(0,1)$ such that if $q\ge\qc$ then
$p\neq\pc$. It is therefore tempting to speculate that the FK heat-bath process exhibits cutoff for all such parameter choices.
\end{remark}

For comparison, in Section~\ref{sec:Ising} we consider analogous questions for the single-spin Ising heat-bath process.  Above the critical
temperature, our results suggest the behaviour is identical to that conjectured above for the FK heat-bath process. Specifically, the mean
and variance of the coupling time are asymptotic to a constant $C\ge1$ multiple of their coupon collector analogues, the standard deviation
is asymptotic to $(\pi/\sqrt{6})\,\expauto$, and the standardized quantity $(\coupling-\mu_\coupling)/\sigma_\coupling$ has limiting
distribution $G(x)$.  At the critical temperature, however, the behaviour is somewhat different. In that case, our evidence suggests
$\sigma_\coupling/\mu_\coupling$ tends to a positive constant, rather than zero.  Moreover, while we do still observe that
$(\coupling-\mu_\coupling)/\sigma_\coupling$ has a non-degenerate limiting distribution, this distribution does not appear to be $G(x)$. We
have not attempted to identify the form of the limiting distribution in this case. Finally, we again find strong evidence that
$\sigma_\coupling \sim C \,\expauto$, but now with $C\neq\pi/\sqrt{6}$.  We state our conjectured behaviour for the Ising heat-bath process
more formally in Conjecture~\ref{conj:ising heat-bath}, in Section~\ref{subsec:ising heat-bath results}.

The observation that $\sigma_\coupling/\mu_\coupling$ tends to zero for the critical FK heat-bath process, but not the critical Ising
heat-bath process, provides another perspective on the improved efficiency of the former compared with the latter, over and above the
empirical observation of critical speeding-up and smaller relaxation time~\cite{DengGaroniSokal07_sweeny}. Moreover, if one postulates
(admittedly, in the absence of any significant evidence) that $\mu_\coupling\asymp \mix$, and assumes the validity of
Conjecture~\ref{conj:variance vs relaxation time} and its analogue for the Ising heat-bath process, then one concludes that $\mix/\rel$
diverges for the critical FK heat-bath process, but not for the critical Ising heat-bath process. As noted in Remark~\ref{rem:cutoff}, this would
immediately rule out cutoff in the Ising heat-bath process, but still allow for its existence in the FK heat-bath process.

\section{Moments}
\label{sec:fk moments}
We now present numerical evidence in support of Conjectures~\ref{conj:off-critical mean}, \ref{conj:off-critical variance}, \ref{conj:mean
  dominates standard deviation} and~\ref{conj:variance vs relaxation time}.  As discussed in Section~\ref{subsec:definitions}, for $d=2$ the
exact value of $\pc(q)$ is known, and it is known that $\qc=4$. Neither $\pc(q)$ nor $\qc$ are known when $d=3$. However, numerical
studies~\cite{DengGaroniSokalZhou,Hartmann05,Gliozzi02} of the case $q=2.2$ have provided convincing evidence that the transition at $q=2.2$
is continuous, suggesting $\qc>2.2$. In our simulations for $d=3$ we relied on the following estimated critical points:
$\pc(1.5)=0.31157497$, $\pc(1.8)=0.34096070$, $\pc(2)=0.35809124$ and $\pc(2.2)=0.37361401$. The values for $q=1.5, 1.8, 2.2$ are taken
from~\cite{DengGaroniSokalZhou}, while the value for $q=2$ is taken from~\cite{DengBlote03}.

\subsection{Off criticality}
\label{subsec:off-critical moments}
We begin by considering the off-critical mean. In order to test Conjecture~\ref{conj:off-critical mean},
Fig.~\ref{fig:mean_ratio_off_critical} plots Monte Carlo estimates of $\mu_\coupling$, scaled by the exact form of $\mu_\coupon$
from~\eqref{coupon collector mean}, on a linear-log scale, for $d=2,3$, with a variety of $q$ values, and off-critical $p$ values.  The
agreement is excellent.  The data are clearly converging to a constant $C(p,q,d)\ge1$.  The solid black line in
Fig.~\ref{fig:mean_ratio_off_critical} corresponds to the case $q=1$, for which $C(p,q,d)=1$ identically.  It
is conceivable, from the data at hand, that $C(p,q,d) = 1$ for all off-critical parameter choices $(p,q,d)$, however the current evidence
does not seem strong enough for us to actually conjecture that this is the case.
\begin{figure}
\centering
\subfigure[\label{fig:mean_ratio_off_critical}]{\includegraphics[width=0.49\columnwidth]{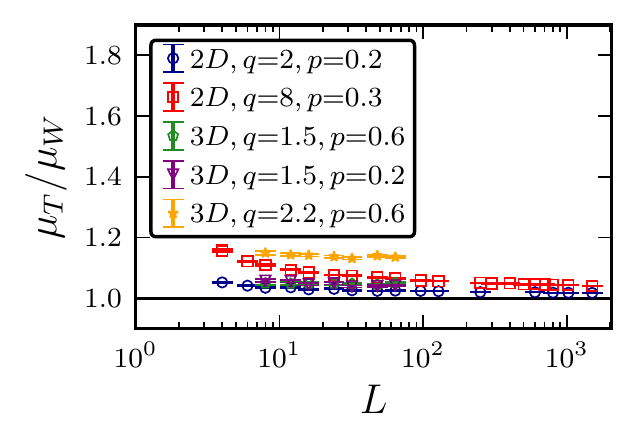}}
\subfigure[\label{fig:std_ratio_off_critical}]{\includegraphics[width=0.49\columnwidth]{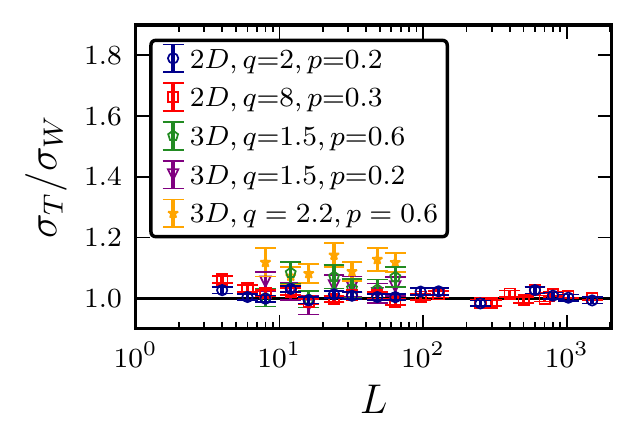}}
\caption{
Monte Carlo estimates of $\mu_\coupling/\mu_\coupon$ (left) and $\sigma_\coupling/\sigma_\coupon$ (right) for the off-critical
random-cluster model with $d=2,3$, with various cluster fugacities $q$ and bond densities $p$. Error bars corresponding to one standard error are shown.
}
\end{figure}
 
Analogously, in order to test Conjecture~\ref{conj:off-critical variance}, Fig.~\ref{fig:std_ratio_off_critical} plots Monte Carlo
estimates of $\sigma_\coupling/\sigma_\coupon$ for $d=2,3$, with a variety of $q$ values, and off-critical $p$ values, with $\sigma_\coupon$
given by~\eqref{coupon collector variance}.  The agreement is again excellent. The solid black line in
Fig.~\ref{fig:std_ratio_off_critical} again corresponds to the case $q=1$, for which $C(p,q,d)=1$
identically. It is again conceivable, based on Fig.~\ref{fig:std_ratio_off_critical}, that $C(p,q,d) = 1$ for all off-critical
parameter choices $(p,q,d)$. 

\subsection{Criticality}
\label{subsec:critical moments}
In this section, we consider $\mu_\coupling$ and $\sigma_\coupling$ at criticality when $q\le\qc$.  We begin by providing numerical evidence
in support of Conjecture~\ref{conj:mean dominates standard deviation}.  Recall that for $q=1$, we have from~\eqref{coupon collector mean}
and~\eqref{coupon collector variance} that $\mu_\coupling/\sigma_\coupling\sim C \ln(L)$ as $L\to\infty$, with $C>0$. It is therefore
natural to ask whether the ratio $\mu_\coupling/\sigma_\coupling$ continues to behave as a simple function of $\ln(L)$ when $q>1$. We
therefore present in Fig.~\ref{fig:ratios_mean_std} a log-log plot of the ratio $\mu_\coupling/\sigma_\coupling$ vs $\ln(L)$, for various
critical random-cluster model instances in two and three dimensions.  Except for $q=\qc=4$ in two dimensions, we observe that
$\mu_\coupling/\sigma_\coupling$ appears to become asymptotic to a straight line with positive slope, on a log-log scale. It appears that
the ratio approaches either a constant or weakly increases with $L$ at $d=2$ and $q=\qc=4$.  Similarly, in three dimensions, we observe that
$\mu_\coupling/\sigma_\coupling$ appears to increase more slowly in $L$ as $q$ approaches $\qc$.  These observations are consistent with the
following possible scenario: $\mu_\coupling/\sigma_\coupling \sim \ln(L)^w$ as $L\to\infty$ with an exponent $w$ that equals 1 at $q=1$ and
which decreases monotonically with $q$ before finally vanishing at $q=\qc$.  Regardless, we conclude that $\mu_\coupling/\sigma_\coupling$
diverges with $L$ at criticality when $q<\qc$, which supports Conjecture~\ref{conj:mean dominates standard deviation}.
\begin{figure}
  \includegraphics[width=0.49\columnwidth]{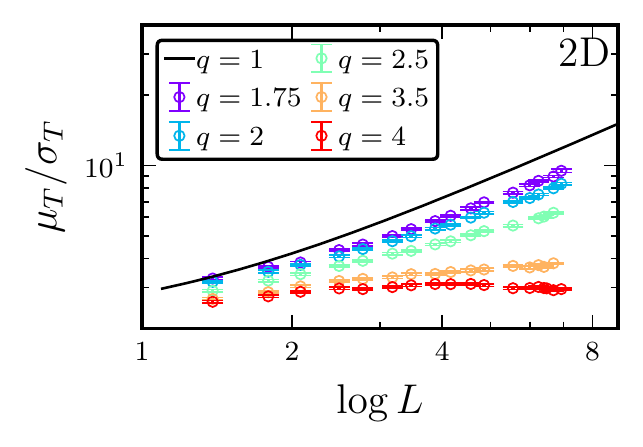}
  \includegraphics[width=0.49\columnwidth]{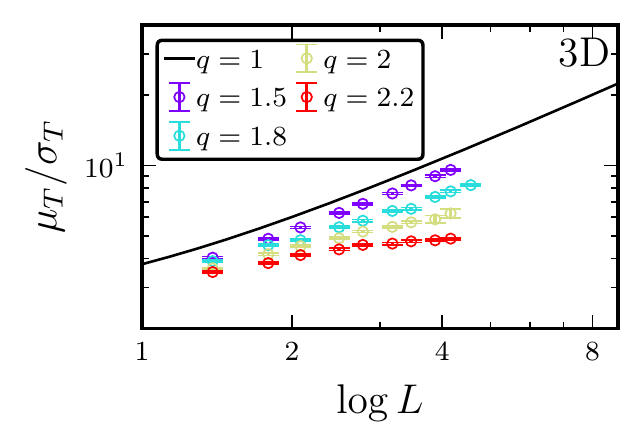}
\centering
\caption{\label{fig:ratios_mean_std}
Monte Carlo estimates of $\mu_\coupling/\sigma_\coupling$ for the critical random-cluster model with $d=2$ (left) and $d=3$ (right), 
with various cluster fugacities $q\le\qc$. 
Error bars corresponding to one standard error are shown.
}
\end{figure}

We next consider $\sigma_\coupling/\sigma_\coupon$. Fig.~\ref{fig:variance_ratio_critical} plots $\sigma_\coupling/\sigma_\coupon$ for
$d=2,3$ with various values of $q\le\qc$. It is clear that $\sigma_\coupling/\sigma_\coupon\to\infty$, which strongly suggests that
Conjecture~\ref{conj:off-critical variance} cannot be extended to $p=\pc$.  As we discuss in more detail in Section~\ref{subsec:dynamic
  critical exponents}, the ratio $\sigma_\coupling/\sigma_\coupon$ appears to grow at least as fast as $\ln(L)$.  Combining this
observation, together with~\eqref{coupon collector mean} and~\eqref{coupon collector variance}, with the above observation that
$\mu_\coupling/\sigma_\coupling$ diverges, implies that $\mu_\coupling/\mu_\coupon$ also diverges, which also rules out the possibility that
Conjecture~\ref{conj:off-critical mean} extends to $p=\pc$. Direct numerical data for the ratio $\mu_\coupling/\mu_\coupon$ support this conclusion.
 \begin{figure}
 \centering
\includegraphics[width=0.49\columnwidth]{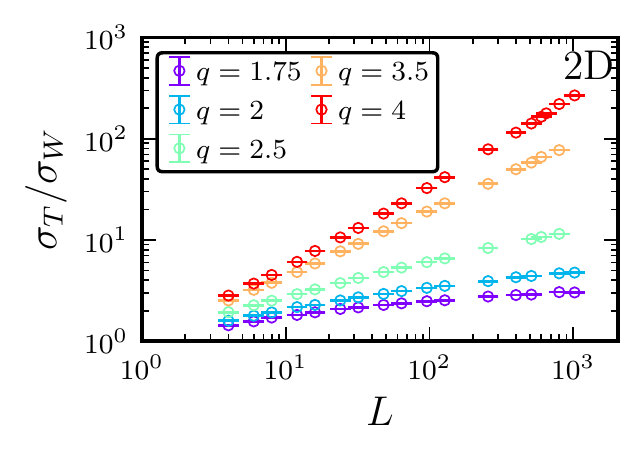}
\includegraphics[width=0.49\columnwidth]{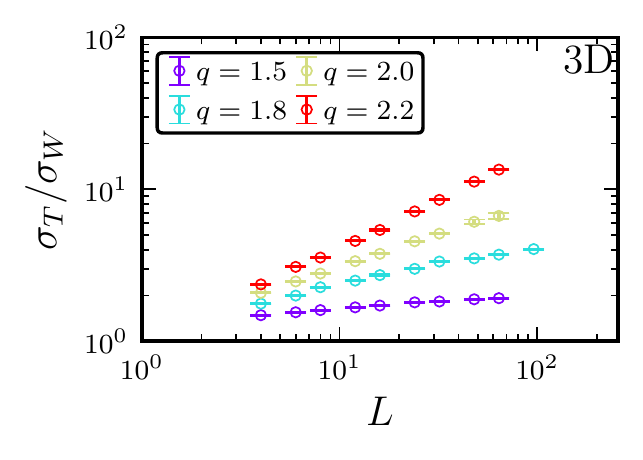}
  \caption{\label{fig:variance_ratio_critical}
 Monte Carlo estimates of $\sigma_\coupling/\sigma_\coupon$ for the critical random-cluster model with $d=2$ (left) and $d=3$ (right), with
 various cluster fugacities $q\le\qc$.
Error bars corresponding to one standard error are shown.
 }
 \end{figure}

\subsection{Variance and relaxation time}
\label{subsec:fk relaxation time}
We now provide evidence in support of Conjecture~\ref{conj:variance vs relaxation time}, in both the critical and off-critical cases.  Let
$(X_t)_{t\in\naturals}$ be a stationary FK heat-bath process, and define $(\sN_t)_{t\in\naturals}$ via $\sN_t=\sN(X_t)$, where $\sN(A)=|A|$
is the number of occupied edges. Since $\sN$ is a strictly increasing function, Proposition~\ref{prop:decay of rho_N} in
Appendix~\ref{appendix:autocorrelation functions of increasing observables} implies that
\begin{equation}
\rho_{\sN}(t) := \frac{\cov(\sN_0,\sN_t)}{\var(\sN_0)} \sim C e^{-t/\expauto}, \qquad t\to\infty,
\label{eq:asymptotics of rho_N}
\end{equation}
for some (parameter-dependent) constant $C>0$. Assuming the validity of Conjecture~\ref{conj:variance vs relaxation time}, it follows
from~\eqref{eq:asymptotics of rho_N} that
\begin{equation}
\ln\rho_{\sN}(k\,\sigma_\coupling) \sim -\frac{\pi}{\sqrt{6}} k 
\label{scaled rho ansatz}
\end{equation}
as $k$ and $L$ tend to infinity. 

For a given time lag $t$, we estimated $\rho_{\sN}(t)$ by performing around 100 independent simulations, estimating $\rho_{\sN}(t)$ from
each simulation using the standard time series estimator (see e.g.~\cite[Equation (3.9)]{Sokal97}), and then calculating the sample mean over
independent runs to obtain our final estimate of $\rho_{\sN}(t)$.  Fig.~\ref{fig:rho_off} plots the resulting estimates of $\rho_{\sN}(k\,
\sigma_\coupling)$ versus $k$, for a variety of values of $q$, $p$ and $L$.  The data collapse evident from the figure clearly supports the
expectation~\eqref{scaled rho ansatz}, and therefore provides direct evidence to support Conjecture~\ref{conj:variance vs relaxation time}.
 \begin{figure}
 \centering
 \includegraphics[width=0.75\columnwidth]{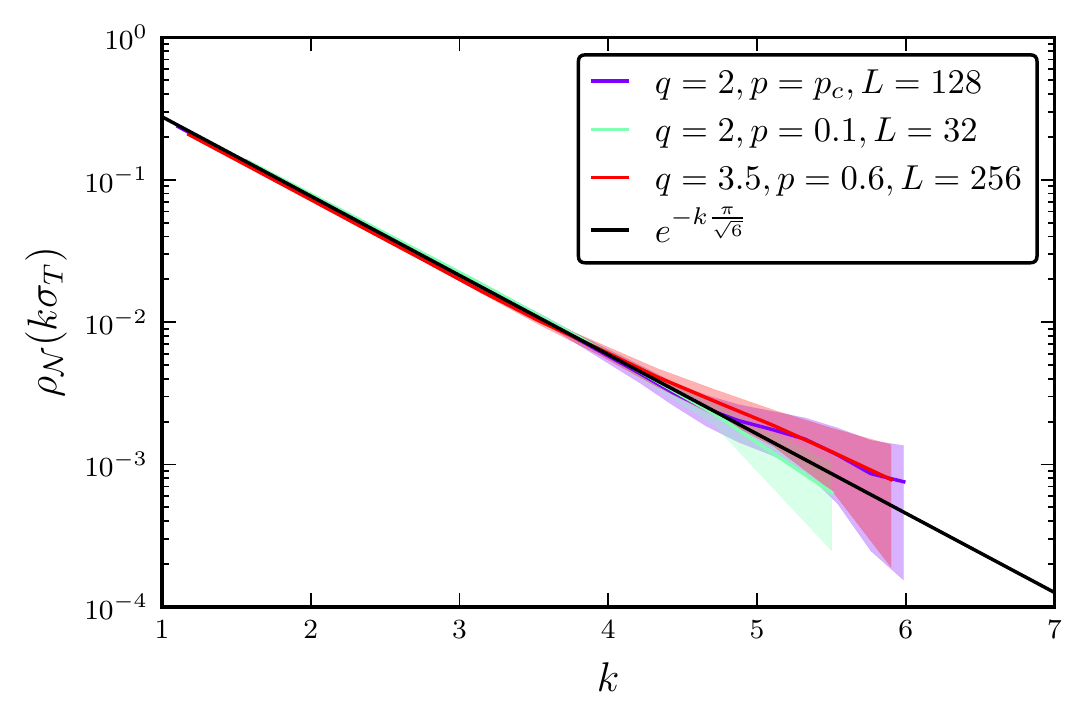}
 \caption{\label{fig:rho_off}
 Monte Carlo estimates of $\ln \rho_{\sN}(\sigma_\coupling\,k)$ for the random-cluster model with $d=2$, and various values of
 $q<\qc$, $p$ and $L$. The pairs $(q,p)=(2,0.1)$ and $(q,p)=(3.5,0.6)$ are off-critical.
The filled area enclosing the curves correspond to one standard error.
 }
 \end{figure}

To further test Conjecture~\ref{conj:variance vs relaxation time}, we used~\eqref{eq:asymptotics of rho_N} to directly estimate $\expauto$,
and then compared these estimates with our estimates of $\sigma_\coupling$. To estimate $\expauto$ from an estimate of $\rho_{\sN}(t)$, we
fitted a linear function $a - t/b$ to the data for $(t,\ln\rho_{\sN}(t))$, with appropriate cutoffs imposed at both small $t$ (to avoid the
pre-asymptotic regime) and large $t$ (to reduce statistical noise). Using these estimates, Fig.~\ref{fig:sigma over tau} shows the $L$ dependence of the ratio
$\sigma_\coupling/\expauto$ for a variety of critical and off-critical $(p,q)$ pairs, with $q<\qc$. The solid black line corresponds to the
asymptote $\pi/\sqrt{6}$ predicted by Conjecture~\ref{conj:variance vs relaxation time}. The data collapse is clearly excellent, lending further
strong support to the conjecture.
 \begin{figure}
 \centering
 \includegraphics[width=0.75\columnwidth]{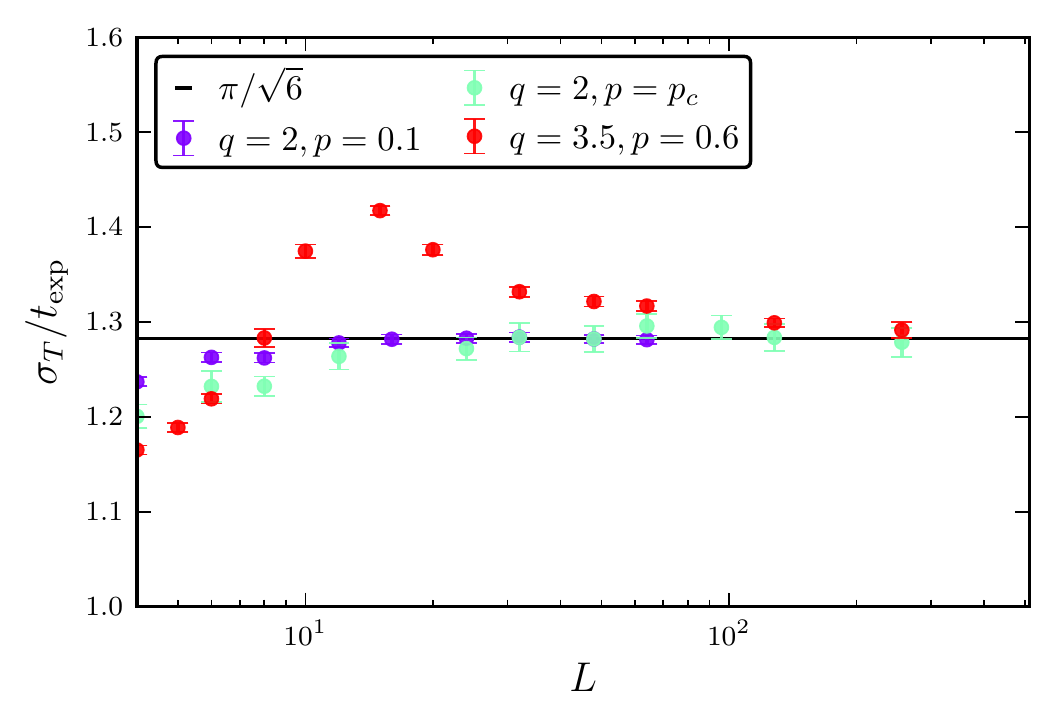}
 \caption{\label{fig:sigma over tau}
 Monte Carlo estimates of $\sigma_\coupling/\expauto$ for the random-cluster model with $d=2$, and various values of $q<\qc$, $p$ and
 $L$. The pairs $(q,p)=(2,0.1)$ and $(q,p)=(3.5,0.6)$ are off-critical.  The solid black line corresponds to the horizontal line
 $\pi/\sqrt{6}$.
Error bars corresponding to one standard error are shown.
 }
 \end{figure}

\subsection{Dynamic critical exponents}
\label{subsec:dynamic critical exponents}
We now briefly discuss a practical application of Conjecture~\ref{conj:variance vs relaxation time}.  Assuming the validity of
Conjectures~\ref{conj:off-critical variance} and~\ref{conj:variance vs relaxation time}, and combining them with~\eqref{coupon collector
  variance}, confirms the intuition mentioned in Section~\ref{subsec:results} that $\expauto\sim L^d$ off criticality. Moreover, a closer
inspection of the data in Fig.~\ref{fig:variance_ratio_critical} suggests that, at least for sufficiently large $q <\qc$, we have
$\sigma_\coupling/\sigma_\coupon\sim L^{z}$ for some exponent $z=z(q,d)>0$. Under the assumption of Conjecture~\ref{conj:variance vs
  relaxation time}, this is then equivalent to $\expauto\sim L^{d+z}$. This behaviour, which is precisely the phenomenon of
critical slowing-down, is expected on physical grounds~\cite{Sokal97} to occur generically at $p=\pc$ when $q < \qc$. The exponent $z$,
controlling the divergence of $\expauto/L^d$ at continuous phase transitions, is an example of a \emph{dynamic critical exponent}. It is
often denoted $z_{\exp}$ in the literature~\cite{Sokal97}.  While being of considerable physical and practical significance, the precise
estimation of $z_{\exp}$, even via simulation, is a highly non-trivial task. However, if Conjecture~\ref{conj:variance vs relaxation time}
holds, then $z_{\exp}$ for the FK heat-bath process can be estimated efficiently and reliably by considering the more tractable problem of
the asymptotics of $\sigma_\coupling$. For clarity, we denote the exponent governing the critical asymptotics of $\sigma_\coupling/L^d$ by
$z_\coupling$.

So motivated, we considered a number of $d$ and $q < \qc$ values, and fitted $\sigma_\coupling/L^d$ to both power-law and logarithmic
finite-size scaling ans\"atze, $aL^z + b$ and $a\ln(L) +b$, both with $b$ free and fixed to zero. For a given ansatz, the quality of the fit
was studied as we varied the lower cutoff on the $L$ values included in the fits.  Table \ref{tab:z estimates} summarises our best estimates
for $z_\coupling$, chosen to be the estimate resulting from the ansatz that yielded the highest confidence level, and stable estimates with
respect to a variation of the lower cutoff. For comparison, we also present corresponding values of $\alpha/\nu$, since a Li-Sokal type
bound~\cite{DengGaroniSokal07_sweeny} implies\footnote{Assuming the relevant exponents exist.} that $z_{\exp}\ge \alpha/\nu$. Here $\alpha$
and $\nu$ are the standard static critical exponents governing the specific heat and correlation length, respectively.  For $d=2$,
conjectured exact expressions for $\alpha$ and $\nu$ follow from the hyperscaling relation $d\nu=2-\alpha$, the identification of $1/\nu$
with the renormalization group thermal exponent, and~\cite[Equation (3.37)]{NienhuisJSP84}. For $d=3$, the reported values of $\alpha/\nu$
correspond to the estimates presented in~\cite{DengGaroniMachtaOssolaPolinSokal07}.
\begin{table}[htb]
\begin{center}
\begin{tabular}{cccc}
$d$ & $q$    & $\alpha/\nu$    & $z_{\coupling}$     \\
\hline \hline
$2$ & $1.75$ & $-0.1093$       & $0(\ln)$            \\
$2$ & $2$    & $0(\ln)$        & $0(\ln)$            \\
$2$ & $2.5$  & $0.2036$        & $0.315(3)$          \\
$2$ & $3$    & $0.4000$        & $0.491(4)$          \\
$2$ & $3.5$  & $0.6101$        & $0.662(2)$          \\
$3$ & $1.5$  & $-0.32(4)$      & $0.090(6)$          \\
$3$ & $1.8$  & $-0.15(5)$      & $0.233(4)$          \\
$3$ & $2$    & $0.174(1)$      & $0.435(9)$          \\
$3$ & $2.2$  & $0.50(4)$       & $0.646(5)$          \\
\end{tabular}
\end{center}
\caption{\label{tab:z estimates}
Estimated critical exponent $z_{\coupling}$ for a variety of values of $d$ and $q<\qc$. If Conjecture~\ref{conj:variance vs relaxation time}
holds, then $z_\coupling=z_{\exp}$. 
}
\end{table}

\section{Limiting Distribution}
\label{sec:fk limiting distribution}
We now turn our attention to the limiting distribution of the coupling time, and provide numerical evidence in support of
Conjecture~\ref{conj:limiting distribution}. To ease notation, in this section we introduce the standardized variable
\begin{equation}
\standard:=(\coupling-\mu_\coupling)/\sigma_\coupling.
\end{equation}

\subsection{Off criticality}
\label{subsec:off-critical limiting distribution}
In this section we present evidence supporting Conjecture~\ref{conj:limiting distribution} in the off-critical case. We defer discussion of the 
critical case until Section~\ref{subsec:critical limiting distribution}.

Fig.~\ref{fig:gumbel_off} compares histograms of $\standard$ with the probability density function corresponding to~\eqref{gumbel distribution function}.
The left panel corresponds to $d=2$ and $q=8$ at $p=0.3 < \pc(8,2)$.
The right panel corresponds to $d=3$ and $q=1.5$ at $p=0.2$;
for reference, it is estimated~\cite{DengGaroniSokalZhou} that $\pc(1.5,3)=0.31157497(59)$.
The agreement is clearly excellent, providing strong support for Conjecture~\ref{conj:limiting distribution}.

\begin{figure}
    \centering
\resizebox{0.49\columnwidth}{!}{\input{gumbel_off_q8_2d.pgf}}
\resizebox{0.49\columnwidth}{!}{\input{gumbel_off_q1p5_3d.pgf}}
  \caption{Histograms of off-critical $\standard$, with
  parameters as specified in the figure. The histograms are based on $19,000$ independent samples for $d=2$ and $11,000$
  for $d=3$. Here $p(s)$ denotes the probability density function of $\standard$. For comparison, the solid
  green line shows the probability density function corresponding to \eqref{gumbel distribution function}.\label{fig:gumbel_off}}
\end{figure}
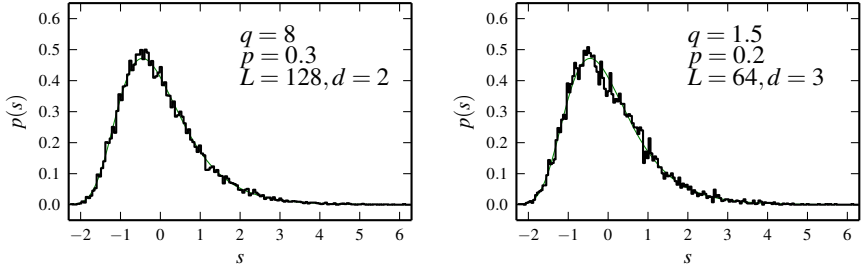

We emphasize that the theoretical curve shown in Fig.~\ref{fig:gumbel_off} does not correspond to a fit to the data; the distribution $G(x)$ 
does not possess any free parameters. In order to obtain a \emph{quantitative} measure of how well the limiting distribution of $\standard$ is described by
$G(x)$, we therefore considered the three-parameter family of distributions known as the Generalized Extreme Value distribution (GEV), defined by the
distribution function
\begin{equation}
    F_{\rm GEV}\big(x; \xi,\eta,\theta\big) := 
\begin{cases} 
      e^{-e^{-(x-\eta)/\theta}} & \xi = 0, \\
      e^{-\big(1+ \xi (x-\eta)/\theta  \big)^{-1/\xi}} &  \xi \neq 0,
   \end{cases}
\label{eq:GEV}
\end{equation}
where $\xi,\eta\in\RR$ and $\theta>0$.  The support of $F_{\mathrm{GEV}}$ is $\RR$ for $\xi=0$, $[\eta-\theta/\xi,\infty)$ for $\xi > 0$, and
$(-\infty,\eta-\theta/\xi]$ for $\xi < 0$.  The case $\xi=0$ corresponds to the Gumbel family of distributions, and the specific values
\begin{equation}
   \xi = 0, \quad \eta = -\frac{\gamma\sqrt{6}}{ \pi} = -0.45005320754\ldots, \quad \theta = \frac{\sqrt{6}}{\pi} = 0.77969780123\ldots
\label{G(x) parameters}
\end{equation}
correspond to $G(x)$ as given in~\eqref{gumbel distribution function}.

Our consideration of $F_{\mathrm{GEV}}$ can be motivated as follows. Consider an iid sequence of random variables $X_1,X_2,\ldots$ and let
$M_n:=\max\{X_1,\ldots,X_n\}$. The extremal types theorem (see e.g.~\cite[Theorem 1.4.2]{LeadbetterLindgrenRootzen83}) states that if the
sequence $M_n$, appropriately standardized\footnote{I.e. $M_n\mapsto (M_n-b_n)/a_n$ for some deterministic sequences $a_n>0$ and $b_n$.}, has a
non-degenerate limit, then the limit must be a GEV distribution. To relate this observation to the coupling time, we can envision coarse-graining the
lattice into regions of size much larger than the spatial correlation length, which is finite off criticality.  To each such region we can
assign a local coupling time, defined to be the last time before $\coupling$ that the state (occupied or unoccupied) of each edge in that region is
the same in the top and bottom chains. Since the correlations between regions are weak, as a first approximation one can envision the local
coupling times as independent. Moreover, the coupling time of the system, $\coupling$, is the maximum of these local coupling times. It is therefore
natural to expect that if an appropriate standardization of $\coupling$ converges to a non-degenerate limit as $L\to\infty$, then the limit should
be of the form~\eqref{eq:GEV}.

We therefore fitted the ansatz~\eqref{eq:GEV} to our data for $\standard$, and computed maximum likelihood-estimates of the parameters
$\xi,\eta,\theta$.  For $d=2$, $q=8$, $p=0.3$ and $L=128$, (left panel of Fig.~\ref{fig:gumbel_off}) we obtain
\begin{equation}
\xi = 0.01(1) \quad  \eta =  -0.45(1) \quad \theta =  0.77(1),
\end{equation}
based on $19000$ independent samples, and with error bars are computed via bootstrap re-sampling~\cite{Young15}. These estimates are in perfect agreement 
with the parameter values corresponding to $G(x)$.
Similarly, for $d=3$, $q=1.5$, $p=0.2$ and $L=64$ (right panel in Fig.~\ref{fig:gumbel_off}), we obtain
\begin{equation}
    \xi = 0.01(1) \quad  \eta =  -0.45(2) \quad \theta =  0.78(1)
\end{equation}
based on $11000$ samples. Finally, we also considered $d=3$, $q=2.2$ and $p=0.6$, which is expected to be in the supercritical
regime~\cite{DengGaroniSokalZhou}, and obtained
\begin{equation}
    \xi = 0.00(1) \quad  \eta = -0.45(2)  \quad \theta =  0.79(1),
\end{equation}
based on $11000$ samples. In each case, the estimates of the GEV parameters are entirely consistent with the parameter values in~\eqref{G(x)
  parameters} corresponding to $G(x)$, as predicted in Conjecture~\ref{conj:limiting distribution}.

\subsection{Criticality}
\label{subsec:critical limiting distribution}
Although we have observed that $\mu_\coupling$ and $\sigma_\coupling$ display non-trivial $L$ dependencies when $p=\pc$,
we now present evidence that Conjecture~\ref{conj:limiting distribution} is valid at $p=\pc$ when $q<\qc$.

Fig.~\ref{fig:gumbel_cri_con} compares histograms of $\standard$ with the probability density function corresponding to~\eqref{gumbel
  distribution function}.
The left panel corresponds to $d=2$ and $q=3.5$, while the right panel corresponds to $d=3$ and $q=1.5$.
The agreement is clearly excellent, providing strong support for Conjecture~\ref{conj:limiting distribution}.

\begin{figure}
    \centering

\resizebox{0.49\columnwidth}{!}{\input{gumbel_critical_q3p5_2d.pgf}}
\resizebox{0.49\columnwidth}{!}{\input{gumbel_critical_q1p5_3d.pgf}}
  \caption{Histograms of $\standard$ at criticality, with parameters as specified in the figure.  The histograms are based on $15,000$
    independent samples for $d=2$ and $10,000$ for $d=3$.  Here $p(s)$ denotes the probability density function of $\standard$.  For
    comparison, the solid green line shows the probability density function corresponding to \eqref{gumbel distribution function}.
  \label{fig:gumbel_cri_con}}
\end{figure}

Analogous to our discussion of the off-critical case in Section~\ref{subsec:off-critical limiting distribution}, we can obtain a more
quantitative test of the agreement between the limiting distribution of $S$ and~\eqref{gumbel distribution function} by fitting the GEV
distribution, \eqref{eq:GEV}.  We considered a number of values of $q<\qc$ with $d=2,3$, and our results are summarized in
Table~\ref{tab:cri_genextreme}. The estimates of the GEV parameters $\eta, \theta, \xi$ are in good agreement
with the parameter values~\eqref{G(x) parameters} corresponding to $G(x)$. The combination of these numerical results strongly support the validity of 
Conjecture~\ref{conj:limiting distribution}.
\begin{table}[htb]
\begin{center}
\begin{tabular}{ccccccc}
\hline
 $d$   &   $q$	   &   $L$        & $\eta$       &  $\theta$   & $\xi$       &  $N_{\mathrm{s}}$   \\
\hline \hline
2      &    $1.75$ &   $ 1024$    &  $ -0.45(3)$ &  $ 0.80(2)$ & $-0.02(2) $ &  $3360$  \\   
2      &    $1.75$ &   $ 512 $    &  $ -0.45(2)$ &  $ 0.79(1)$ & $-0.01(2) $ &  $9800$  \\ 
2      &    $2.00$ &   $ 512 $    &  $ -0.45(2)$ &  $ 0.78(1)$ & $0.00(1) $  &  $9000$  \\ 
2      &    $2.00$ &   $ 800 $    &  $ -0.45(2)$ &  $ 0.80(1)$ & $-0.02(1) $ &  $9000$  \\ 
2      &    $2.50$ &   $ 512 $    &  $ -0.45(2)$ &  $ 0.79(2)$ & $-0.01(1) $ &  $9000$  \\ 
2      &    $3.00$ &   $ 512 $    &  $ -0.45(2)$ &  $ 0.78(1)$ & $0.01(1) $  &  $9000$  \\ 
2      &    $3.50$ &   $ 512 $    &  $ -0.46(2)$ &  $ 0.77(1)$ & $0.02(2)$   &  $8990$  \\ 
2      &    $3.50$ &   $ 800 $    &  $ -0.46(1)$ &  $ 0.77(1)$ & $0.01(1)$   &  $15730$ \\
3      &    $1.5$  &   $ 64 $     &  $ -0.45(1)$ &  $ 0.78(1)$ & $0.01(2) $  &  $10000$ \\
3      &    $1.8$  &   $ 64 $     &  $ -0.44(1)$ &  $ 0.81(1)$ & $0.03(1) $  &  $10000$ \\
3      &    $2.0$  &   $ 64 $     &  $ -0.45(2)$ &  $ 0.80(3)$ & $0.03(3) $  &  $700$   \\
3      &    $2.2$  &   $ 64 $     &  $ -0.44(1)$ &  $ 0.80(1)$ & $0.03(1)$   &  $10000$ \\
\hline \hline 
\end{tabular}
\end{center}
\caption{Parameter estimates obtained by fitting the GEV distribution~\eqref{eq:GEV} to the empirical distribution for $S$
  at $p=\pc$, for various choices of $q<\qc$. Here $N_{\mathrm{s}}$ denotes the number of samples.
  Error bars correspond to one standard error.
\label{tab:cri_genextreme}
}
\end{table}

We conclude this section with some comments on the case of $p=\pc$ and $q\ge\qc$, which is excluded from our statement of
Conjecture~\ref{conj:limiting distribution}.  Because of the slow mixing inherent at discontinuous phase transitions, it is much more
difficult to obtain accurate simulation results at $p=\pc$ when $q>\qc$. We did however perform a simulation study for $d=2$ at $\pc$ for
$q=5>\qc$. While it does appear that the standardized variable $S$ again converges to a non-degenerate limit, it appears that this limit is
not $G(x)$.  To illustrate this, we generated 10,000 samples of $\coupling$ with $L=256$, and obtained the following GEV parameter
estimates: $\xi = 0.19(2)$, $\eta = -0.49(2)$, $\theta= 0.62(1)$. The deviation of $\xi$ away from the Gumbel value $\xi=0$ seems to provide
strong evidence that Conjecture~\ref{conj:limiting distribution} cannot be extended to the case of $p=\pc$ when $q>\qc$.

The case $q=\qc$ is more subtle. In this case, it is not slow mixing that constitutes an impediment, but rather the notorious issue of
multiplicative logarithms arising in finite-size scaling ans\"atze.  We simulated the case $d=2$ and $q=\qc=4$ at $p=\pc$, at a variety of
different $L$ values. We again observe that the distribution of $S$ appears to converge to a non-degenerate limit.  The GEV distribution was
fitted to the data for $S$, and the corresponding estimates for $\eta, \theta, \xi$ are reported in Table~\ref{tab:gumbel_cri_q4}. 
The resulting estimates of $\theta$ and $\xi$ are not consistent with the values corresponding to $G(x)$. In particular, the estimates
suggest $\xi$ is strictly positive, which would rule out a Gumbel limit law. Therefore, based on these estimates, there does not appear to
be any evidence suggesting Conjecture~\ref{conj:limiting distribution} can be extended to the case of $(p,q)=(\pc,\qc)$. However, the
discrepancies of these parameter estimates with the values corresponding to $G(x)$ are relatively small. Therefore, we also believe that
there is insufficient evidence to conclude that the distribution of $S$ at $(p,q)=(\pc,\qc)$ is actually different to $G(x)$.
Determining the limiting distribution of $S$ at $(p,q)=(\pc,\qc)$ therefore remains an open problem. 
\begin{table}[htb]
\begin{center}
\begin{tabular}{ccccc}
$L$  & $\eta$ & $\theta$ & $\xi$ & $N_{\mathrm{s}}$ \\
\hline \hline
$ 128$   &  $-0.46(2)$ &  $0.73(1)$ & $0.05(2) $ &  $8990$ \\
$ 256$   &  $-0.47(2)$ &  $0.71(1)$ & $0.08(2) $ &  $9000$ \\
$ 512$   &  $-0.47(1)$ &  $0.73(1)$ & $0.06(1)$ &  $18940$ \\
$ 800$   &  $-0.47(1)$ &  $0.72(1)$ & $0.07(1)$ & $20000$ \\
$ 1024$  &  $-0.47(2)$ &  $0.72(1)$ & $0.07(1)$ & $9850$ \\
\hline \hline 
\end{tabular}
\end{center}
\caption{Parameter estimates obtained by fitting the GEV distribution~\eqref{eq:GEV} to the empirical distribution for $S$
  at $(p,q)=(\pc,\qc)$ and $d=2$, for various choices of $L$. Here $N_{\mathrm{s}}$ denotes the number of samples.
  Error bars correspond to one standard error.
\label{tab:gumbel_cri_q4}
}
\end{table}

\section{Arbitrary graphs}
\label{sec:FK arbitrary graphs}
In this section, we consider the FK heat-bath process on arbitrary graphs, and prove Theorem~\ref{thm:FK bounds on arbitrary graphs}.
\begin{proof}[Proof of Theorem~\ref{thm:FK bounds on arbitrary graphs}]
Consider the FK heat-bath coupling on a finite connected graph $G=(V,E)$ with $|E|=m\ge1$, and let
\begin{equation}
d(t):=\max_{A\subseteq E} \|P^t(A,\cdot) - \fk\|_{\TV}.
\end{equation}
It follows from~\cite[Theorem 5]{ProppWilson96} and~\cite[Equation (4.24)]{LevinPeresWilmer09} that
\begin{equation}
d(t) \le \PP(\coupling >t) \le 2\,(m+1)\, d(t),
\label{Propp Wilson bound on d bar}
\end{equation}
for any $t\in\naturals$.
Combining the lower bound in~\eqref{Propp Wilson bound on d bar} with~\cite[Equation (12.13)]{LevinPeresWilmer09} yields the stated lower bound for the 
tail distribution:
$$
\PP(\coupling >t) \ge d(t) \ge \frac{e^{-t/\expauto}}{2}.
$$
Similarly, combining~\eqref{Propp Wilson bound on d bar} with Markov's inequality implies
$$
\EE(\coupling) \ge (\mix-1)\,\PP(\coupling > \mix-1) \ge (\mix-1) d(\mix-1) \ge \frac{(\mix-1)}{4}.
$$
This establishes the stated lower bounds.

We now consider the upper bounds. Let $M\in(1,\infty)$ be such that 
\begin{equation}
\PP(\coupling>l\,M)\le 2^{-l}, \qquad \text{for all } l\in\naturals.
\label{M bound}
\end{equation}
Then for $k\in\{0,1\}$ we have
\begin{align*}
\sum_{t=0}^\infty t^k\,\PP(\coupling>t) 
& \le \sum_{l=0}^\infty \sum_{t=l \lceil M\rceil}^{(l+1)\lceil M\rceil-1} t^k\, \PP(\coupling>l\,M)\\
& \le \sum_{l=0}^\infty \sum_{t=l \lceil M\rceil}^{(l+1)\lceil M\rceil-1} t^k\, 2^{-l}\\
& = (k+2) \lceil M\rceil^{k+1} - k\,\lceil M\rceil.
\end{align*}
The $k=0$ case then immediately yields an upper bound for $\EE(\coupling)$, via
\begin{equation}
\EE(\coupling) = \sum_{t=0}^\infty \PP(\coupling >t) \le 2\,\lceil M\rceil  \le 4\,M.
\end{equation}
Similarly, standard manipulations of probability generating functions (see e.g.~\cite[Chapter XI]{Feller68}) show that 
\begin{equation}
\EE\,[\coupling(\coupling-1)] = 2 \sum_{t=0}^\infty t \,\PP(\coupling>t),
\label{M bound on E(T)}
\end{equation}
and so the $k=1$ case yields
\begin{equation}
\var(\coupling) \le \EE[\coupling(\coupling -1)] + \EE(\coupling) \le 6\lceil M\rceil^2,
\end{equation}
which implies
\begin{equation}
\sqrt{\var(\coupling)}\le \sqrt{6}(M+1) \le 2\sqrt{6}\,M\le 5\,M.
\label{M bound on var(T)}
\end{equation}

We now determine suitable choices of $M$ for which~\eqref{M bound} holds. Since the bound is trivial for $l=0$, we assume $l\ge1$. 
We begin by considering bounds in terms of $\expauto$.
Letting $\fk_{\min}:=\min\{\fk(A) : A\subseteq E\}$, 
and combining~\cite[Lemma 6.13]{LevinPeresWilmer09} and~\cite[Equation (12.11)]{LevinPeresWilmer09} with~\eqref{Propp Wilson bound on d bar} yields
\begin{equation}
\PP(\coupling > t) \le 2(m+1)\,d(t) \le \frac{2(m+1)}{\fk_{\min}}e^{-t/\expauto} \le 2\,(m+1)\,\psi^m\,e^{-t/\expauto},
\label{expauto upper tail bound}
\end{equation}
where in the last step we applied Lemma~\ref{lem: mu min}. Since $\ln[2(m+1)]\le 2m$, this immediately yields the stated upper bound for $\PP(\coupling>t)$.
Likewise, since ${\blog[2(m+1)])\le2m}$, if we set $M=[\blog(\psi)+3]m\,\expauto$ then it follows from~\eqref{expauto upper tail bound} that
\begin{align*}
\blog \PP(\coupling>l\,M) 
&\le [\blog(\psi)+2]m - l[\blog(\psi)+3]m \\
&= -(l-1)[2+\blog(\psi)]m - m\,l\\
&\le -m\,l\\
&\le -l,
\end{align*}
and so~\eqref{M bound} holds. Inserting this choice of $M$ into~\eqref{M bound on E(T)} and~\eqref{M bound on var(T)} then yields the stated
upper bounds for $\EE(\coupling)$ and $\sqrt{\var(\coupling)}$ in terms of $\expauto$.

Finally, we consider bounds in terms of $\mix$.
Combining~\eqref{Propp Wilson bound on d bar} with~\cite[Equation (4.35)]{LevinPeresWilmer09} we obtain
$$
\PP(\coupling>t) \le 2(m+1)\,d(t) \le 4\,m\,d(t) \le 2^{\blog(4m) - \lfloor t/\mix\rfloor}
$$
Setting $M=3\blog(4m)\,\mix$, it follows that
\begin{align*}
\blog \PP(\coupling>l\,M) 
&\le \blog(4m) -3\blog(4m)l+1\\
&= [1-(2l-1)\blog(4m)] - \blog(4m)\,l\\
&\le -\blog(4m)\,l\\
&\le -l,
\end{align*}
which implies that~\eqref{M bound} holds. Inserting this choice of $M$ into~\eqref{M bound on E(T)} and~\eqref{M
  bound on var(T)} then yields the stated upper bounds for $\EE(\coupling)$ and $\sqrt{\var(\coupling)}$ in terms of $\mix$.

\end{proof}

\begin{lemma}
\label{lem: mu min}
Consider the FK model with $q\ge1$ and $p\in(0,1)$, on a finite connected graph $G=(V,E)$ with $m$ edges, and let
$\fk_{\min}:=\min\{\fk(A):A\subseteq E\}$. Then
$$
\fk_{\min} \ge \left(\frac{p(1-p)}{q^2}\right)^m. 
$$
\begin{proof}
Since
$$
\sum_{A\subseteq E} p^{|A|} (1-p)^{m-|A|} q^{k(A)} \le q^n \sum_{A\subseteq E} p^{|A|} (1-p)^{m-|A|} = q^n,
$$
for any $A\subseteq E$ we have
$$
\fk(A) \ge p^{|A|} (1-p)^{m-|A|} q^{k(A)-n} \ge p^m (1-p)^m q^{-n}.
$$
The stated result then follows since $G$ being connected implies $n\le m+1 \le 2m$.
\end{proof}
\end{lemma}

\section{The cycle}
\label{sec:FK one dimension}
In this section, we consider the FK heat-bath coupling, with parameters $p\in(0,1)$ and $q\ge1$, on the graph $\ZZ_L$, and prove
Theorem~\ref{thm: d=1}.  We begin, in Section~\ref{sec:FK one dimension coupong collector behaviour}, by showing that $\coupling$ equals
$\coupon$ with high probability, as $L\to\infty$.  This observation is then used in Section~\ref{sec:FK one dimension mean and variance} to
prove Parts~\ref{thm:d=1 coupling time mean} and~\ref{thm:d=1 coupling time variance}, of Theorem~\ref{thm: d=1}, and again in
Section~\ref{sec:FK one dimension distribution} to prove Part~\ref{thm:d=1 coupling time distribution}.  Finally, in Section~\ref{sec:FK one
  dimension relaxation time}, we prove Part~\ref{thm:d=1 relaxation time}.

\subsection{Asymptotic Coupon Collector Behaviour}
\label{sec:FK one dimension coupong collector behaviour}
For each $e\in E$, define
$$
\lastvisit(e) = \sup\{t\le \coupon: \sE_t=e\}.
$$
We refer to the time $\lastvisit(e)$ as the \emph{last visit} to $e$.  Let $(\lastvisit_i)_{i=1}^m$ denote the sequence of the
$\lastvisit(e)$, arranged in increasing order. In particular, $\lastvisit_1$ is the first time that a last visit occurs. And likewise,
$\lastvisit_i$ is the time that the $i$th last visit occurs.  

\begin{proposition}
\label{prop:coupon and coupling are asymptotic when d=1}
Consider the FK heat-bath coupling on $\ZZ_L$ with $p\in(0,1)$ and $q\ge1$. There exists $\epsilon>0$ such that $\PP(\coupling=\coupon)= 1 -
O(L^{-\epsilon})$.
\end{proposition}

\begin{proof}
Fix $p\in(0,1)$ and $q\ge1$, and let $a_L:=\floor{\ln L}$. Let $\sP_j$ be the event that the edge $\sE_{\lastvisit_j}$ is pivotal in the
top process at time $\lastvisit_{j}-1$. By monotonicity, whenever an edge is pivotal to the top chain, it is also pivotal to the bottom chain.
Lemma~\ref{lem:last visits are pivotal whp} implies that there exists $\omega>0$ such that
\begin{equation}
\PP(\coupling > \coupon) = \PP\left(\coupling > \coupon \bigg| \bigcap_{j=1}^{a_L}\sP_j \right) + O(L^{-\omega}).
\label{coupling exceeds coupon effect of conditioning on pivotality}
\end{equation}

Now suppose $\cap_{j=1}^{a_L}\sP_j$ occurs, and let $1\le i \le a_L$.  If $U_{\lastvisit_i}\le\tilde{p}$, then
$\sE_{\lastvisit_i}\in\bottom_t$ and $\sE_{\lastvisit_i}\in\top_t$ for all $t\in[\lastvisit_i,\coupon]$, while if $U_{\lastvisit_i}>\tilde{p}$
then $\sE_{\lastvisit_i}\not\in\bottom_t$ and $\sE_{\lastvisit_i}\not\in\top_t$ for all $t\in[\lastvisit_i,\coupon]$.  Consequently, on $\ZZ_L$, if
$U_{\lastvisit_i}>\tilde{p}$, then every edge $e\neq \sE_{\lastvisit_i}$ is pivotal to $\top_t$ and $\bottom_t$, for all
$t\in(\lastvisit_i,\coupon]$, so that after any update of such an edge in this time window, its state (occupied or unoccupied) in the top
and bottom chain will agree.  Since each $e\in E\setminus\{\sE_{\lastvisit_1},\ldots,\sE_{\lastvisit_i}\}$ must be updated in
$(\lastvisit_i,\coupon]$, this implies that the top and bottom chains agree at time $\coupon$, and so $\coupling\le \coupon$. It follows
that $\coupling > \coupon$ can occur only if $U_{\lastvisit_i}\le \tilde{p}$ for each $1\le i \le a_L$, and so
\begin{equation}
\PP\left(\coupling > \coupon \bigg| \bigcap_{j=1}^{a_L}\sP_j \right) \le \tilde{p}^{a_L} \le \frac{L^{\ln\tilde{p}}}{\tilde{p}}.
\label{coupling exceeds coupon conditioned on pivotality}
\end{equation}
Since $p\in(0,1)$, we have $\tilde{p}\in(0,1)$, and so $\ln(\tilde{p})<0$. Choosing $\epsilon=\min\{\omega,-\ln\tilde{p}\}$, and combining
\eqref{coupling exceeds coupon effect of conditioning on pivotality} and \eqref{coupling exceeds coupon conditioned on pivotality}, we
therefore obtain
$$
\PP(\coupling > \coupon) = O(L^{-\epsilon}).
$$
Combining this observation with~\eqref{coupon lower bounds coupling} yields the stated result.
\end{proof}

\begin{lemma}\label{lem:last visits are pivotal whp}
Consider the top process on $\ZZ_L$, with fixed $p\in(0,1)$ and $q\ge1$. Let $\sP_j$ be the event that the edge $\sE_{\lastvisit_j}$ is pivotal at
time $\lastvisit_j-1$, and let $a_L=\floor{\ln L}$.  Then there exists $\omega>0$ such that
$$
\PP\left(\bigcup_{j=1}^{a_L} \sP_j^c\right) = O(L^{-\omega}), \qquad L\to\infty.
$$
\begin{proof}
Let $R_t:=\{e\in E: \sE_s = e \textrm{ for some } s\leq t\}$, the set of distinct edges visited up to time $t$.
Let $\sD=\{|R_{\lastvisit_1}|>a_L\}$ be the event that more than $a_L$ distinct edges have been visited by time $\lastvisit_1$. If $\sD$
holds, then it also holds that more than $a_L$ distinct edges have been visited by time $\lastvisit_j$, for any $j\ge1$. Fix $1\le j \le L$
and $p,\xi\in(0,1)$, and define $\sA_j$ to be the event that at least $\xi(1-p)a_L$ edges are unoccupied at time $\lastvisit_j-1$. For any
choice of $\xi,p\in(0,1)$, we have $\xi(1-p)a_L\ge2$ for all sufficiently large $L$; let $L$ be so chosen in what follows. Then, if $\sA_j$
occurs, there are at least 2 unoccupied edges at time $\lastvisit_j-1$, which in turn means that all edges are pivotal at time
$\lastvisit_j-1$.  Therefore, $\sA_j\subseteq\sP_j$.

Let $\sW$ denote the set of the first $a_L$ distinct edges visited by $(\sE_t)_{t\in\posint}$.  On $\sD$, the edges in $\sW$ are all visited
prior to $\lastvisit_j$.  Denote the times of last visit to the edges in $\sW$, prior to $\lastvisit_j$, by $M_1 < M_2 < \ldots < M_{a_L}$,
and let $1\le i\le a_L$. Since $p\ge\tilde{p}$, if $U_{M_i}> p$, then regardless of the structure of $\top_{M_i-1}$, we have
$\sE_{M_i}\not\in\top_{M_i}$. It follows that if $X_i=\indicator(U_{M_i} >p)$, then
$\PP(\sE_{M_i}\not\in\top_{M_i})\ge\PP(X_i=1)=(1-p)$. Therefore, Chernoff's bound~\cite[Equation (4.5)]{MitzenmacherUpfal05} implies
$$
\PP\left(\sP_j^c | \sD \right) \le
\PP\left(\sA_j^c | \sD \right) \le
\PP\left(\sum_{i=1}^{a_L}X_i < \xi(1-p)a_L \right) \le e^{-\gamma a_L}
$$
with $\gamma=(1-\xi)^2(1-p)/2>0$.

Combining this with the union bound gives
\begin{align*}
\PP\left(\bigcup_{j=1}^{a_L}\sP_j^c \bigg|\sD \right) & \le \sum_{j=1}^{a_L}\PP(\sP_j^c|\sD)\\
&\le   a_L e^{-\gamma a_L}\\
&\le   e^{\gamma}\,\ln(L) \, L^{-\gamma}.
\end{align*}
It follows that for any $0 < \delta < \gamma$ we have
$$
\PP\left(\bigcup_{j=1}^{a_L}\sP_j^c \bigg|\sD \right) = O(L^{-\delta}).
$$
Finally, Lemma~\ref{lem:the first visited edges are revisited often} in Appendix~\ref{appendix:coupon collecting} implies that there
exists $\varphi>0$ such that
$$
\PP(\sD^c) = O(L^{-\varphi}).
$$
Choosing $\omega=\min\{\delta,\varphi\}$ then implies 
$$
\PP\left(\bigcup_{j=1}^{a_L}\sP_j^c\right) = O(L^{-\omega}).
$$
Since $\omega>0$, the stated result follows.
\end{proof}
\end{lemma}

\subsection{Mean and Variance}
\label{sec:FK one dimension mean and variance}
Define the sequence of random times $W_j$ such that $W_0=0$ and for $j\in\posint$
$$
W_j:= \min\{ t > W_{j-1} \,:\, \{\sE_{W_{j-1}+1},\ldots,\sE_t\} = E\}. 
$$
We define new processes $(\tilde{\sT}_t)_{t\in\naturals}$ and $(\tilde{\sB}_t)_{t\in\naturals}$ as follows, 
which proceed analogously to the top and bottom chains, except that they are restarted at times $W_j<\coupling$. 
More precisely, let $\tilde{\sT}_0=E$, and for $t\in\naturals$ set
\begin{equation}
\label{tilde top chain}
\tilde{\sT}_{t+1} =
\begin{cases}
f(E, \sE_{t+1},U_{t+1}), & t = W_j \text{ and } \coupling > W_j,\\
f(\tilde{\sT}_{t}, \sE_{t+1},U_{t+1}), & \text{otherwise.}
\end{cases}
\end{equation}
Similarly, $\tilde{\sB}_0=\emptyset$, and for $t\in\naturals$ we set 
\begin{equation}
\label{tilde bottom chain}
\tilde{\sB}_{t+1} = 
\begin{cases}
f(\emptyset, \sE_{t+1},U_{t+1}), & t = W_j \text{ and } \coupling > W_j,\\
f(\tilde{\sB}_{t}, \sE_{t+1},U_{t+1}), & \text{otherwise.}
\end{cases}
\end{equation}
By monotonicity, it is clear that for all $t\in\naturals$ we have
\begin{equation}
\tilde{\sB}_t \le \sB_t \le \sT_t \le \tilde{\sT}_t.
\label{bounding tilded and untilded bottom and top chains}
\end{equation}
We can now consider the coupling time corresponding to $\tilde{\sT}_t$ and $\tilde{\sB}_t$, 
\begin{equation}
\tilde{\coupling}:= \min\{t\in\naturals \,:\, \tilde{\sT}_t = \tilde{\sB}_t \}.
\label{tilde coupling time}
\end{equation}
It follows from~\eqref{bounding tilded and untilded bottom and top chains} that $\coupling\le\tilde{\coupling}$. 
Combining this with~\eqref{coupon lower bounds coupling} implies
\begin{equation}
\coupon \le \coupling \le \tilde{\coupling}. 
\label{eq:sandwiching coupling between coupon and tilde coupling}
\end{equation}

\begin{proof}[Proof of Theorem~\ref{thm: d=1}, Parts~\ref{thm:d=1 coupling time mean} and~\ref{thm:d=1 coupling time variance}]
Combining~\eqref{eq:sandwiching coupling between coupon and tilde coupling} and Lemma~\ref{lem:T tilde moments} immediately yields 
$\EE(\coupling)\sim\EE(\coupon)$, which establishes Part~\ref{thm:d=1 coupling time mean}.

To establish Part~\ref{thm:d=1 coupling time variance} we note that~\eqref{eq:sandwiching coupling between coupon and tilde coupling} implies
$$
\EE(\coupon^2) - (\EE\tilde{\coupling})^2 \le \var(\coupling) \le \EE(\tilde{\coupling}^2) - (\EE\coupon)^2,
$$
and rearranging, we obtain
\begin{equation}
1 - \frac{(\EE\coupon)^2}{\var(\coupon)}\,\left[\left(\frac{\EE\tilde{\coupling}}{\EE\coupon}\right)^2-1\right] \le 
\frac{\var(\coupling)}{\var(\coupon)} \le 
\frac{\var(\tilde{\coupling})}{\var(\coupon)} 
+ \frac{(\EE\coupon)^2}{\var(\coupon)}\,\left[\left(\frac{\EE\tilde{\coupling}}{\EE\coupon}\right)^2-1\right].
\label{variance of coupling on variance of coupon inequality}
\end{equation}
Lemma~\ref{lem:T tilde moments} implies 
\begin{equation}
\left[\left(\frac{\EE\tilde{\coupling}}{\EE\coupon}\right)^2-1\right] = O(L^{-\epsilon}),
\label{mean of tilde coupling on mean of coupon asymptotics}
\end{equation}
while \eqref{coupon collector mean} and~\eqref{coupon collector variance} imply
\begin{equation}
\frac{(\EE\coupon)^2}{\var(\coupon)} = O(\ln^2(L)),
\label{ratio of coupon squared mean and variance}
\end{equation}
and so~\eqref{variance of coupling on variance of coupon inequality} yields
$$
1 \le \liminf_{L\to\infty} \frac{\var(\coupling)}{\var(\coupon)} \le 
\limsup_{L\to\infty} \frac{\var(\coupling)}{\var(\coupon)} \le 
\limsup_{L\to\infty} \frac{\var(\tilde{\coupling})}{\var(\coupon)}.
$$
Combining this with Part~\ref{tilde variance} of Lemma~\ref{lem:T tilde moments} then implies that $\var(\coupling)\sim\var(\coupon)$.
\end{proof}

\begin{lemma}
\label{lem:T tilde moments}
Let $\tilde{\coupling}$ be as defined in~\eqref{tilde coupling time}. Then:
\begin{enumerate}[label=(\roman*)]
\item\label{tilde mean} $\EE(\tilde{\coupling})=\EE(\coupon)[1+O(L^{-\epsilon})]$.
\item\label{tilde variance} $\limsup_{L\to\infty}\var(\tilde{\coupling})/\var(\coupon)\le1$.
\end{enumerate}
\end{lemma}
\begin{proof}
By construction of the processes $\tilde{\bottom}_t$ and $\tilde{\top}_t$, we have $\tilde{\coupling}\in\{W_1,W_2,W_3,\ldots\}$.
Defining the random index $J$ via 
$$
J:=\inf\{j\in\naturals : \tilde{\bottom}_{W_j} = \tilde{\top}_{W_j}\},
$$
we therefore have $\tilde{\coupling}=W_J$. It follows that 
$$
\tilde{\coupling} = \sum_{j=1}^J Y_j,
$$
where $Y_j:=(W_j-W_{j-1})$ form an iid sequence of copies of $W_1=\coupon$.
Moreover, $J$ is geometrically distributed, with success probability $\PP(\coupling = \coupon)$.
From Proposition~\ref{prop:coupon and coupling are asymptotic when d=1} we therefore obtain
\begin{equation}
\begin{split}
\EE(J) &= 1 + O(L^{-\epsilon}),\\
\var(J) &= O(L^{-\epsilon}).
\label{J mean and variance asymptotics}
\end{split}
\end{equation}

Let $\sF_t:=\sigma(\sE_1,U_1,\ldots,\sE_t,U_t)$ denote the natural filtration of the auxiliary noise. For each $j\in\posint$, the time
$W_j$ is a stopping time with respect to $\sF_t$, and we can define the $\sigma$-algebra $\sG_j:=\sF_{W_j}$. Since $W_{j-1}<W_j$, the
sequence $(\sG_j)_{j\in\posint}$ is a filtration, and moreover, $(Y_j)_{j\in\posint}$ is adapted to it.  It is easily verified that
$\sigma(Y_j)$ and $\sG_{j-1}$ are independent, for each $j\in\posint$. Furthermore, $J$ is a stopping time with respect to
$(\sG_j)_{j\in\posint}$. It therefore follows from Wald's first equation~\cite[Theorem 5.3.1]{ChowTeicher78} that
\begin{equation}
\EE(\tilde{\coupling}) = \EE(J)\,\EE(\coupon).
\label{output of wald's first equation}
\end{equation}
Combining~\eqref{J mean and variance asymptotics} with~\eqref{output of wald's first equation} yields statement~\ref{tilde mean}.

We now turn to statement~\ref{tilde variance}. Consider the random variable $\tilde{\coupling}-\EE(\coupon)\,J$. 
We clearly have
$$
\tilde{\coupling} - \EE(\coupon)\,J = \sum_{j=1}^J [Y_j - \EE(\coupon)],
$$
and it follows from~\eqref{output of wald's first equation} that $\tilde{\coupling} - \EE(\coupon)\,J$ has mean zero. 
Wald's second equation~\cite[Theorem 5.3.3]{ChowTeicher78} therefore yields
\begin{equation}
\EE[(\tilde{\coupling} - \EE(\coupon)\,J)^2] = \EE(J)\,\var(Y_1-\EE(\coupon)) = \EE(J)\,\var(\coupon).
\label{output of wald's second equation}
\end{equation}

We can upper-bound $\var(\tilde{\coupling})$ using~\eqref{output of wald's second equation} as follows.
Fix a parameter $a>1$. Jensen's inequality implies that for any $b,c\in\RR$ we have
\begin{equation}
  (b+c)^2  = \left(\frac{1}{a} ba + \frac{a-1}{a}  \frac{ca}{(a-1)}\right)^2 \le b^2 a + c^2 \frac{a}{(a-1)}.
\label{elementary jensen inequality}
\end{equation}
From~\eqref{output of wald's first equation} and~\eqref{elementary jensen inequality} it follows that, for any $a>1$,
\begin{equation}
\begin{split}
 \var(\tilde{\coupling})  &=  \EE\left(\Big[\tilde{\coupling} -  \EE(\coupon)\EE(J)\Big]^2\right)\\
 &= \EE\left(\left[(\tilde{\coupling} - \EE(\coupon)\,J) +   \EE(\coupon)(J - \EE(J))\right]^2\right)\\
 &\le  a\, \EE\left[\left(\tilde{\coupling} - \EE(\coupon)\,J\right)^2\right] + \frac{a}{(a-1)} (\EE\coupon)^2 \var(J)\\
 &= a\,\EE(J)\,\var(\coupon) + \frac{a}{a-1}\,(\EE\coupon)^2 \var(J),
\label{bound for tilde coupling via coupon and geometric}
\end{split}
\end{equation}
where the last step follows from~\eqref{output of wald's second equation}. From~\eqref{ratio of coupon squared mean and variance}
and~\eqref{J mean and variance asymptotics} we therefore obtain that, for any $a>1$,
$$
\frac{\var(\tilde{\coupling})}{\var(\coupon)} \le a\,[1+O(L^{-\epsilon})] +\frac{a}{1-a} \frac{(\EE\coupon)^2}{\var(\coupon)} O(L^{-\epsilon})
\le a + o(1),
$$
and we conclude that 
\begin{equation}
\limsup_{L\to\infty} \frac{\var(\tilde{\coupling})}{\var(\coupon)} \le a. 
\label{variance of tilde coupling over variance of coupon}
\end{equation}
Finally, since~\eqref{variance of tilde coupling over variance of coupon} holds for all $a>1$, we in fact have
$$
\limsup_{L\to\infty} \frac{\var(\tilde{\coupling})}{\var(\coupon)} \le 1,
$$
as claimed.
\end{proof}

\subsection{Distribution}
\label{sec:FK one dimension distribution}
By combining Proposition~\ref{prop:coupon and coupling are asymptotic when d=1} with Parts~\ref{thm:d=1 coupling time mean} and~\ref{thm:d=1
  coupling time variance} of Theorem~\ref{thm: d=1}, we can now prove Part~\ref{thm:d=1 coupling time distribution}.
\begin{proof}[Proof of Theorem~\ref{thm: d=1}, Part~\ref{thm:d=1 coupling time distribution}]
Fix $q\ge1$ and $p\in(0,1)$, and let $\PP_L$ denote the corresponding measure for the FK heat-bath coupling on $\ZZ_L$, with analogous notation for
expectation and variance. Define the sequences $g_L:=\EE_L(\coupling)$, $\gamma_L:=\EE_L(\coupon)$, $h_L:=\sqrt{\var_L(\coupling)}$ and $\eta_L:=\sqrt{\var_L(\coupon)}$.
Proposition~\ref{prop:coupon and coupling are asymptotic when d=1} implies that for any fixed $x\in\RR$, we have
\begin{equation}
\begin{split}
\PP_L[\coupling \le g_L + x\,h_L] 
&= \PP_L[\coupon \le g_L + x\,h_L | \coupling=\coupon] + O(L^{-\epsilon}) \\
&= \PP_L[\coupon \le g_L + x\,h_L] + O(L^{-\epsilon}).
\end{split}
\label{eq:distribution of T vs distribution of W}
\end{equation}
Since Parts~\ref{thm:d=1 coupling time mean} and~\ref{thm:d=1 coupling time variance} of Theorem~\ref{thm: d=1} imply, respectively, that
$\gamma_L\sim g_L$ and $\eta_L\sim h_L$, the stated result follows from~\eqref{eq:distribution of T vs distribution of W}
via the Convergence of Types theorem~\cite[Theorem 14.2]{Billingsley94}.
\end{proof}

\subsection{Relaxation time}
\label{sec:FK one dimension relaxation time}
\begin{proof}[Proof of Theorem~\ref{thm: d=1}, Part~\ref{thm:d=1 relaxation time}]
Let $P$ denote the transition matrix of the FK process on $\ZZ_L$ with parameters $(p,q)$, and let $\fk$ denote the corresponding stationary
distribution.  For $g,h:2^{E}\to\RR$, let $\<g,h\>_{\fk}:=\sum_{A\subseteq E}g(A)\,h(A)\,\fk(A)$ denote the inner product on $l^2(\fk)$.  The
Dirichlet form $\sE$ corresponding to $P$ and $\fk$ is defined by $\sE(g,h):=\<(I-P)g,h\>_{\fk}$.  It is well-known (see e.g.~\cite[Lemma
  13.11]{LevinPeresWilmer09}) that
\begin{equation}
\sE(g) := \sE(g,g) = \frac{1}{2} \sum_{A,B\subseteq E} \nabla_g(A,B)\,Q(A,B),
\label{nice expression for dirichlet form}
\end{equation}
where
\begin{equation}
\begin{split}
\nabla_g(A,B) &:= [g(A) - g(B)]^2 = \nabla_g(B,A),\\
Q(A,B) &:= \fk(A)\,P(A,B) = Q(B,A)
\end{split}
\label{dirichlet difference definition}
\end{equation}
We denote the spectral gap of $P$ by $\gamma:=1-\secondeigenvalue$.
The Rayleigh-Ritz characterization~\cite[Remark 13.13]{LevinPeresWilmer09} of the spectral gap implies that
\begin{equation}
\gamma = \min_{\substack{g:2^E\to\RR \\ \var_{\fk}(g)\neq0}} \frac{\sE(g)}{\var_{\fk}(g)}.
\label{eq:Rayleigh-Ritz}
\end{equation}

We can bound the spectral gap of $P$ via a comparison with percolation with edge probability $\tilde{p}$. In what
follows, the quantities $\tilde{P}$, $\tilde{Q}$, $\tilde{\sE}$, $\tilde{\gamma}$, $\tilde{\fk}$ are defined analogously to $P$, $Q$, $\sE$,
$\gamma$, $\fk$, but with parameters $(\tilde{p},1)$ rather than $(p,q)$.

Replacing the number of components $k(A)$ in~\eqref{FK distribution} with the cyclomatic number $c(A)=L-|A|+k(A)$, and using the fact that
$c(A)=1$ iff $A=E$, and $c(A)=0$ otherwise, we find 
\begin{equation}
\begin{split}
\fk(A) &= r_L\,q^{\indicator(A=E)}\,\tilde{\fk}(A),\\
r_L&:=\frac{1}{1+(q-1)\tilde{p}^L}.
\end{split}
\label{relating mu and tilde mu}
\end{equation}

If $A$ and $B$ are both different from $E$, then $P(A,B)=\tilde{P}(A,B)$ and so 
$$
Q(A,B)=r_L\,\tilde{Q}(A,B). 
$$
By contrast, for any $e\in E$, 
$$
Q(E,E_e) = r_L \frac{p}{\tilde{p}} \,\tilde{Q}(E,E_e) = r_L \,\tilde{Q}(E,E_e) + c\, \frac{\tilde{p}^L}{L}
$$
where $c>0$ depends only on $p$ and $q$. It follows that
\begin{equation}
\sE(g) = r_L\,\tilde{\sE}(g) + c\,r_L\,\frac{\tilde{p}^L}{L}\,\sum_{e\in E}\nabla_g(E,E_e).
\label{explicit dirichlet form}
\end{equation}

Due to the product form of $\tilde{P}$, an explicit diagonalization can be easily obtained. A discussion of the case $\tilde{p}=1/2$ can be
found in~\cite[Example 12.15]{LevinPeresWilmer09}, which can be extended to any $\tilde{p}\in(0,1)$, to show that the eigenvalues of
$\tilde{P}$ have the form $1-k/L$ for $0\le k \le L$, and to obtain explicit forms for the corresponding eigenfunctions.  In particular,
this shows that the second-largest eigenvalue is $\tilde{\lambda}_2=1-1/L$, and so $\tilde{\gamma}=1/L$.

Fix an edge $e\in E$, and let $J_e:=\{A\subseteq E : A\ni e\}$, the event that $e$ is occupied.
It can be easily verified directly that the function $\Psi:2^E\to\RR$ defined by
$$
\Psi(A) = \indicator_{J_e}(A) - \tilde{p}
$$
is an eigenfunction of $\tilde{P}$ with eigenvalue $\tilde{\lambda}_2$. It follows that
\begin{equation}
\tilde{\sE}(\Psi)=\<(I-\tilde{P})\Psi,\Psi\>_{\tilde{\fk}} = \tilde{\gamma}\,\var_{\tilde{\fk}}(\Psi) = \dfrac{\tilde{p}(1-\tilde{p})}{L}.
\label{percolation dirichlet form of percolation eigenvector}
\end{equation}
Since $\fk(J_e^c) = r_L\tilde{\fk}(J_e^c)$, we have
\begin{equation}
\var_{\fk}(\Psi) = r_L^2\,\tilde{p}(1-\tilde{p})\left(1+(q-1)\tilde{p}^{L-1}\right).
\label{variance of percolation eigenfunction}
\end{equation}
Combining~\eqref{eq:Rayleigh-Ritz}, \eqref{explicit dirichlet form}, \eqref{percolation dirichlet form of percolation eigenvector}
and~\eqref{variance of percolation eigenfunction} we see that, as $L\to\infty$,
$$
\gamma \le \frac{\sE(\Psi)}{\var_{\fk}(\Psi)} = \frac{1}{L}\left[1 + O\left(\tilde{p}^L\right)\right].
$$
This establishes the stated lower bound for $\rel$. 

To establish the upper bound, first note that~\eqref{explicit dirichlet form} implies $\sE(g)\ge r_L\,\tilde{\sE}(g)$ for all $g:2^E\to\RR$.
Similarly, for any $g:2^E\to\RR$ we have 
\begin{equation*}
\begin{split}
\var_{\fk}(g) &= \frac{1}{2} \sum_{A,B\subseteq E} \nabla_g(A,B)\,\fk(A)\,\fk(B) \\
&\le q r_L^2 \var_{\tilde{\fk}}(g),
\end{split}
\end{equation*}
where the inequality follows by inserting~\eqref{relating mu and tilde mu} and noting that $\nabla_g(E,E)=0$.  It follows that, for any
non-constant $g:2^E\to\RR$, we have
$$
\frac{\sE(g)}{\var_{\fk}(g)} \ge \frac{1}{q\,r_L} \frac{\tilde{\sE}(g)}{\var_{\tilde{\fk}}(g)},
$$
and so
$$
\gamma \ge \frac{\tilde{\gamma}}{q\,r_L} = \frac{1}{q\,L}\left[1+O(\tilde{p}^L)\right].
$$
\end{proof}

\section{Single-spin Ising heat-bath process}
\label{sec:Ising}
In this section, we present a brief discussion of the coupling time for the single-spin Ising heat-bath process.
Since the process has exponentially slow mixing below the critical temperature, we focus on temperatures at and above criticality.
At temperatures above criticality, we find that the coupling time again displays the same coupon-collector-like behaviour observed for the
FK heat-bath process. As we shall see, however, at the critical temperature the behaviour is somewhat different.

We define the Ising heat-bath process precisely in Section~\ref{subsec:definition of ising process}, and in Section~\ref{subsec:ising
  heat-bath results} we summarise our conjectures for the behaviour of its coupling time. Sections~\ref{subsec:ising moments}
to~\ref{subsec:ising limit law} then outline the numerical evidence in support of these conjectures.

\subsection{Definition of the process}
\label{subsec:definition of ising process}
The zero-field ferromagnetic Ising model on finite graph $G=(V,E)$ at inverse temperature $\beta\ge0$ is defined by the Gibbs measure
\begin{equation}
\ising(\omega) \propto \exp\left(\beta\sum_{ij\in E}\omega_i \omega_j \right), \qquad \omega\in\{-1,1\}^V.
\label{Ising distribution}
\end{equation}
It is intimately related to the $q=2$ Fortuin-Kasteleyn random-cluster model. The correlated percolation transition displayed
by the FK model on $\ZZ^d$, when $d\ge2$, manifests itself as an order-disorder transition in the Ising model at a critical $0 < \betac <
\infty$.  This transition is known to be continuous~\cite{AizenmanDuminilCopinSidoravicius15}.  The two-dimensional model is particularly
well-understood~\cite{McCoyWu73}, where it is known that $\betac=\ln\sqrt{1+\sqrt{2}}$.

The single-spin Ising heat-bath process is a Markov chain with stationary distribution~\eqref{Ising distribution}, which can be defined by
the following random mapping representation.  Let $\sV$ and $U$ be independent random variables, with $\sV$ uniform on $V$ and $U$ uniform
on $[0,1]$.  For $v\in V$ and $\omega\in\{-1,1\}^V$, let $S_v(\omega)=\sum_{w\sim v}\omega_w$ denote the local magnetization at $v$ in
configuration $\omega$, where the notation $w\sim v$ denotes adjacency between vertices $w$ and $v$.
Then define {$f:\{-1,1\}^V\times V\times[0,1]\to\{-1,1\}^V$} so that $f(\omega,v,u)=\omega'$ where, for each $w\in
V$,
\begin{equation}
  \omega_w' := 
  \begin{cases}
  \omega_w, & w\neq v,\\
  +1, & w=v \text{ and } u \le   \dfrac{e^{\beta S_v(\omega)}}{e^{\beta S_v(\omega)}+e^{-\beta S_v(\omega)}},\\
    \\
  -1, & w=v \text{ and } u > \dfrac{e^{\beta S_v(\omega)}}{e^{\beta S_v(\omega)}+e^{-\beta S_v(\omega)}}.
\end{cases}
\label{Ising random mapping representation}
\end{equation}
The set $\{-1,1\}^V$ has a natural partial order such that $\omega<\omega'$ iff $\omega_v<\omega_v'$ for all $v\in V$. It is straightforward
to verify that $f$ is monotonic with respect to this partial order; i.e. for any fixed $v\in V$ and $u\in[0,1]$, if $\omega<\omega'$, then
$f(\omega,v,u)< f(\omega',v,u)$.

Let $(\sV_t,U_t)_{t\in\posint}$ be an iid sequence of copies of $(\sV,U)$. Analogous to the FK heat-bath process, we define top and bottom
chains corresponding to the random mapping representation~\eqref{Ising random mapping representation}.  Specifically, we define the top
chain $(\top_t)_{t\in\naturals}$ so that $\top_0=(+1,\ldots,+1)$ and $\top_{t+1}=f(\top_t,\sV_{t+1},U_{t+1})$, and the bottom chain
$(\bottom_t)_{t\in\naturals}$ so that $\bottom_0=(-1,\ldots,-1)$ and $\bottom_{t+1}=f(\bottom_t,\sV_{t+1},U_{t+1})$.  We refer to the
coupled process $(\sB_t,\sF_t)_{t\in\posint}$ as ``the Ising heat-bath coupling''.  With these definitions for the top and bottom chains,
the \emph{coupling time} of the Ising heat-bath process is again defined by~\eqref{coupling time definition}.

\subsection{Behaviour of the coupling time for the Ising heat-bath process}
\label{subsec:ising heat-bath results}
We now summarise our expectations for the behaviour of the coupling time for the Ising heat-bath process. Numerical evidence in support of these
conjectures will be presented in the following sections.
\begin{conjecture}
\label{conj:ising heat-bath}
Consider the Ising heat-bath process on $\ZZ_L^d$ with $d\ge1$. As $L\to\infty$:
\begin{enumerate}[label=(\roman*)]
\item\label{conj-part:ising supercritical moments} $\mu_\coupling\sim C_1(\beta,d)\, \mu_\coupon$ and $\sigma_\coupling\sim C_2(\beta,d)\, \sigma_\coupon$
  when $\beta<\betac$ with {$C_1(\beta,d),C_2(\beta,d)>0$}
\item\label{conj-part:ising critical moments} $\mu_\coupling/\sigma_\coupling \to C_3(d)$ at $\beta=\betac$, with $C_3(d)>0$
\item\label{conj-part:ising relaxation time} $\sigma_\coupling \sim C_4(\beta,d)\,\rel$ for all $\beta\le \betac$, with $C_4(\beta,d)>0$.
Moreover, $C_4(\beta,d)=\pi/\sqrt{6}$ for all $\beta<\betac$ and all $d$.
\item\label{conj-part:ising limit law} If $\beta\le\betac$
$$
\lim_{L\rightarrow \infty} \PP[\coupling_L \leq \EE(\coupling_L) + x \sqrt{\var(\coupling_L)}]=
F(x), \qquad \text{ for each } x\in \RR
$$
for some non-degenerate distribution function $F$. Moreover, $F(x)=G(x)$ for all $\beta<\betac$, where $G(x)$ is the Gumbel distribution
defined by~\eqref{gumbel distribution function}.
\end{enumerate}
\end{conjecture}

The numerical results presented in Section~\ref{subsec:ising limit law} strongly suggest that the limit law conjectured in
Part~\ref{conj-part:ising limit law} is not a Gumbel distribution when $\beta=\betac$. We offer no conjecture on the form of the limiting
distribution in this case; it appears to be an interesting open question. Similarly, we offer no conjecture for the exact form of $C_4(\beta,d)$ 
at $\beta=\betac$.

Preliminary results, for very small $L$ values with $d=2$, suggest that $(\coupling-\mu_\coupling)/\sigma_\coupling$ also converges to a
non-degenerate limit law as $L\to\infty$ when $\beta>\betac$, which again appears not to be $G(x)$. Furthermore, it also seems plausible
that $\sigma_\coupling \asymp \expauto$ remains true when $\beta>\betac$. However, given the computational difficulties in simulating this
regime, we have not attempted to test these predictions for $\beta>\betac$ in a detailed manner, and we therefore do not include their
statements in Conjecture~\ref{conj:ising heat-bath}.

\subsection{Moments}
\label{subsec:ising moments}
We begin by considering the high-temperature regime.  Fig.~\ref{fig:supercritical_ising_moments} plots the $L$ dependence of
$\mu_\coupling/\mu_\coupon$ and $\sigma_\coupling/\sigma_\coupon$ with $d=2$ and $\beta=0.4<\betac$. The data clearly support
Part~\ref{conj-part:ising supercritical moments} of Conjecture~\ref{conj:ising heat-bath}. We note that $C(\beta,d)$ seems to be strictly larger
than 1, and strongly $\beta$ dependent.
\begin{figure}
\centering
\includegraphics[width=0.49\columnwidth]{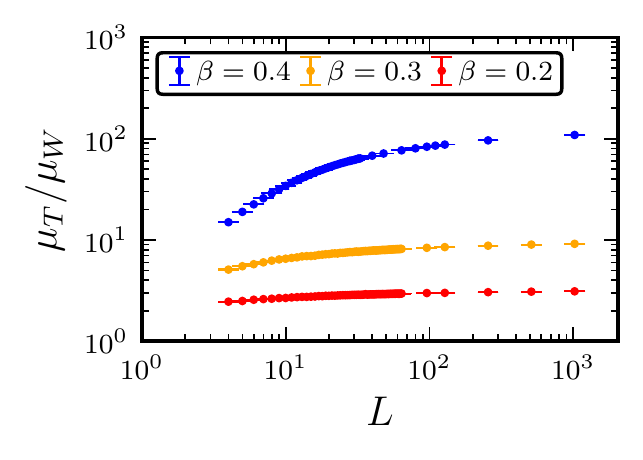}
\includegraphics[width=0.49\columnwidth]{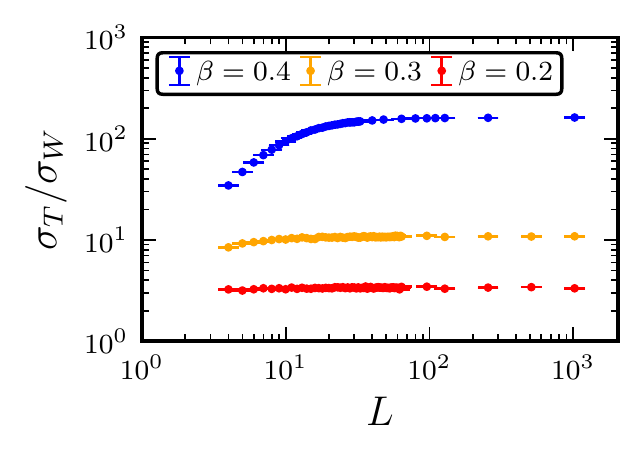}
    \caption{
      Monte Carlo estimates of $\mu_\coupling/\mu_\coupon$ (left) and $\sigma_\coupling/\sigma_\coupon$ (right) for the Ising heat-bath process with $d=2$
      and $\beta<\betac$ values as specified in the figure.
      Error bars corresponding to one standard error are shown.
\label{fig:supercritical_ising_moments}}
\end{figure}

Turning to the critical case, Fig.~\ref{fig:ising_critical_moments} shows the $L$ dependence of $\mu_\coupling$ and $\sigma_\coupling$ for $d=2$.
The figure clearly suggests that both $\mu_\coupling/L^d$ and $\sigma_\coupling/L^d$ diverge like a power law in $L$, with the \emph{same} exponent.
A least squares analysis for $\mu_\coupling$ produces a power-law exponent $2.168(4)$, while an analogous analysis for $\sigma_\coupling$ produces 
an exponent
\begin{equation}
z_\coupling=2.166(9).
\label{eq:ising z_T estimate}
\end{equation}
The combination of the figure and the fits lends strong support to Part~\ref{conj-part:ising critical moments} of Conjecture~\ref{conj:ising heat-bath}, 
that $\mu_\coupling/\sigma_\coupling$ approaches a constant as $L\to\infty$.
 \begin{figure}
 \centering
 \includegraphics[width=0.75\columnwidth]{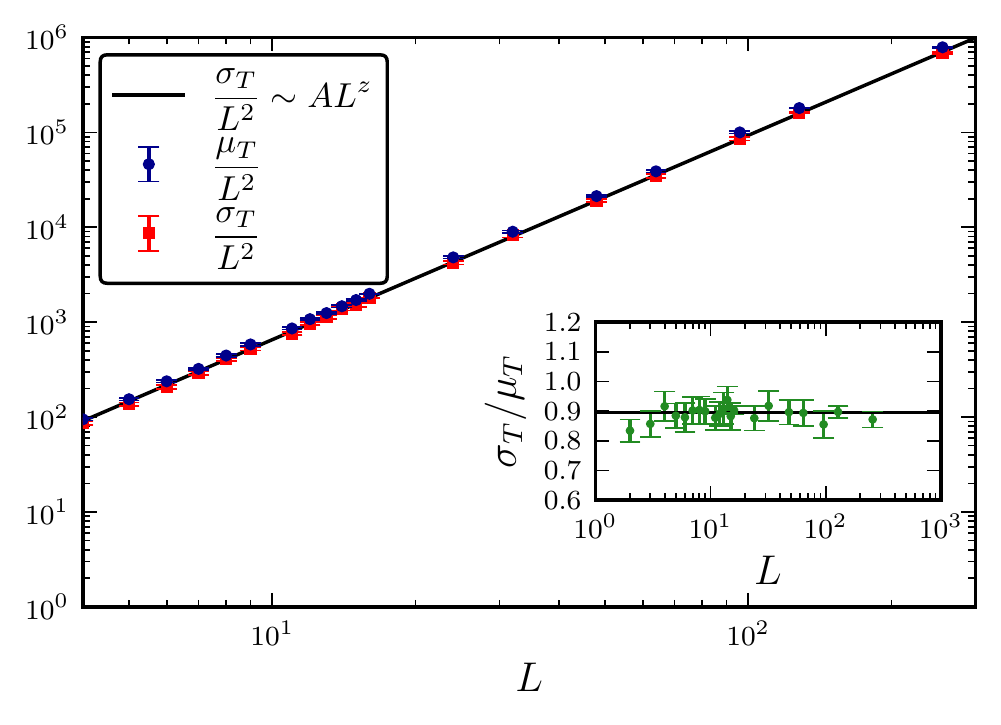}
 \caption{
 \label{fig:ising_critical_moments}
 Monte Carlo estimates of $\mu_\coupling/L^d$ and $\sigma_\coupling/L^d$ for the critical Ising heat-bath process with $d=2$.
 The solid black line shows the curve $A L^z$, with the estimated values of $A$ and $z=2.166$.
 The inset shows the ratio $\sigma_\coupling/\mu_\coupling$. The solid line within the inset corresponds to the estimated
 asymptotic limit of $\sigma_\coupling/\mu_\coupling\rightarrow 0.895(8)$.
 Error bars corresponding to one standard error are shown.
 }
 \end{figure}

\subsection{Variance and relaxation time}
\label{subsec:ising std vs relaxation time}
We now turn attention to Part~\ref{conj-part:ising relaxation time} of Conjecture~\ref{conj:ising heat-bath}. We first consider the case
$d=1$, where the relaxation time can be calculated explicitly. It was shown in~\cite[Lemma 4]{Nacu03} that if the transition matrix, $P$, of
the Ising heat-bath process (on any graph) has a strictly increasing eigenfunction, then its eigenvalue is the second-largest eigenvalue, $\lambda_2$. The
total magnetization $\mathcal{M} = \sum_{i=1}^L \omega_i$ is clearly strictly increasing. Moreover, on $\mathbb{Z}_L$ it is known (see
e.g. the proof of Theorem 15.4 in~\cite{LevinPeresWilmer09}) that $\sM$ is an eigenfunction of $P$ with eigenvalue
\begin{equation}
\lambda(\beta) = 1-\frac{1-\tanh{(2\beta)}}{L}.
\end{equation}
This immediately yields the following closed-form expression for the relaxation time on $\ZZ_L$
\begin{equation}
\rel(L) = \frac{L}{1 - \tanh(2 \beta)}.
\label{eq:ising 1d trel}
\end{equation}
Fig.~\ref{fig:ising1dstd} compares Monte Carlo estimates of $\sigma_\coupling$ on $\ZZ_L$ with the exact expression for $\rel$ given
in~\eqref{eq:ising 1d trel}. The agreement is clearly excellent, over the entire range of $\beta$ considered, thus lending strong support to 
Part~\ref{conj-part:ising relaxation time} of Conjecture~\ref{conj:ising heat-bath} in the case $d=1$.
\begin{figure}
\centering
\includegraphics[width=0.75\columnwidth]{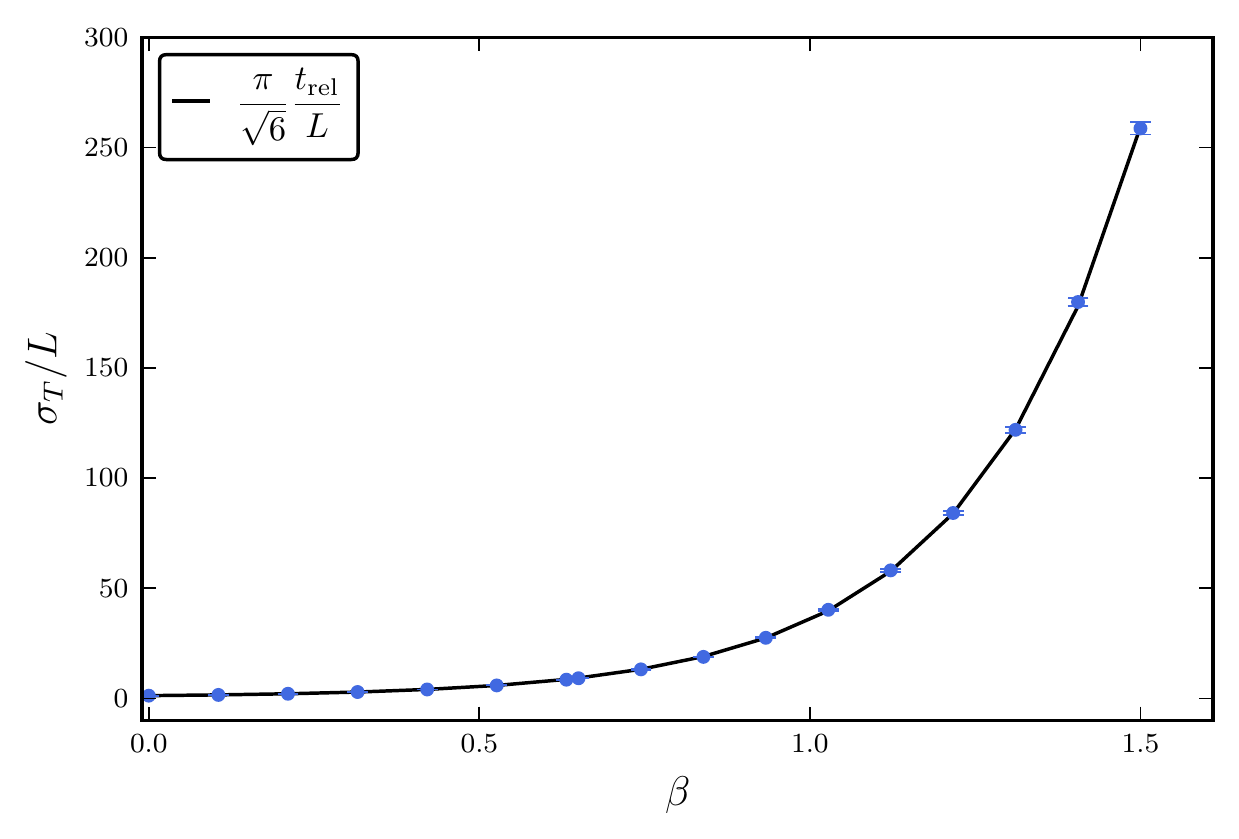}
    \caption{Monte Carlo estimates of $\sigma_\coupling/L$ for the Ising heat-bath process on $\ZZ_L$ with $L=10^4$. The blue curve 
      corresponds to the exact expression for $\rel/L$ given in~\eqref{eq:ising 1d trel}. 
      Error bars corresponding to one standard error are shown.
\label{fig:ising1dstd}}
\end{figure}

We now consider the case $d>1$, using analogous arguments to those presented in Section~\ref{subsec:fk relaxation time} in the FK setting.  Let
$(X_t)_{t\in\naturals}$ be a stationary Ising heat-bath process, and define $(\sM_t)_{t\in\naturals}$ via $\sM_t=\sM(X_t)$.  Although
Proposition~\ref{prop:decay of rho_N} is stated in the specific context of the FK heat-bath process, the positive association of the Ising
measure~\eqref{Ising distribution} (see e.g.~\cite[Theorem 3.31]{FriedliVelenik16}) implies that the proof of Lemma~\ref{lem:projection},
and then also the proof of Proposition~\ref{prop:decay of rho_N}, immediately extend to the Ising heat-bath process.  It follows that, since
the magnetization is strictly increasing, we have
\begin{equation}
\rho_{\sM}(t) \sim C e^{-t/\expauto}, \qquad t\to\infty
\label{eq:asymptotics of rho_M}
\end{equation}
for some (parameter dependent) constant $C>0$. Assuming the validity of 
Part~\ref{conj-part:ising relaxation time} of Conjecture~\ref{conj:ising heat-bath}, it follows
from~\eqref{eq:asymptotics of rho_M} that
\begin{equation}
\ln\rho_{\sM}(k\,\sigma_\coupling) \sim -C(\beta,d)\, k
\label{ising scaled rho ansatz}
\end{equation}
as $k$ and $L$ tend to infinity, with $C(\beta,d)>0$, and with $C(\beta,d)=\pi/\sqrt{6}$ for all $\beta<\betac$. 

For a given time lag $t$, we estimated $\rho_{\sM}(t)$ using the procedure described in Section~\ref{subsec:fk relaxation time} for the
estimation of $\rho_{\sN}(t)$ for the FK model.  Fig.~\ref{fig:ising rho scaled} shows the resulting estimates of
$\rho_{\sM}(k\,\sigma_\coupling)$ versus $k$ for $d=2$, in the high-temperature regime (left panel) and at criticality (right panel), for a
variety of $L$ values. In both cases, the data collapse evident in the figure clearly supports the expectation~\eqref{ising scaled rho
  ansatz}, and therefore provides direct evidence to support the conjecture that $\sigma_\coupling \sim C(\beta,d)\, \rel$. Moreover, in the
high-temperature case, the collapse of the curves arising from distinct temperature values onto a single curve corresponding to
$\exp(-k\pi/\sqrt{6})$, supports the claim that $C(\beta,d)=\pi/\sqrt{6}$ when $\beta<\betac$.
 \begin{figure}
 \centering
 \includegraphics[width=0.45\columnwidth]{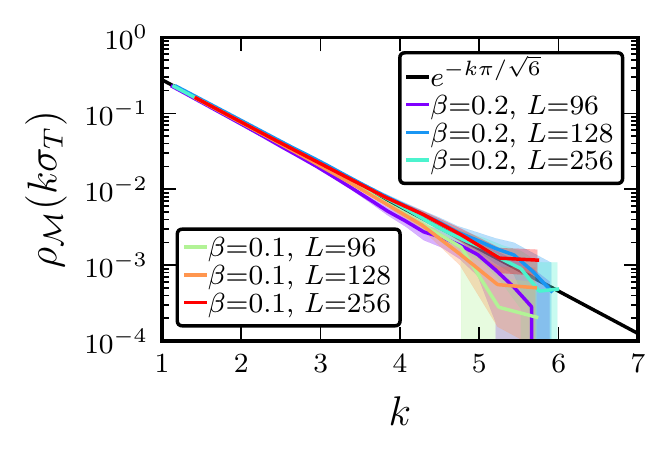}
 \includegraphics[width=0.45\columnwidth]{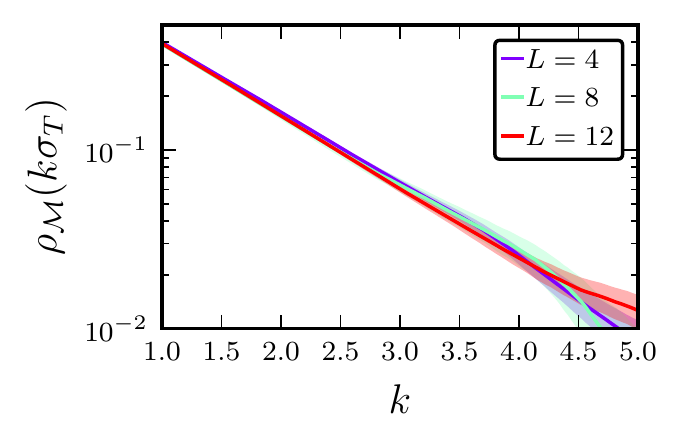}
 \caption{\label{fig:ising rho scaled}
 Monte Carlo estimates of $\ln \rho_{\sM}(\sigma_\coupling\,k)$ for the Ising heat-bath process with $d=2$ in high temperature (left) and at
 criticality (right).  The enclosing filled regions correspond to one standard error.
 }
 \end{figure}

\begin{table}
\begin{center}
\begin{tabular}{|c|c|}
\hline
$L$ & $\sigma_\coupling/\expauto$ \\
\hline
 4 & 0.895(1) \\
 5 & 0.894(3) \\
 6 & 0.901(3) \\
 7 & 0.897(3) \\
 8 & 0.889(3) \\
 9 & 0.898(4) \\
10 & 0.890(3) \\
11 & 0.893(3) \\
12 & 0.904(3) \\
13 & 0.893(3) \\
14 & 0.896(4) \\
15 & 0.894(3) \\
\hline
\end{tabular}
\end{center}
\caption{\label{tab:d=2 ising sigma/tauexp ratios}
Ratios of estimated $\sigma_\coupling$ values to the estimated values of $\expauto$ from \cite{NightingaleBloete96}, for the critical Ising
heat-bath process when $d=2$.
Error bars corresponding to one standard error are shown.
}
\end{table}

In the critical case, we have no explicit conjecture for the value of $C(\beta,d)$. However, using the critical $d=2$ values of $\expauto$
reported in~\cite{NightingaleBloete96}, we computed the ratios $\sigma_\coupling/\expauto$, which are reported in Table~\ref{tab:d=2 ising
  sigma/tauexp ratios}. The first observation to make is that, for the $L$ values considered, there appears to be
extremely weak $L$ dependence; in fact, the size of any $L$ dependence appears to be smaller than our statistical errors. In particular,
this gives direct, independent, support to the conjectured asymptotic proportionality of $\sigma_\coupling$ and $\expauto$. Furthermore, it
suggests that we have $C(\betac,2)\approx 0.895$. It is interesting to note that this constant agrees, within error bars, with the constant
of proportionality relating $\sigma_\coupling$ to $\mu_\coupling$, reported in Fig.~\ref{fig:ising_critical_moments}, suggesting the
possibility that $\mu_\coupling \sim \expauto$ at criticality, at least when $d=2$.

Finally, as yet further evidence to support Part~\ref{conj-part:ising relaxation time} of Conjecture~\ref{conj:ising
  heat-bath} in the critical case, we note that the estimated value of the exponent~\eqref{eq:ising z_T estimate}, governing
$\sigma_\coupling$ at criticality for $d=2$, agrees, within error bars, with Grassberger's~\cite{Grassberger95} estimate for the dynamic
exponent $z_{\exp} = 2.172(6)$.

\subsection{Limit law}
\label{subsec:ising limit law}
Fig.~\ref{fig:ising_off_critical_gumbel} plots the empirical distribution of the standardized coupling time
$\standard:=(\coupling-\mu_\coupling)/\sigma_\coupling$ for a high-temperature Ising heat-bath process with $d=2$ and $L=1024$.  The
agreement with the Gumbel distribution~\eqref{gumbel distribution function} clearly supports Part~\ref{conj-part:ising limit law} of
Conjecture~\ref{conj:ising heat-bath} in the case $\beta<\betac$.  Fig.~\ref{fig:ising_critical_limit_law} shows the critical case, again
with $d=2$. The data collapse of the $L=128$ and $L=256$ curves strongly supports the claim that $\standard$ converges in distribution to a
non-degenerate limit, thus supporting Part~\ref{conj-part:ising limit law} of Conjecture~\ref{conj:ising heat-bath} in the case
$\beta=\betac$. However, it is clear that this limiting distribution is not $G(x)$.
\begin{figure}
\centering
\subfigure[\label{fig:ising_off_critical_gumbel}]{\includegraphics[width=0.49\columnwidth]{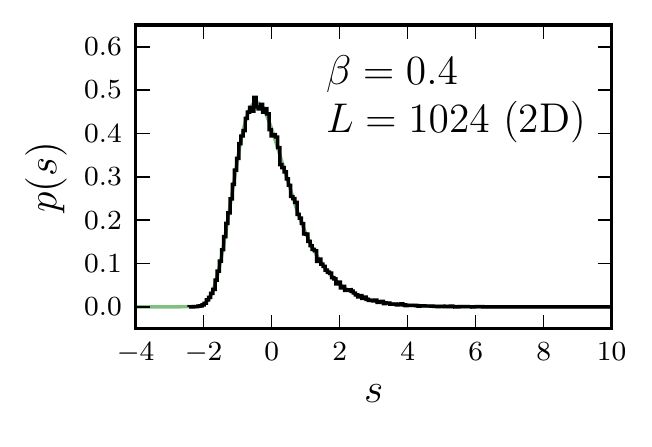}}
\subfigure[\label{fig:ising_critical_limit_law}]{\includegraphics[width=0.49\columnwidth]{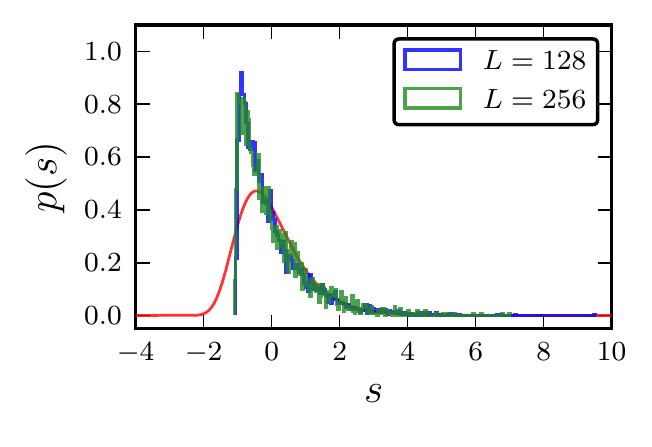}}
    \caption{
      Histogram of $\standard$ at high temperature (left) and criticality (right), with parameters as specified in the figure.
      Here $p(s)$ denotes the probability density function of $\standard$. For comparison, the solid
      green line shows the probability density function corresponding to \eqref{gumbel distribution function}.
}
\end{figure}

\begin{acknowledgements}
  The authors thank Youjin Deng, Alan Sokal, and Ulli Wolff for useful discussions. 
  This work was supported under the Australian Research Council's Discovery Projects funding scheme (project numbers DP140100559 \&
  DP110101141), and T.G. is the recipient of an Australian Research Council Future Fellowship (project number FT100100494). A.C. would like
  to thank STREP project MATHEMACS.
\end{acknowledgements}

\appendix
\section{Autocorrelation functions of strictly increasing observables}
\label{appendix:autocorrelation functions of increasing observables}
Let $P$ denote the transition matrix of the FK heat-bath process on a finite graph $G=(V,E)$ with parameters $p\in(0,1)$ and $q\ge1$, and let $k=2^{|E|}$.
To avoid trivialities, we assume $|E|>1$. We regard elements of $\RR^k$ as functions from $2^E$ to $\RR$, and
we endow $\RR^k$ with the inner product $\<\cdot,\cdot\>$ defined by
$$
\<g,h\> := \sum_{A\subseteq E} g(A)\,h(A)\,\fk(A).
$$
Denote the eigenvalues of $P$ by $1=\lambda_1 > \lambda_2\ge \ldots \ge\lambda_k$. 
As mentioned in Section~\ref{subsec:previous studies}, general results for heat-bath chains~\cite{DyerGreenhillUllrich14} imply that all $\lambda_i$ are 
non-negative. Let $\{\psi_i\}_{i=1}^k$ be an orthonormal basis for $\RR^k$ such that $\psi_i$ is an eigenfunction of $P$ corresponding to $\lambda_i$.
The Perron-Frobenius theorem implies that the eigenspace of $\lambda_1$ is one-dimensional, and that we can take $\psi_1(A)=1$ for all $A\subseteq E$. 
Let $W$ denote the eigenspace of $\lambda_2$. For $g\in\RR^k$, we let $g_W$ denote its projection onto $W$.

We say $g\in\RR^k$ is \emph{increasing} if $A\subset B$ implies $g(A)\le g(B)$, and \emph{strictly increasing} if $A\subset B$ implies
$g(A) < g(B)$.

\begin{proposition}\label{prop:decay of rho_N}
Let $(X_t)_{t\in\naturals}$ be a stationary FK heat-bath process, and for $g\in\RR^k$ define $(g_t)_{t\in\naturals}$ via $g_t:=g(X_t)$. 
If $g$ is strictly increasing, then its autocorrelation function satisfies
$$
\rho_{g}(t) := \frac{\cov(g_0,g_t)}{\var(g_0)} \sim C e^{-t/\expauto},\qquad t\to\infty,
$$
for constant $C>0$.
\begin{proof}
Let $\Pi$ denote the projection matrix onto the space of constant functions. 
General arguments (see e.g.~\cite{Sokal97} or~\cite[Chapter 9]{MadrasSlade96}) imply
$$
\cov(g_0,g_t) = \<g,(P^t-\Pi)g\> = \sum_{l=2}^k \<g,\psi_l\>^2\lambda_l^t = \|g_W\|^2\lambda_2^t \,+\, \sum_{l=\dim(W)+2}^k\<g,\psi_l\>^2\lambda_l^t.
$$
Since $g$ is strictly increasing, Lemma~\ref{lem:projection} implies that $\|g_W\|^2>0$, and therefore
$$
\cov(g_0,g_t) \sim \|g_W\|^2 e^{-t/\expauto}, \qquad t\to\infty.
$$
It follows that
$$
\rho_g(t)  \sim  \frac{\|g_W\|^2}{\var(g)} e^{-t/\expauto}, \qquad t\to\infty.
$$
\end{proof}
\end{proposition}

\begin{lemma}\label{lem:projection}
If $g$ is strictly increasing, then its projection onto $W$ is non-zero.
\begin{proof}
Lemma~\ref{lem:increasing eigenfunction} implies there exists $\psi\in W$ which is non-zero and increasing.
Positive association (see e.g.~\cite[Theorem 3.8 (b)]{Grimmett06}) then implies that for any other increasing $g$ we have
\begin{equation}
\<g,\psi\> \ge \EE(g)\,\EE(\psi) = 0,
\label{positive association}
\end{equation}
since $\EE(\psi)=\<\psi_1,\psi\>=0$. In particular, suppose that $g$ is strictly increasing. Choosing $\alpha>0$ so that 
$$
g(B)-g(A) > \alpha[\psi(B)-\psi(A)], \qquad \text{ for all } A \subset B\subseteq E,
$$
implies that $g-\alpha \psi$ is also strictly increasing. Applying~\eqref{positive association} to $g-\alpha \psi$ then yields
$$
\<g-\alpha \psi,\psi\>\ge0.
$$
Rearranging, and using the fact that $\psi$ is non-zero then implies
$$
\<g,\psi\> \ge \alpha\<\psi,\psi\> >0.
$$
Therefore, $g$ has a non-zero projection onto $\psi\in W$, and the stated result follows.
\end{proof}
\end{lemma}

The following lemma is the natural analogue, in the FK setting, of the result~\cite[Lemma 3]{Nacu03} established for the Ising heat-bath process.
\begin{lemma}\label{lem:increasing eigenfunction}
There exists $\psi\in W$ which is non-zero and increasing.
\begin{proof}
Let $g=\psi_2+ C (\sN-\EE(\sN))$, where $\sN\in\RR^k$ is defined so that $\sN(A)=|A|$ for each $A\subseteq E$, and $C>0$ is a constant. We have
$$
g = [1+C\<\sN,\psi_2\>]\psi_2 + C \sum_{j=3}^k \<\sN,\psi_j\>\psi_j.
$$
If $\<\sN,\psi_2\>=0$, then $g$ has a non-zero projection onto $\psi_2$, for any choice of $C>0$. If $\<\sN,\psi_2\>\neq0$, then choosing 
$C>|\<\sN,\psi_2\>|^{-1}$ suffices to guarantee that $g$ again has a non-zero projection onto $\psi_2$. In either case, assume $C$ is so chosen. 
It follows that $g_W$ is non-zero. 

If $A\subset B$, then
$$
g(B) - g(A) = \psi_2(B)-\psi_2(A) + C[\sN(B)-\sN(A)] \ge \min_{A\subset B\subseteq E} [\psi_2(B) - \psi_2(A)] + C.
$$
Therefore, by choosing $C>\left|\min\limits_{A\subset B\subseteq E} [\psi_2(B) - \psi_2(A)]\right|$ we guarantee that $g$ is increasing.
Lemma~\ref{lem: increasing functions have increasing projections} then implies that $g_W$ is increasing. Therefore, $\psi=g_W$ is an increasing, non-zero 
element of $W$.
\end{proof}
\end{lemma}

\begin{lemma}\label{lem: increasing functions have increasing projections}
If $g$ is increasing and has zero-mean, then its projection onto $W$ is also increasing.
\begin{proof}
Let $g\in\RR^k$ be any increasing observable with mean zero, and let $t\in\posint$. Since Lemma~\ref{lem: lambda_2 is positive} implies
$\lambda_2>0$, we can write
$$
\frac{P^t g}{\lambda_2^t}  
= g_W + \sum_{l=\dim(W)+2}^k \<g,\psi_l\>\psi_l\,\left(\frac{\lambda_l}{\lambda_2}\right)^t.
$$
It follows that
\begin{equation}
\lim_{t\to\infty} \frac{P^t g}{\lambda_2^t}  = g_W.
\label{eq:limit of P^t g}
\end{equation}

Now, for any given $t\ge1$, Lemma~\ref{lem:P preserves monotonicity} implies that $P^t g(A)$ is an increasing function of $A$, and so
$\lambda_2^{-t}P^t g (A)$ is also an increasing function of $A$. It then follows, as an elementary consequence of~\eqref{eq:limit of P^t g},
that $g_W$ is also increasing. We have therefore established that if $g$ is an increasing zero-mean function, then its
projection $g_W$ is also increasing.
\end{proof}
\end{lemma}

\begin{lemma}\label{lem:P preserves monotonicity}
If $g\in\RR^k$ is increasing, then $P^tg$ is also increasing, for every $t\ge1$.
\begin{proof}
Let $(f,\sE,U)$ be the random mapping representation for $P$ given in Section~\ref{subsec:definitions}; see~\eqref{random mapping definition}.
Let $A_1\subset A_2\subseteq E$, and let $B_i=f(A_i,\sE,U)$ for $i=1,2$. Clearly, $(B_1,B_2)$ is a coupling of the distributions $P(A_1,\cdot)$ and
$P(A_2,\cdot)$, and the monotonicity of $f$ implies $B_1\subseteq B_2$. Strassen's theorem (see e.g.~\cite[Theorem 4.2]{Grimmett10}) then implies that 
$$
\EE_{P(A_1,\cdot)}(g) \le \EE_{P(A_2,\cdot)}(g)
$$
for any increasing $g\in\RR^k$. It follows that 
$$
(Pg)(A_1) = \sum_{B\subseteq E} P(A_1,B) g(B) = \EE_{P(A_1,\cdot)}(g) \le \EE_{P(A_2,\cdot)}(g) = \sum_{B\subseteq E} P(A_2,B) g(B) = (Pg)(A_2).
$$
Since this holds for any $A_1\subset A_2\subseteq E$, it follows that $Pg$ is increasing. It then follows by a simple induction that $P^tg$ is 
increasing for any $t\ge1$.
\end{proof}
\end{lemma}

\begin{lemma}\label{lem: lambda_2 is positive}
The second-largest eigenvalue of $P$ is positive.
\begin{proof}
Since $P$ is reversible and irreducible we have the spectral decomposition (see e.g.~\cite[Lemma 12.2]{LevinPeresWilmer09}) 
$$
\frac{P(A,B)}{\fk(B)} = 1 + \sum_{j=2}^k\psi_i(A)\psi_i(B)\lambda_i.
$$
Since $\lambda_2\ge\lambda_j\ge0$ for all $j>2$, it follows that if $\lambda_2=0$, then $P(A,B)=\fk(B)$ for all $A,B\subseteq E$. But since,
by assumption, we have $|E|>1$, we can choose $A,B\subseteq E$ with $|A\symdif B| > 1$, where $\symdif$ denotes symmetric difference,
and~\eqref{transition matrix} then implies
$$
P(A,B) = 0 \neq \fk(B).
$$
We have therefore reached a contradiction, and we conclude that $\lambda_2>0$.
\end{proof}
\end{lemma}

\section{Coupon Collecting}
\label{appendix:coupon collecting}
Let $n\in\posint$, and let $C_1,C_2,\ldots$ be an iid sequence of uniformly random elements of $[n]:=\{1,2,\ldots,n\}$. For $t\in\posint$, we think
of $C_t$ as the \emph{coupon collected at time $t$}. For $i\in[n]$, let $D_i\in[n]$ denote the $i$th distinct type of coupon collected; i.e. the
$i$th distinct element of the sequence $C_1,C_2,\ldots$. Let $S_i(t):=\#\{s\le t : C_s = D_i\}$, the number of copies of $D_i$ collected by
time $t$. Define $R_t:=\{c\in[n]: C_s = c \textrm{ for some } s\leq t\}$, the set of distinct coupon types collected up to time $t$. For
any $1\le k \le n$, let $\coupon_k=\inf\{t\in\posint : |R_t| = k\}$, and note that $\coupon_k$ is simply the hitting time of $D_{k}$.
The \emph{coupon collector's time} is then defined as $\coupon:=\coupon_n$.

For each $c\in[n]$, define
$$
\lastvisit(c) = \sup\{t\le \coupon: C_t=c\}.
$$
We refer to the time $\lastvisit(c)$ as the \emph{last visit} to $c$. Let $(\lastvisit_i)_{i=1}^n$ denote the sequence of the
$\lastvisit(c)$, arranged in increasing order. In particular, $\lastvisit_1$ is the first time that a last visit occurs.

\begin{lemma}
There exists $\varphi>0$ such that $\PP(|R_{\lastvisit_{1}}| \le \floor{\ln n}) = O(n^{-\varphi})$.
\label{lem:the first visited edges are revisited often}
\end{lemma}
\begin{proof}
 Inserting $a_n=\lfloor\ln(n)\rfloor$ and $c_n=\lfloor\ln(n)/4\rfloor$ into Lemma~\ref{lem:bounding the number of copies of coupons at
   coupon collector time} and applying the union bound, implies
 \begin{align*}
 \PP\left(\bigcup_{i=1}^{a_{n}} \left\{S_{i}(\coupon) \le c_n \right\} \right)
 &\le \ln(n) \exp\left(-\frac{1}{2}\ln(n-a_n) + \frac{\ln(n)}{4} + 1\right)\\
 &= e\,\ln(n) \exp\left(-\frac{1}{4}\ln(n) - \frac{1}{2}\ln\left(1-\frac{a_n}{n}\right)\right)\\
 &= \frac{e}{\sqrt{1-\lfloor \ln(n)\rfloor/n}}\,\ln(n)\, n^{-1/4}.
 \end{align*}
 Therefore, for any $0<\rho<1/4$, we have
 $$
 \PP\left(\bigcup_{i=1}^{a_{n}} \left\{S_{i}(\coupon) \le c_n \right\} \right)
 = O(n^{-\rho}), \qquad n\to\infty.
 $$
 It follows that, 
 \begin{equation}
 \PP(|R_{\lastvisit_1}| \le a_n) = \PP\left(|R_{\lastvisit_1}|\le a_n, \bigcap_{i=1}^{a_n}\{S_i(\coupon)>c_n\}\right) + O(n^{-\rho})
 \label{eq:conditioning on c_n returns to a_n first edges}
 \end{equation}

 Let $I:=\inf\{t\in\posint : S_i(t) = c_n \textrm{ for some } i\in [n] \}$,
 the first time that there exists a coupon type for which exactly $c_n$ copies have been collected, and define the random
 variable $K\in[n]$ via $C_{\lastvisit_1}=D_K$. If $|R_{\lastvisit_1}|\le a_n$, then $1\le K \le a_n$.  Therefore, observing that
 $S_{K}(\coupon)=S_{K}(\lastvisit_1)$, we find
\begin{equation}
  \begin{split}
 \PP\left(|R_{\lastvisit_1}|\le a_n, \bigcap_{i=1}^{a_n}\{S_i(\coupon)>c_n\}\right)
 &\le
 \PP(|R_{\lastvisit_1}|\le a_n, S_{K}(\coupon) > c_n) \\
 &=
 \PP(|R_{\lastvisit_1}|\le a_n, S_{K}(\lastvisit_1) > c_n) \\
 &\le
 \PP(|R_I|\le a_n)
 \end{split}
  \label{eq:bounding conditioned R_\lastvisit_1 in terms of R_I}
 \end{equation}
 since if $|R_{\lastvisit_1}|\le a_n$ and $S_{K}(\lastvisit_1) > c_n$ then $|R_I|\le a_n$.  Combining~\eqref{eq:conditioning on c_n returns to a_n first
   edges} and~\eqref{eq:bounding conditioned R_\lastvisit_1 in terms of R_I} then implies
 $$
 \PP(|R_{\lastvisit_1}| \le a_n) \le \PP(|R_I|\le a_n) + O(n^{-\rho}).
 $$
 However, Lemma~\ref{lem:collecting log copies of a coupon requires visiting log distinct coupons} implies that there exists $\delta>0$ such
 that $\PP(|R_I|\le a_n) = O(n^{-\delta})$. We therefore conclude that, if $\varphi=\min\{\rho,\delta\}$, then
 $$
 \PP(|R_{\lastvisit_{1}}| \le \floor{\ln n}) = O(n^{-\varphi}).
 $$
\end{proof}

\begin{lemma}
\label{lem:bounding the number of copies of coupons at coupon collector time}
Let $(a_n)_{n\in\posint}$ and $(c_n)_{n\in\posint}$ be any two sequences of natural numbers.
For $n\in\posint$, if $a_n<n$ then for each $1\le i\le a_n$ we have
$$
\PP\left(S_i(\coupon) \le c_n \right) \le \exp\left(-\ln(\sqrt{n-a_n}) + c_n + 1\right).
$$
\end{lemma}
\begin{proof}
Fix $n\in\posint$ and $1\le i \le a_n$, and assume $a_n<n$. Adopting the convention $\coupon_0=0$, for $0\le k \le n-1$ we define
$$ 
Y_i(k) := \sum_{j=\coupon_{k}+1}^{\coupon_{k+1}-1} \indicator_{\{C_{j}= D_{i}\}}.
$$
Since $Y_i(k)=0$ for all $k<i$, and $C_{\coupon_k}=D_i$ iff $k=i$, we then have
$$
S_{i}(\coupon) = 1 + \sum_{k=i}^{n-1} Y_i(k).
$$
And since the random variables $Y_i(k)$ are independent, for any $\theta<0$, we have
\begin{equation}
\begin{split}
 \PP\left(S_i(\coupon) \le c_n \right)
 &\le \PP \left(\sum_{k=a_n}^{n-1} Y_i(k)  \le c_n \right)\\
 &= \PP \left(\exp\left[\theta \sum_{k=a_n}^{n-1} Y_i(k)\right]  \ge e^{\theta c_n} \right) \\
 &\le \exp\left(-\theta c_n + \sum_{k=a_n}^{n-1} \ln \EE[e^{\theta Y_i(k)}]\right),
\label{eq:upper tail bound for S_i(\coupon)}
\end{split}
\end{equation}
where the final step follows from Markov's inequality.

The moment generating function of $Y_i(k)$ can be calculated explicitly.
Let $i \le k \le n-1$. Given $\coupon_{k}$ and $\coupon_{k+1}$, the random variable $Y_i(k)$ has binomial distribution with
$\coupon_{k+1} -\coupon_{k}-1$ trials and success probability $1/k$, which implies
\begin{align*}
  \EE(e^{\theta Y_i(k)}) &= 
  \EE(\EE[e^{\theta Y_i(k)}|\coupon_{k},\coupon_{k+1}]) \\
  &= 
  \EE\left[\left(\frac{e^{\theta}}{k} + 1 - \frac{1}{k}\right)^{\coupon_{k+1} - \coupon_{k}-1}\right].
\end{align*}
But since $\coupon_{k+1}-\coupon_{k}$ has geometric distribution with parameter $1-k/n$, this becomes
$$
 \EE(e^{\theta Y_i(k)}) = \frac{n-k}{n-k+1-e^{\theta}}.
$$
Therefore, setting $\lambda=1-e^\theta$ and $b_n=n-a_n$, it follows from the fact that $\ln(1+\lambda/k)$ is a decreasing function of $k$ that
\begin{equation}
\begin{split}
   -\sum_{k=a_{n}}^{n-1} \ln \EE(e^{\theta Y_i(k)}) &= \sum_{k=1}^{b_n} \ln \left(1 + \frac{\lambda}{k} \right) \\
   &\ge \int_1^{b_n} \ln\left(1+\frac{\lambda}{x}\right) \d x \\
   &= \lambda \ln(b_n) + (b_n + \lambda) \ln(1 + \lambda/b_n) - (1+\lambda)\ln(1+\lambda)\\
   &\ge \lambda\ln(b_n) + \lambda - (1+\lambda)\ln(1+\lambda)\\
   &\ge \lambda\ln(b_n) - 1
\label{eq:moment generating function for sum of Y_i(k)}
\end{split}
\end{equation}
where, in the penultimate step, we used the fact that $\ln(1+x)\ge x/(1+x)$ holds for all $x>-1$, and in the last step we used the fact
that $(1+\lambda)-(1+\lambda)\ln(1+\lambda)>0$ for any $\lambda\in(0,1)$. Combining~\eqref{eq:upper tail bound for S_i(\coupon)} and~\eqref{eq:moment
generating function for sum of Y_i(k)}, we conclude that for all $\lambda\in(0,1)$ we have
$$
\PP\left(S_i(\coupon) \le c_n \right) \le \exp\left[-\lambda\ln(n-a_n) -\ln(1-\lambda)c_n + 1) \right].
$$
Choosing $\lambda=1/2$ yields the stated result.
\end{proof}

\begin{lemma}
\label{lem:collecting log copies of a coupon requires visiting log distinct coupons}
 Fix $c\in(0,1)$, and define sequences $(a_n)_{n\in\posint}$ and $(c_n)_{n\in\posint}$ such that $a_n=\lfloor\ln(n)\rfloor$ and 
 $c_n=\lfloor c\ln(n)\rfloor$. 
 Let
 $$
 I:=\inf\{t\in\posint : S_i(t) = c_n \textrm{ for some } i\in [n] \},
 $$
 the first time that there exists a coupon type for which exactly $c_n$ copies have been collected.
 Then there exists $\delta>0$ such that
 $$
 \PP(|R_{I}|\le a_n)=O(n^{-\delta}), \qquad n \to \infty. 
 $$
\begin{proof}
We assume, in all that follows, that $n$ is sufficiently large that $c_n>1$. For $k\in[n]$, let
$$
I_k=\inf\{t\in\posint : S_{k}(t) =c_n\}
$$
be the first time that $c_n$ copies of coupon type $D_k$ have been collected. For any sequence of natural numbers $(b_n)_{n\in\posint}$, we have
\begin{equation}
\begin{split}
\PP(|R_{I_k}| \le a_n) &= \PP(|R_{I_k}|\le a_n, I_k\le b_n) + \PP(|R_{I_k}|\le a_n, I_k > b_n)\\
&\le \PP(I_k\le b_n) + \PP(|R_{I_k}|\le a_n, I_k > b_n)\\
&\le \PP(I_k\le b_n) + \PP(\coupon_{a_n+1}>b_n),
\end{split}
\label{eq:partitioning R_I_k bound}
\end{equation}
where the last inequality follows by observing that if $|R_{I_k}|\le a_n$ and $I_k>b_n$, then $\coupon_{a_n+1}>b_n$.

To find an upper bound for $\PP(I_k\le b_n)$, note that, for any $s\ge1$, the random time between the $s$th and $(s+1)$th arrival of coupon
type $D_k$ is a geometric random variable with success probability $1/n$. It follows that {$\Delta_k:=I_k-\coupon_k$} is a sum of $c_n-1$
independent geometric random variables\footnote{Since the time $\coupon_k$ of the first arrival of $D_k$ is not geometrically distributed, $I_k$
  is not itself a sum of geometric random variables.}, each with success probability $1/n$.  Lemma~\ref{lem:tail bounds for geometric
  variables} therefore implies that for any $0<\lambda<1$,
$$
\PP(\Delta_k\le \lambda n(c_n-1)) \le e^{-f(\lambda) c_n + f(\lambda)}
$$
where $f(\lambda)>0$. But from the trivial lower bound $\coupon_k\ge1$, it follows that $\Delta_k\le I_k -1$. Therefore, for any
$b_n\le\lambda n(c_n-1)+1$, we have
\begin{equation}
\label{I_k tail bound}
\PP(I_k \le b_n) \le \PP(I_k\le\lambda n (c_n-1) + 1) \le \PP(\Delta_k\le \lambda n (c_n-1)) \le e^{-f(\lambda) c_n + f(\lambda)}.
\end{equation}

To find an upper bound for $\PP(\coupon_{a_n+1}>b_n)$, we begin with the observation that,
with the convention $\coupon_0=0$, we have
$$
\coupon_{a_n+1} = \sum_{i=0}^{a_n}(\coupon_{i+1}-\coupon_i).
$$
For $0\le i \le a_n$, the random variables $\coupon_{i+1}-\coupon_i$ are
independent, and distributed according to a geometric distribution with success probability $1-i/n$. 
Therefore, Lemma~\ref{lem:tail bounds for geometric variables} implies that for any $\zeta>1$
$$
\PP(\coupon_{a_n+1}\ge \zeta \,\EE[\coupon_{a_n+1}]) \le e^{-f(\zeta)(1-a_n/n)\EE(\coupon_{a_n+1})}
$$
with $f(\zeta)>0$. But explicit calculation shows that
$$
\EE(\coupon_{a_n+1}) = n(H_n - H_{n-a_n-1})\sim a_n, \qquad n\to\infty,
$$
where $H_i$ is the $i$th harmonic number, and the asymptotic result follows from $H_n\sim\ln(n)$ and the fact that $a_n=o(n)$.
It follows that for any choice of $b_n \ge \zeta\, \EE(\coupon_{a_n+1})$ and $\alpha\in(0,f(\zeta))$, for sufficiently large $n$, we have
\begin{equation}
\label{coupon_a_n tail bound}
\PP(\coupon_{a_n+1}\ge b_n) \le e^{-\alpha\, a_n}.
\end{equation}

Any choice of $b_n$ satisfying $\zeta\,\EE(\coupon_{a_n+1}) \le b_n \le \lambda n(c_n-1)+1$, for sufficiently large $n$, suffices to
ensure~\eqref{I_k tail bound} and~\eqref{coupon_a_n tail bound} hold simultaneously. It therefore suffices to set $b_n=n$. For simplicity,
$\lambda\in(0,1)$ and $\zeta>1$ can be chosen so that $f(\lambda)=1=f(\zeta)$. Combining~\eqref{eq:partitioning R_I_k bound}, \eqref{I_k
  tail bound} and~\eqref{coupon_a_n tail bound} then implies that for any $\alpha<1$ we have
$$
\PP(|R_{I_k}| \le a_n) \le e^{-c_n+1} + e^{-\alpha\,a_n}
$$
for sufficiently large $n$.

Finally, since $|R_I|\le a_n$ implies $|R_{I_k}|\le a_n$ for some $1\le k \le a_n$, it follows from the union bound that, for sufficiently
large $n$,
\begin{align*}
 \PP(|R_I|\le a_n) \le \PP\left(\bigcup_{k=1}^{a_n}\{|R_{I_k}|\le a_n\}\right) \le \sum_{k=1}^{a_n} \PP(|R_{I_k}|\le a_n)
 &\le a_n\, e^{-c_n+1} + a_n\, e^{-\alpha\,a_n}\\
 &\le e^2\,\ln(n)\,n^{-c} + e^{\alpha}\,\ln(n)\,n^{-\alpha}.
\end{align*}
Since $c,\alpha>0$, we can choose $0< \delta <\min\{c,\alpha\}$, and we obtain $\PP(|R_I|\le a_n) = O(n^{-\delta})$.
\end{proof}
\end{lemma}

\begin{lemma}
\label{lem:tail bounds for geometric variables}
Let $X_1,X_2,\ldots,X_n$ be independent random variables, such that $X_i$ has geometric distribution with success probability $p_i$, and let 
$X=\sum_{i=1}^n X_i$. Then 
\begin{align*}
\PP(X\le \lambda\mu) &\le e^{-p_{\ast}\mu f(\lambda)}, \qquad \forall\,\lambda\le 1,\\
\PP(X\ge \zeta\mu) &\le e^{-p_{\ast}\mu f(\zeta)}, \qquad \forall\,\zeta\ge 1,
\end{align*}
where $\mu=\EE(X) = \sum_{i=1}^n 1/p_i$, $p_\ast = \min_{i\in[n]} p_i$ and $f(x)=x -1 -\ln(x)$.
\begin{proof}
These results can be established, in the standard way, by applying Markov's inequality to $\EE(e^{t X})$, and using the 
explicit form for $\EE(e^{t X_i})$; see e.g.~\cite{Janson14}.
\end{proof}
\end{lemma}

\bibliographystyle{spmpsci}      

\end{document}

%% file: gumbel_off_q8_2d.pgf
\begingroup%
\makeatletter%
\begin{pgfpicture}%
\pgfpathrectangle{\pgfpointorigin}{\pgfqpoint{2.923228in}{1.806655in}}%
\pgfusepath{use as bounding box}%
\begin{pgfscope}%
\pgfsetbuttcap%
\pgfsetroundjoin%
\definecolor{currentfill}{rgb}{1.000000,1.000000,1.000000}%
\pgfsetfillcolor{currentfill}%
\pgfsetlinewidth{0.000000pt}%
\definecolor{currentstroke}{rgb}{1.000000,1.000000,1.000000}%
\pgfsetstrokecolor{currentstroke}%
\pgfsetdash{}{0pt}%
\pgfpathmoveto{\pgfqpoint{0.000000in}{0.000000in}}%
\pgfpathlineto{\pgfqpoint{2.923228in}{0.000000in}}%
\pgfpathlineto{\pgfqpoint{2.923228in}{1.806655in}}%
\pgfpathlineto{\pgfqpoint{0.000000in}{1.806655in}}%
\pgfpathclose%
\pgfusepath{fill}%
\end{pgfscope}%
\begin{pgfscope}%
\pgfsetbuttcap%
\pgfsetroundjoin%
\definecolor{currentfill}{rgb}{1.000000,1.000000,1.000000}%
\pgfsetfillcolor{currentfill}%
\pgfsetlinewidth{0.000000pt}%
\definecolor{currentstroke}{rgb}{0.000000,0.000000,0.000000}%
\pgfsetstrokecolor{currentstroke}%
\pgfsetstrokeopacity{0.000000}%
\pgfsetdash{}{0pt}%
\pgfpathmoveto{\pgfqpoint{0.365404in}{0.180665in}}%
\pgfpathlineto{\pgfqpoint{2.630906in}{0.180665in}}%
\pgfpathlineto{\pgfqpoint{2.630906in}{1.625989in}}%
\pgfpathlineto{\pgfqpoint{0.365404in}{1.625989in}}%
\pgfpathclose%
\pgfusepath{fill}%
\end{pgfscope}%
\begin{pgfscope}%
\pgfpathrectangle{\pgfqpoint{0.365404in}{0.180665in}}{\pgfqpoint{2.265502in}{1.445324in}} %
\pgfusepath{clip}%
\pgfsetrectcap%
\pgfsetroundjoin%
\pgfsetlinewidth{0.501875pt}%
\definecolor{currentstroke}{rgb}{0.000000,0.500000,0.000000}%
\pgfsetstrokecolor{currentstroke}%
\pgfsetdash{}{0pt}%
\pgfpathmoveto{\pgfqpoint{0.365402in}{0.284527in}}%
\pgfpathlineto{\pgfqpoint{0.399712in}{0.286652in}}%
\pgfpathlineto{\pgfqpoint{0.422035in}{0.290183in}}%
\pgfpathlineto{\pgfqpoint{0.439719in}{0.295182in}}%
\pgfpathlineto{\pgfqpoint{0.454999in}{0.301789in}}%
\pgfpathlineto{\pgfqpoint{0.468932in}{0.310240in}}%
\pgfpathlineto{\pgfqpoint{0.482146in}{0.320882in}}%
\pgfpathlineto{\pgfqpoint{0.495058in}{0.334186in}}%
\pgfpathlineto{\pgfqpoint{0.507981in}{0.350766in}}%
\pgfpathlineto{\pgfqpoint{0.521170in}{0.371402in}}%
\pgfpathlineto{\pgfqpoint{0.534856in}{0.397070in}}%
\pgfpathlineto{\pgfqpoint{0.549283in}{0.429009in}}%
\pgfpathlineto{\pgfqpoint{0.564755in}{0.468847in}}%
\pgfpathlineto{\pgfqpoint{0.581712in}{0.518886in}}%
\pgfpathlineto{\pgfqpoint{0.600945in}{0.582870in}}%
\pgfpathlineto{\pgfqpoint{0.624316in}{0.668805in}}%
\pgfpathlineto{\pgfqpoint{0.662244in}{0.818241in}}%
\pgfpathlineto{\pgfqpoint{0.699574in}{0.962080in}}%
\pgfpathlineto{\pgfqpoint{0.722418in}{1.041430in}}%
\pgfpathlineto{\pgfqpoint{0.741422in}{1.099736in}}%
\pgfpathlineto{\pgfqpoint{0.758159in}{1.144192in}}%
\pgfpathlineto{\pgfqpoint{0.773261in}{1.178245in}}%
\pgfpathlineto{\pgfqpoint{0.787065in}{1.204101in}}%
\pgfpathlineto{\pgfqpoint{0.799782in}{1.223385in}}%
\pgfpathlineto{\pgfqpoint{0.811572in}{1.237384in}}%
\pgfpathlineto{\pgfqpoint{0.822585in}{1.247133in}}%
\pgfpathlineto{\pgfqpoint{0.832978in}{1.253456in}}%
\pgfpathlineto{\pgfqpoint{0.842946in}{1.256976in}}%
\pgfpathlineto{\pgfqpoint{0.852738in}{1.258100in}}%
\pgfpathlineto{\pgfqpoint{0.862630in}{1.256988in}}%
\pgfpathlineto{\pgfqpoint{0.872919in}{1.253558in}}%
\pgfpathlineto{\pgfqpoint{0.883898in}{1.247492in}}%
\pgfpathlineto{\pgfqpoint{0.895855in}{1.238264in}}%
\pgfpathlineto{\pgfqpoint{0.909084in}{1.225145in}}%
\pgfpathlineto{\pgfqpoint{0.923932in}{1.207169in}}%
\pgfpathlineto{\pgfqpoint{0.940864in}{1.183037in}}%
\pgfpathlineto{\pgfqpoint{0.960598in}{1.150869in}}%
\pgfpathlineto{\pgfqpoint{0.984497in}{1.107444in}}%
\pgfpathlineto{\pgfqpoint{1.016063in}{1.045131in}}%
\pgfpathlineto{\pgfqpoint{1.145560in}{0.784401in}}%
\pgfpathlineto{\pgfqpoint{1.178974in}{0.724888in}}%
\pgfpathlineto{\pgfqpoint{1.210027in}{0.674035in}}%
\pgfpathlineto{\pgfqpoint{1.239764in}{0.629536in}}%
\pgfpathlineto{\pgfqpoint{1.268744in}{0.590114in}}%
\pgfpathlineto{\pgfqpoint{1.297338in}{0.554928in}}%
\pgfpathlineto{\pgfqpoint{1.325824in}{0.523363in}}%
\pgfpathlineto{\pgfqpoint{1.354436in}{0.494950in}}%
\pgfpathlineto{\pgfqpoint{1.383398in}{0.469301in}}%
\pgfpathlineto{\pgfqpoint{1.412932in}{0.446103in}}%
\pgfpathlineto{\pgfqpoint{1.443273in}{0.425090in}}%
\pgfpathlineto{\pgfqpoint{1.474677in}{0.406046in}}%
\pgfpathlineto{\pgfqpoint{1.507434in}{0.388786in}}%
\pgfpathlineto{\pgfqpoint{1.541884in}{0.373158in}}%
\pgfpathlineto{\pgfqpoint{1.578436in}{0.359033in}}%
\pgfpathlineto{\pgfqpoint{1.617580in}{0.346312in}}%
\pgfpathlineto{\pgfqpoint{1.659945in}{0.334909in}}%
\pgfpathlineto{\pgfqpoint{1.706325in}{0.324761in}}%
\pgfpathlineto{\pgfqpoint{1.757796in}{0.315815in}}%
\pgfpathlineto{\pgfqpoint{1.815827in}{0.308033in}}%
\pgfpathlineto{\pgfqpoint{1.882562in}{0.301384in}}%
\pgfpathlineto{\pgfqpoint{1.961279in}{0.295844in}}%
\pgfpathlineto{\pgfqpoint{2.057411in}{0.291393in}}%
\pgfpathlineto{\pgfqpoint{2.180934in}{0.288013in}}%
\pgfpathlineto{\pgfqpoint{2.353317in}{0.285680in}}%
\pgfpathlineto{\pgfqpoint{2.630907in}{0.284363in}}%
\pgfpathlineto{\pgfqpoint{2.630907in}{0.284363in}}%
\pgfusepath{stroke}%
\end{pgfscope}%
\begin{pgfscope}%
\pgfsetbuttcap%
\pgfsetroundjoin%
\definecolor{currentfill}{rgb}{0.000000,0.000000,0.000000}%
\pgfsetfillcolor{currentfill}%
\pgfsetlinewidth{0.501875pt}%
\definecolor{currentstroke}{rgb}{0.000000,0.000000,0.000000}%
\pgfsetstrokecolor{currentstroke}%
\pgfsetdash{}{0pt}%
\pgfsys@defobject{currentmarker}{\pgfqpoint{0.000000in}{0.000000in}}{\pgfqpoint{0.000000in}{0.055556in}}{%
\pgfpathmoveto{\pgfqpoint{0.000000in}{0.000000in}}%
\pgfpathlineto{\pgfqpoint{0.000000in}{0.055556in}}%
\pgfusepath{stroke,fill}%
}%
\begin{pgfscope}%
\pgfsys@transformshift{0.444433in}{0.180665in}%
\pgfsys@useobject{currentmarker}{}%
\end{pgfscope}%
\end{pgfscope}%
\begin{pgfscope}%
\pgfsetbuttcap%
\pgfsetroundjoin%
\definecolor{currentfill}{rgb}{0.000000,0.000000,0.000000}%
\pgfsetfillcolor{currentfill}%
\pgfsetlinewidth{0.501875pt}%
\definecolor{currentstroke}{rgb}{0.000000,0.000000,0.000000}%
\pgfsetstrokecolor{currentstroke}%
\pgfsetdash{}{0pt}%
\pgfsys@defobject{currentmarker}{\pgfqpoint{0.000000in}{-0.055556in}}{\pgfqpoint{0.000000in}{0.000000in}}{%
\pgfpathmoveto{\pgfqpoint{0.000000in}{0.000000in}}%
\pgfpathlineto{\pgfqpoint{0.000000in}{-0.055556in}}%
\pgfusepath{stroke,fill}%
}%
\begin{pgfscope}%
\pgfsys@transformshift{0.444433in}{1.625989in}%
\pgfsys@useobject{currentmarker}{}%
\end{pgfscope}%
\end{pgfscope}%
\begin{pgfscope}%
\pgftext[x=0.444433in,y=0.125110in,,top]{{\rmfamily\fontsize{8.000000}{9.600000}\selectfont \(\displaystyle -2\)}}%
\end{pgfscope}%
\begin{pgfscope}%
\pgfsetbuttcap%
\pgfsetroundjoin%
\definecolor{currentfill}{rgb}{0.000000,0.000000,0.000000}%
\pgfsetfillcolor{currentfill}%
\pgfsetlinewidth{0.501875pt}%
\definecolor{currentstroke}{rgb}{0.000000,0.000000,0.000000}%
\pgfsetstrokecolor{currentstroke}%
\pgfsetdash{}{0pt}%
\pgfsys@defobject{currentmarker}{\pgfqpoint{0.000000in}{0.000000in}}{\pgfqpoint{0.000000in}{0.055556in}}{%
\pgfpathmoveto{\pgfqpoint{0.000000in}{0.000000in}}%
\pgfpathlineto{\pgfqpoint{0.000000in}{0.055556in}}%
\pgfusepath{stroke,fill}%
}%
\begin{pgfscope}%
\pgfsys@transformshift{0.707863in}{0.180665in}%
\pgfsys@useobject{currentmarker}{}%
\end{pgfscope}%
\end{pgfscope}%
\begin{pgfscope}%
\pgfsetbuttcap%
\pgfsetroundjoin%
\definecolor{currentfill}{rgb}{0.000000,0.000000,0.000000}%
\pgfsetfillcolor{currentfill}%
\pgfsetlinewidth{0.501875pt}%
\definecolor{currentstroke}{rgb}{0.000000,0.000000,0.000000}%
\pgfsetstrokecolor{currentstroke}%
\pgfsetdash{}{0pt}%
\pgfsys@defobject{currentmarker}{\pgfqpoint{0.000000in}{-0.055556in}}{\pgfqpoint{0.000000in}{0.000000in}}{%
\pgfpathmoveto{\pgfqpoint{0.000000in}{0.000000in}}%
\pgfpathlineto{\pgfqpoint{0.000000in}{-0.055556in}}%
\pgfusepath{stroke,fill}%
}%
\begin{pgfscope}%
\pgfsys@transformshift{0.707863in}{1.625989in}%
\pgfsys@useobject{currentmarker}{}%
\end{pgfscope}%
\end{pgfscope}%
\begin{pgfscope}%
\pgftext[x=0.707863in,y=0.125110in,,top]{{\rmfamily\fontsize{8.000000}{9.600000}\selectfont \(\displaystyle -1\)}}%
\end{pgfscope}%
\begin{pgfscope}%
\pgfsetbuttcap%
\pgfsetroundjoin%
\definecolor{currentfill}{rgb}{0.000000,0.000000,0.000000}%
\pgfsetfillcolor{currentfill}%
\pgfsetlinewidth{0.501875pt}%
\definecolor{currentstroke}{rgb}{0.000000,0.000000,0.000000}%
\pgfsetstrokecolor{currentstroke}%
\pgfsetdash{}{0pt}%
\pgfsys@defobject{currentmarker}{\pgfqpoint{0.000000in}{0.000000in}}{\pgfqpoint{0.000000in}{0.055556in}}{%
\pgfpathmoveto{\pgfqpoint{0.000000in}{0.000000in}}%
\pgfpathlineto{\pgfqpoint{0.000000in}{0.055556in}}%
\pgfusepath{stroke,fill}%
}%
\begin{pgfscope}%
\pgfsys@transformshift{0.971294in}{0.180665in}%
\pgfsys@useobject{currentmarker}{}%
\end{pgfscope}%
\end{pgfscope}%
\begin{pgfscope}%
\pgfsetbuttcap%
\pgfsetroundjoin%
\definecolor{currentfill}{rgb}{0.000000,0.000000,0.000000}%
\pgfsetfillcolor{currentfill}%
\pgfsetlinewidth{0.501875pt}%
\definecolor{currentstroke}{rgb}{0.000000,0.000000,0.000000}%
\pgfsetstrokecolor{currentstroke}%
\pgfsetdash{}{0pt}%
\pgfsys@defobject{currentmarker}{\pgfqpoint{0.000000in}{-0.055556in}}{\pgfqpoint{0.000000in}{0.000000in}}{%
\pgfpathmoveto{\pgfqpoint{0.000000in}{0.000000in}}%
\pgfpathlineto{\pgfqpoint{0.000000in}{-0.055556in}}%
\pgfusepath{stroke,fill}%
}%
\begin{pgfscope}%
\pgfsys@transformshift{0.971294in}{1.625989in}%
\pgfsys@useobject{currentmarker}{}%
\end{pgfscope}%
\end{pgfscope}%
\begin{pgfscope}%
\pgftext[x=0.971294in,y=0.125110in,,top]{{\rmfamily\fontsize{8.000000}{9.600000}\selectfont \(\displaystyle 0\)}}%
\end{pgfscope}%
\begin{pgfscope}%
\pgfsetbuttcap%
\pgfsetroundjoin%
\definecolor{currentfill}{rgb}{0.000000,0.000000,0.000000}%
\pgfsetfillcolor{currentfill}%
\pgfsetlinewidth{0.501875pt}%
\definecolor{currentstroke}{rgb}{0.000000,0.000000,0.000000}%
\pgfsetstrokecolor{currentstroke}%
\pgfsetdash{}{0pt}%
\pgfsys@defobject{currentmarker}{\pgfqpoint{0.000000in}{0.000000in}}{\pgfqpoint{0.000000in}{0.055556in}}{%
\pgfpathmoveto{\pgfqpoint{0.000000in}{0.000000in}}%
\pgfpathlineto{\pgfqpoint{0.000000in}{0.055556in}}%
\pgfusepath{stroke,fill}%
}%
\begin{pgfscope}%
\pgfsys@transformshift{1.234724in}{0.180665in}%
\pgfsys@useobject{currentmarker}{}%
\end{pgfscope}%
\end{pgfscope}%
\begin{pgfscope}%
\pgfsetbuttcap%
\pgfsetroundjoin%
\definecolor{currentfill}{rgb}{0.000000,0.000000,0.000000}%
\pgfsetfillcolor{currentfill}%
\pgfsetlinewidth{0.501875pt}%
\definecolor{currentstroke}{rgb}{0.000000,0.000000,0.000000}%
\pgfsetstrokecolor{currentstroke}%
\pgfsetdash{}{0pt}%
\pgfsys@defobject{currentmarker}{\pgfqpoint{0.000000in}{-0.055556in}}{\pgfqpoint{0.000000in}{0.000000in}}{%
\pgfpathmoveto{\pgfqpoint{0.000000in}{0.000000in}}%
\pgfpathlineto{\pgfqpoint{0.000000in}{-0.055556in}}%
\pgfusepath{stroke,fill}%
}%
\begin{pgfscope}%
\pgfsys@transformshift{1.234724in}{1.625989in}%
\pgfsys@useobject{currentmarker}{}%
\end{pgfscope}%
\end{pgfscope}%
\begin{pgfscope}%
\pgftext[x=1.234724in,y=0.125110in,,top]{{\rmfamily\fontsize{8.000000}{9.600000}\selectfont \(\displaystyle 1\)}}%
\end{pgfscope}%
\begin{pgfscope}%
\pgfsetbuttcap%
\pgfsetroundjoin%
\definecolor{currentfill}{rgb}{0.000000,0.000000,0.000000}%
\pgfsetfillcolor{currentfill}%
\pgfsetlinewidth{0.501875pt}%
\definecolor{currentstroke}{rgb}{0.000000,0.000000,0.000000}%
\pgfsetstrokecolor{currentstroke}%
\pgfsetdash{}{0pt}%
\pgfsys@defobject{currentmarker}{\pgfqpoint{0.000000in}{0.000000in}}{\pgfqpoint{0.000000in}{0.055556in}}{%
\pgfpathmoveto{\pgfqpoint{0.000000in}{0.000000in}}%
\pgfpathlineto{\pgfqpoint{0.000000in}{0.055556in}}%
\pgfusepath{stroke,fill}%
}%
\begin{pgfscope}%
\pgfsys@transformshift{1.498155in}{0.180665in}%
\pgfsys@useobject{currentmarker}{}%
\end{pgfscope}%
\end{pgfscope}%
\begin{pgfscope}%
\pgfsetbuttcap%
\pgfsetroundjoin%
\definecolor{currentfill}{rgb}{0.000000,0.000000,0.000000}%
\pgfsetfillcolor{currentfill}%
\pgfsetlinewidth{0.501875pt}%
\definecolor{currentstroke}{rgb}{0.000000,0.000000,0.000000}%
\pgfsetstrokecolor{currentstroke}%
\pgfsetdash{}{0pt}%
\pgfsys@defobject{currentmarker}{\pgfqpoint{0.000000in}{-0.055556in}}{\pgfqpoint{0.000000in}{0.000000in}}{%
\pgfpathmoveto{\pgfqpoint{0.000000in}{0.000000in}}%
\pgfpathlineto{\pgfqpoint{0.000000in}{-0.055556in}}%
\pgfusepath{stroke,fill}%
}%
\begin{pgfscope}%
\pgfsys@transformshift{1.498155in}{1.625989in}%
\pgfsys@useobject{currentmarker}{}%
\end{pgfscope}%
\end{pgfscope}%
\begin{pgfscope}%
\pgftext[x=1.498155in,y=0.125110in,,top]{{\rmfamily\fontsize{8.000000}{9.600000}\selectfont \(\displaystyle 2\)}}%
\end{pgfscope}%
\begin{pgfscope}%
\pgfsetbuttcap%
\pgfsetroundjoin%
\definecolor{currentfill}{rgb}{0.000000,0.000000,0.000000}%
\pgfsetfillcolor{currentfill}%
\pgfsetlinewidth{0.501875pt}%
\definecolor{currentstroke}{rgb}{0.000000,0.000000,0.000000}%
\pgfsetstrokecolor{currentstroke}%
\pgfsetdash{}{0pt}%
\pgfsys@defobject{currentmarker}{\pgfqpoint{0.000000in}{0.000000in}}{\pgfqpoint{0.000000in}{0.055556in}}{%
\pgfpathmoveto{\pgfqpoint{0.000000in}{0.000000in}}%
\pgfpathlineto{\pgfqpoint{0.000000in}{0.055556in}}%
\pgfusepath{stroke,fill}%
}%
\begin{pgfscope}%
\pgfsys@transformshift{1.761585in}{0.180665in}%
\pgfsys@useobject{currentmarker}{}%
\end{pgfscope}%
\end{pgfscope}%
\begin{pgfscope}%
\pgfsetbuttcap%
\pgfsetroundjoin%
\definecolor{currentfill}{rgb}{0.000000,0.000000,0.000000}%
\pgfsetfillcolor{currentfill}%
\pgfsetlinewidth{0.501875pt}%
\definecolor{currentstroke}{rgb}{0.000000,0.000000,0.000000}%
\pgfsetstrokecolor{currentstroke}%
\pgfsetdash{}{0pt}%
\pgfsys@defobject{currentmarker}{\pgfqpoint{0.000000in}{-0.055556in}}{\pgfqpoint{0.000000in}{0.000000in}}{%
\pgfpathmoveto{\pgfqpoint{0.000000in}{0.000000in}}%
\pgfpathlineto{\pgfqpoint{0.000000in}{-0.055556in}}%
\pgfusepath{stroke,fill}%
}%
\begin{pgfscope}%
\pgfsys@transformshift{1.761585in}{1.625989in}%
\pgfsys@useobject{currentmarker}{}%
\end{pgfscope}%
\end{pgfscope}%
\begin{pgfscope}%
\pgftext[x=1.761585in,y=0.125110in,,top]{{\rmfamily\fontsize{8.000000}{9.600000}\selectfont \(\displaystyle 3\)}}%
\end{pgfscope}%
\begin{pgfscope}%
\pgfsetbuttcap%
\pgfsetroundjoin%
\definecolor{currentfill}{rgb}{0.000000,0.000000,0.000000}%
\pgfsetfillcolor{currentfill}%
\pgfsetlinewidth{0.501875pt}%
\definecolor{currentstroke}{rgb}{0.000000,0.000000,0.000000}%
\pgfsetstrokecolor{currentstroke}%
\pgfsetdash{}{0pt}%
\pgfsys@defobject{currentmarker}{\pgfqpoint{0.000000in}{0.000000in}}{\pgfqpoint{0.000000in}{0.055556in}}{%
\pgfpathmoveto{\pgfqpoint{0.000000in}{0.000000in}}%
\pgfpathlineto{\pgfqpoint{0.000000in}{0.055556in}}%
\pgfusepath{stroke,fill}%
}%
\begin{pgfscope}%
\pgfsys@transformshift{2.025016in}{0.180665in}%
\pgfsys@useobject{currentmarker}{}%
\end{pgfscope}%
\end{pgfscope}%
\begin{pgfscope}%
\pgfsetbuttcap%
\pgfsetroundjoin%
\definecolor{currentfill}{rgb}{0.000000,0.000000,0.000000}%
\pgfsetfillcolor{currentfill}%
\pgfsetlinewidth{0.501875pt}%
\definecolor{currentstroke}{rgb}{0.000000,0.000000,0.000000}%
\pgfsetstrokecolor{currentstroke}%
\pgfsetdash{}{0pt}%
\pgfsys@defobject{currentmarker}{\pgfqpoint{0.000000in}{-0.055556in}}{\pgfqpoint{0.000000in}{0.000000in}}{%
\pgfpathmoveto{\pgfqpoint{0.000000in}{0.000000in}}%
\pgfpathlineto{\pgfqpoint{0.000000in}{-0.055556in}}%
\pgfusepath{stroke,fill}%
}%
\begin{pgfscope}%
\pgfsys@transformshift{2.025016in}{1.625989in}%
\pgfsys@useobject{currentmarker}{}%
\end{pgfscope}%
\end{pgfscope}%
\begin{pgfscope}%
\pgftext[x=2.025016in,y=0.125110in,,top]{{\rmfamily\fontsize{8.000000}{9.600000}\selectfont \(\displaystyle 4\)}}%
\end{pgfscope}%
\begin{pgfscope}%
\pgfsetbuttcap%
\pgfsetroundjoin%
\definecolor{currentfill}{rgb}{0.000000,0.000000,0.000000}%
\pgfsetfillcolor{currentfill}%
\pgfsetlinewidth{0.501875pt}%
\definecolor{currentstroke}{rgb}{0.000000,0.000000,0.000000}%
\pgfsetstrokecolor{currentstroke}%
\pgfsetdash{}{0pt}%
\pgfsys@defobject{currentmarker}{\pgfqpoint{0.000000in}{0.000000in}}{\pgfqpoint{0.000000in}{0.055556in}}{%
\pgfpathmoveto{\pgfqpoint{0.000000in}{0.000000in}}%
\pgfpathlineto{\pgfqpoint{0.000000in}{0.055556in}}%
\pgfusepath{stroke,fill}%
}%
\begin{pgfscope}%
\pgfsys@transformshift{2.288446in}{0.180665in}%
\pgfsys@useobject{currentmarker}{}%
\end{pgfscope}%
\end{pgfscope}%
\begin{pgfscope}%
\pgfsetbuttcap%
\pgfsetroundjoin%
\definecolor{currentfill}{rgb}{0.000000,0.000000,0.000000}%
\pgfsetfillcolor{currentfill}%
\pgfsetlinewidth{0.501875pt}%
\definecolor{currentstroke}{rgb}{0.000000,0.000000,0.000000}%
\pgfsetstrokecolor{currentstroke}%
\pgfsetdash{}{0pt}%
\pgfsys@defobject{currentmarker}{\pgfqpoint{0.000000in}{-0.055556in}}{\pgfqpoint{0.000000in}{0.000000in}}{%
\pgfpathmoveto{\pgfqpoint{0.000000in}{0.000000in}}%
\pgfpathlineto{\pgfqpoint{0.000000in}{-0.055556in}}%
\pgfusepath{stroke,fill}%
}%
\begin{pgfscope}%
\pgfsys@transformshift{2.288446in}{1.625989in}%
\pgfsys@useobject{currentmarker}{}%
\end{pgfscope}%
\end{pgfscope}%
\begin{pgfscope}%
\pgftext[x=2.288446in,y=0.125110in,,top]{{\rmfamily\fontsize{8.000000}{9.600000}\selectfont \(\displaystyle 5\)}}%
\end{pgfscope}%
\begin{pgfscope}%
\pgfsetbuttcap%
\pgfsetroundjoin%
\definecolor{currentfill}{rgb}{0.000000,0.000000,0.000000}%
\pgfsetfillcolor{currentfill}%
\pgfsetlinewidth{0.501875pt}%
\definecolor{currentstroke}{rgb}{0.000000,0.000000,0.000000}%
\pgfsetstrokecolor{currentstroke}%
\pgfsetdash{}{0pt}%
\pgfsys@defobject{currentmarker}{\pgfqpoint{0.000000in}{0.000000in}}{\pgfqpoint{0.000000in}{0.055556in}}{%
\pgfpathmoveto{\pgfqpoint{0.000000in}{0.000000in}}%
\pgfpathlineto{\pgfqpoint{0.000000in}{0.055556in}}%
\pgfusepath{stroke,fill}%
}%
\begin{pgfscope}%
\pgfsys@transformshift{2.551876in}{0.180665in}%
\pgfsys@useobject{currentmarker}{}%
\end{pgfscope}%
\end{pgfscope}%
\begin{pgfscope}%
\pgfsetbuttcap%
\pgfsetroundjoin%
\definecolor{currentfill}{rgb}{0.000000,0.000000,0.000000}%
\pgfsetfillcolor{currentfill}%
\pgfsetlinewidth{0.501875pt}%
\definecolor{currentstroke}{rgb}{0.000000,0.000000,0.000000}%
\pgfsetstrokecolor{currentstroke}%
\pgfsetdash{}{0pt}%
\pgfsys@defobject{currentmarker}{\pgfqpoint{0.000000in}{-0.055556in}}{\pgfqpoint{0.000000in}{0.000000in}}{%
\pgfpathmoveto{\pgfqpoint{0.000000in}{0.000000in}}%
\pgfpathlineto{\pgfqpoint{0.000000in}{-0.055556in}}%
\pgfusepath{stroke,fill}%
}%
\begin{pgfscope}%
\pgfsys@transformshift{2.551876in}{1.625989in}%
\pgfsys@useobject{currentmarker}{}%
\end{pgfscope}%
\end{pgfscope}%
\begin{pgfscope}%
\pgftext[x=2.551876in,y=0.125110in,,top]{{\rmfamily\fontsize{8.000000}{9.600000}\selectfont \(\displaystyle 6\)}}%
\end{pgfscope}%
\begin{pgfscope}%
\pgftext[x=1.498155in,y=-0.042459in,,top]{{\rmfamily\fontsize{10.000000}{12.000000}\selectfont \(\displaystyle s\)}}%
\end{pgfscope}%
\begin{pgfscope}%
\pgfsetbuttcap%
\pgfsetroundjoin%
\definecolor{currentfill}{rgb}{0.000000,0.000000,0.000000}%
\pgfsetfillcolor{currentfill}%
\pgfsetlinewidth{0.501875pt}%
\definecolor{currentstroke}{rgb}{0.000000,0.000000,0.000000}%
\pgfsetstrokecolor{currentstroke}%
\pgfsetdash{}{0pt}%
\pgfsys@defobject{currentmarker}{\pgfqpoint{0.000000in}{0.000000in}}{\pgfqpoint{0.055556in}{0.000000in}}{%
\pgfpathmoveto{\pgfqpoint{0.000000in}{0.000000in}}%
\pgfpathlineto{\pgfqpoint{0.055556in}{0.000000in}}%
\pgfusepath{stroke,fill}%
}%
\begin{pgfscope}%
\pgfsys@transformshift{0.365404in}{0.283903in}%
\pgfsys@useobject{currentmarker}{}%
\end{pgfscope}%
\end{pgfscope}%
\begin{pgfscope}%
\pgfsetbuttcap%
\pgfsetroundjoin%
\definecolor{currentfill}{rgb}{0.000000,0.000000,0.000000}%
\pgfsetfillcolor{currentfill}%
\pgfsetlinewidth{0.501875pt}%
\definecolor{currentstroke}{rgb}{0.000000,0.000000,0.000000}%
\pgfsetstrokecolor{currentstroke}%
\pgfsetdash{}{0pt}%
\pgfsys@defobject{currentmarker}{\pgfqpoint{-0.055556in}{0.000000in}}{\pgfqpoint{0.000000in}{0.000000in}}{%
\pgfpathmoveto{\pgfqpoint{0.000000in}{0.000000in}}%
\pgfpathlineto{\pgfqpoint{-0.055556in}{0.000000in}}%
\pgfusepath{stroke,fill}%
}%
\begin{pgfscope}%
\pgfsys@transformshift{2.630906in}{0.283903in}%
\pgfsys@useobject{currentmarker}{}%
\end{pgfscope}%
\end{pgfscope}%
\begin{pgfscope}%
\pgftext[x=0.309848in,y=0.283903in,right,]{{\rmfamily\fontsize{8.000000}{9.600000}\selectfont \(\displaystyle 0.0\)}}%
\end{pgfscope}%
\begin{pgfscope}%
\pgfsetbuttcap%
\pgfsetroundjoin%
\definecolor{currentfill}{rgb}{0.000000,0.000000,0.000000}%
\pgfsetfillcolor{currentfill}%
\pgfsetlinewidth{0.501875pt}%
\definecolor{currentstroke}{rgb}{0.000000,0.000000,0.000000}%
\pgfsetstrokecolor{currentstroke}%
\pgfsetdash{}{0pt}%
\pgfsys@defobject{currentmarker}{\pgfqpoint{0.000000in}{0.000000in}}{\pgfqpoint{0.055556in}{0.000000in}}{%
\pgfpathmoveto{\pgfqpoint{0.000000in}{0.000000in}}%
\pgfpathlineto{\pgfqpoint{0.055556in}{0.000000in}}%
\pgfusepath{stroke,fill}%
}%
\begin{pgfscope}%
\pgfsys@transformshift{0.365404in}{0.490378in}%
\pgfsys@useobject{currentmarker}{}%
\end{pgfscope}%
\end{pgfscope}%
\begin{pgfscope}%
\pgfsetbuttcap%
\pgfsetroundjoin%
\definecolor{currentfill}{rgb}{0.000000,0.000000,0.000000}%
\pgfsetfillcolor{currentfill}%
\pgfsetlinewidth{0.501875pt}%
\definecolor{currentstroke}{rgb}{0.000000,0.000000,0.000000}%
\pgfsetstrokecolor{currentstroke}%
\pgfsetdash{}{0pt}%
\pgfsys@defobject{currentmarker}{\pgfqpoint{-0.055556in}{0.000000in}}{\pgfqpoint{0.000000in}{0.000000in}}{%
\pgfpathmoveto{\pgfqpoint{0.000000in}{0.000000in}}%
\pgfpathlineto{\pgfqpoint{-0.055556in}{0.000000in}}%
\pgfusepath{stroke,fill}%
}%
\begin{pgfscope}%
\pgfsys@transformshift{2.630906in}{0.490378in}%
\pgfsys@useobject{currentmarker}{}%
\end{pgfscope}%
\end{pgfscope}%
\begin{pgfscope}%
\pgftext[x=0.309848in,y=0.490378in,right,]{{\rmfamily\fontsize{8.000000}{9.600000}\selectfont \(\displaystyle 0.1\)}}%
\end{pgfscope}%
\begin{pgfscope}%
\pgfsetbuttcap%
\pgfsetroundjoin%
\definecolor{currentfill}{rgb}{0.000000,0.000000,0.000000}%
\pgfsetfillcolor{currentfill}%
\pgfsetlinewidth{0.501875pt}%
\definecolor{currentstroke}{rgb}{0.000000,0.000000,0.000000}%
\pgfsetstrokecolor{currentstroke}%
\pgfsetdash{}{0pt}%
\pgfsys@defobject{currentmarker}{\pgfqpoint{0.000000in}{0.000000in}}{\pgfqpoint{0.055556in}{0.000000in}}{%
\pgfpathmoveto{\pgfqpoint{0.000000in}{0.000000in}}%
\pgfpathlineto{\pgfqpoint{0.055556in}{0.000000in}}%
\pgfusepath{stroke,fill}%
}%
\begin{pgfscope}%
\pgfsys@transformshift{0.365404in}{0.696852in}%
\pgfsys@useobject{currentmarker}{}%
\end{pgfscope}%
\end{pgfscope}%
\begin{pgfscope}%
\pgfsetbuttcap%
\pgfsetroundjoin%
\definecolor{currentfill}{rgb}{0.000000,0.000000,0.000000}%
\pgfsetfillcolor{currentfill}%
\pgfsetlinewidth{0.501875pt}%
\definecolor{currentstroke}{rgb}{0.000000,0.000000,0.000000}%
\pgfsetstrokecolor{currentstroke}%
\pgfsetdash{}{0pt}%
\pgfsys@defobject{currentmarker}{\pgfqpoint{-0.055556in}{0.000000in}}{\pgfqpoint{0.000000in}{0.000000in}}{%
\pgfpathmoveto{\pgfqpoint{0.000000in}{0.000000in}}%
\pgfpathlineto{\pgfqpoint{-0.055556in}{0.000000in}}%
\pgfusepath{stroke,fill}%
}%
\begin{pgfscope}%
\pgfsys@transformshift{2.630906in}{0.696852in}%
\pgfsys@useobject{currentmarker}{}%
\end{pgfscope}%
\end{pgfscope}%
\begin{pgfscope}%
\pgftext[x=0.309848in,y=0.696852in,right,]{{\rmfamily\fontsize{8.000000}{9.600000}\selectfont \(\displaystyle 0.2\)}}%
\end{pgfscope}%
\begin{pgfscope}%
\pgfsetbuttcap%
\pgfsetroundjoin%
\definecolor{currentfill}{rgb}{0.000000,0.000000,0.000000}%
\pgfsetfillcolor{currentfill}%
\pgfsetlinewidth{0.501875pt}%
\definecolor{currentstroke}{rgb}{0.000000,0.000000,0.000000}%
\pgfsetstrokecolor{currentstroke}%
\pgfsetdash{}{0pt}%
\pgfsys@defobject{currentmarker}{\pgfqpoint{0.000000in}{0.000000in}}{\pgfqpoint{0.055556in}{0.000000in}}{%
\pgfpathmoveto{\pgfqpoint{0.000000in}{0.000000in}}%
\pgfpathlineto{\pgfqpoint{0.055556in}{0.000000in}}%
\pgfusepath{stroke,fill}%
}%
\begin{pgfscope}%
\pgfsys@transformshift{0.365404in}{0.903327in}%
\pgfsys@useobject{currentmarker}{}%
\end{pgfscope}%
\end{pgfscope}%
\begin{pgfscope}%
\pgfsetbuttcap%
\pgfsetroundjoin%
\definecolor{currentfill}{rgb}{0.000000,0.000000,0.000000}%
\pgfsetfillcolor{currentfill}%
\pgfsetlinewidth{0.501875pt}%
\definecolor{currentstroke}{rgb}{0.000000,0.000000,0.000000}%
\pgfsetstrokecolor{currentstroke}%
\pgfsetdash{}{0pt}%
\pgfsys@defobject{currentmarker}{\pgfqpoint{-0.055556in}{0.000000in}}{\pgfqpoint{0.000000in}{0.000000in}}{%
\pgfpathmoveto{\pgfqpoint{0.000000in}{0.000000in}}%
\pgfpathlineto{\pgfqpoint{-0.055556in}{0.000000in}}%
\pgfusepath{stroke,fill}%
}%
\begin{pgfscope}%
\pgfsys@transformshift{2.630906in}{0.903327in}%
\pgfsys@useobject{currentmarker}{}%
\end{pgfscope}%
\end{pgfscope}%
\begin{pgfscope}%
\pgftext[x=0.309848in,y=0.903327in,right,]{{\rmfamily\fontsize{8.000000}{9.600000}\selectfont \(\displaystyle 0.3\)}}%
\end{pgfscope}%
\begin{pgfscope}%
\pgfsetbuttcap%
\pgfsetroundjoin%
\definecolor{currentfill}{rgb}{0.000000,0.000000,0.000000}%
\pgfsetfillcolor{currentfill}%
\pgfsetlinewidth{0.501875pt}%
\definecolor{currentstroke}{rgb}{0.000000,0.000000,0.000000}%
\pgfsetstrokecolor{currentstroke}%
\pgfsetdash{}{0pt}%
\pgfsys@defobject{currentmarker}{\pgfqpoint{0.000000in}{0.000000in}}{\pgfqpoint{0.055556in}{0.000000in}}{%
\pgfpathmoveto{\pgfqpoint{0.000000in}{0.000000in}}%
\pgfpathlineto{\pgfqpoint{0.055556in}{0.000000in}}%
\pgfusepath{stroke,fill}%
}%
\begin{pgfscope}%
\pgfsys@transformshift{0.365404in}{1.109802in}%
\pgfsys@useobject{currentmarker}{}%
\end{pgfscope}%
\end{pgfscope}%
\begin{pgfscope}%
\pgfsetbuttcap%
\pgfsetroundjoin%
\definecolor{currentfill}{rgb}{0.000000,0.000000,0.000000}%
\pgfsetfillcolor{currentfill}%
\pgfsetlinewidth{0.501875pt}%
\definecolor{currentstroke}{rgb}{0.000000,0.000000,0.000000}%
\pgfsetstrokecolor{currentstroke}%
\pgfsetdash{}{0pt}%
\pgfsys@defobject{currentmarker}{\pgfqpoint{-0.055556in}{0.000000in}}{\pgfqpoint{0.000000in}{0.000000in}}{%
\pgfpathmoveto{\pgfqpoint{0.000000in}{0.000000in}}%
\pgfpathlineto{\pgfqpoint{-0.055556in}{0.000000in}}%
\pgfusepath{stroke,fill}%
}%
\begin{pgfscope}%
\pgfsys@transformshift{2.630906in}{1.109802in}%
\pgfsys@useobject{currentmarker}{}%
\end{pgfscope}%
\end{pgfscope}%
\begin{pgfscope}%
\pgftext[x=0.309848in,y=1.109802in,right,]{{\rmfamily\fontsize{8.000000}{9.600000}\selectfont \(\displaystyle 0.4\)}}%
\end{pgfscope}%
\begin{pgfscope}%
\pgfsetbuttcap%
\pgfsetroundjoin%
\definecolor{currentfill}{rgb}{0.000000,0.000000,0.000000}%
\pgfsetfillcolor{currentfill}%
\pgfsetlinewidth{0.501875pt}%
\definecolor{currentstroke}{rgb}{0.000000,0.000000,0.000000}%
\pgfsetstrokecolor{currentstroke}%
\pgfsetdash{}{0pt}%
\pgfsys@defobject{currentmarker}{\pgfqpoint{0.000000in}{0.000000in}}{\pgfqpoint{0.055556in}{0.000000in}}{%
\pgfpathmoveto{\pgfqpoint{0.000000in}{0.000000in}}%
\pgfpathlineto{\pgfqpoint{0.055556in}{0.000000in}}%
\pgfusepath{stroke,fill}%
}%
\begin{pgfscope}%
\pgfsys@transformshift{0.365404in}{1.316277in}%
\pgfsys@useobject{currentmarker}{}%
\end{pgfscope}%
\end{pgfscope}%
\begin{pgfscope}%
\pgfsetbuttcap%
\pgfsetroundjoin%
\definecolor{currentfill}{rgb}{0.000000,0.000000,0.000000}%
\pgfsetfillcolor{currentfill}%
\pgfsetlinewidth{0.501875pt}%
\definecolor{currentstroke}{rgb}{0.000000,0.000000,0.000000}%
\pgfsetstrokecolor{currentstroke}%
\pgfsetdash{}{0pt}%
\pgfsys@defobject{currentmarker}{\pgfqpoint{-0.055556in}{0.000000in}}{\pgfqpoint{0.000000in}{0.000000in}}{%
\pgfpathmoveto{\pgfqpoint{0.000000in}{0.000000in}}%
\pgfpathlineto{\pgfqpoint{-0.055556in}{0.000000in}}%
\pgfusepath{stroke,fill}%
}%
\begin{pgfscope}%
\pgfsys@transformshift{2.630906in}{1.316277in}%
\pgfsys@useobject{currentmarker}{}%
\end{pgfscope}%
\end{pgfscope}%
\begin{pgfscope}%
\pgftext[x=0.309848in,y=1.316277in,right,]{{\rmfamily\fontsize{8.000000}{9.600000}\selectfont \(\displaystyle 0.5\)}}%
\end{pgfscope}%
\begin{pgfscope}%
\pgfsetbuttcap%
\pgfsetroundjoin%
\definecolor{currentfill}{rgb}{0.000000,0.000000,0.000000}%
\pgfsetfillcolor{currentfill}%
\pgfsetlinewidth{0.501875pt}%
\definecolor{currentstroke}{rgb}{0.000000,0.000000,0.000000}%
\pgfsetstrokecolor{currentstroke}%
\pgfsetdash{}{0pt}%
\pgfsys@defobject{currentmarker}{\pgfqpoint{0.000000in}{0.000000in}}{\pgfqpoint{0.055556in}{0.000000in}}{%
\pgfpathmoveto{\pgfqpoint{0.000000in}{0.000000in}}%
\pgfpathlineto{\pgfqpoint{0.055556in}{0.000000in}}%
\pgfusepath{stroke,fill}%
}%
\begin{pgfscope}%
\pgfsys@transformshift{0.365404in}{1.522752in}%
\pgfsys@useobject{currentmarker}{}%
\end{pgfscope}%
\end{pgfscope}%
\begin{pgfscope}%
\pgfsetbuttcap%
\pgfsetroundjoin%
\definecolor{currentfill}{rgb}{0.000000,0.000000,0.000000}%
\pgfsetfillcolor{currentfill}%
\pgfsetlinewidth{0.501875pt}%
\definecolor{currentstroke}{rgb}{0.000000,0.000000,0.000000}%
\pgfsetstrokecolor{currentstroke}%
\pgfsetdash{}{0pt}%
\pgfsys@defobject{currentmarker}{\pgfqpoint{-0.055556in}{0.000000in}}{\pgfqpoint{0.000000in}{0.000000in}}{%
\pgfpathmoveto{\pgfqpoint{0.000000in}{0.000000in}}%
\pgfpathlineto{\pgfqpoint{-0.055556in}{0.000000in}}%
\pgfusepath{stroke,fill}%
}%
\begin{pgfscope}%
\pgfsys@transformshift{2.630906in}{1.522752in}%
\pgfsys@useobject{currentmarker}{}%
\end{pgfscope}%
\end{pgfscope}%
\begin{pgfscope}%
\pgftext[x=0.309848in,y=1.522752in,right,]{{\rmfamily\fontsize{8.000000}{9.600000}\selectfont \(\displaystyle 0.6\)}}%
\end{pgfscope}%
\begin{pgfscope}%
\pgftext[x=0.089553in,y=0.903327in,,bottom,rotate=90.000000]{{\rmfamily\fontsize{10.000000}{12.000000}\selectfont \(\displaystyle p(s)\)}}%
\end{pgfscope}%
\begin{pgfscope}%
\pgfsetbuttcap%
\pgfsetroundjoin%
\pgfsetlinewidth{1.003750pt}%
\definecolor{currentstroke}{rgb}{0.000000,0.000000,0.000000}%
\pgfsetstrokecolor{currentstroke}%
\pgfsetdash{}{0pt}%
\pgfpathmoveto{\pgfqpoint{0.365404in}{1.625989in}}%
\pgfpathlineto{\pgfqpoint{2.630906in}{1.625989in}}%
\pgfusepath{stroke}%
\end{pgfscope}%
\begin{pgfscope}%
\pgfsetbuttcap%
\pgfsetroundjoin%
\pgfsetlinewidth{1.003750pt}%
\definecolor{currentstroke}{rgb}{0.000000,0.000000,0.000000}%
\pgfsetstrokecolor{currentstroke}%
\pgfsetdash{}{0pt}%
\pgfpathmoveto{\pgfqpoint{2.630906in}{0.180665in}}%
\pgfpathlineto{\pgfqpoint{2.630906in}{1.625989in}}%
\pgfusepath{stroke}%
\end{pgfscope}%
\begin{pgfscope}%
\pgfsetbuttcap%
\pgfsetroundjoin%
\pgfsetlinewidth{1.003750pt}%
\definecolor{currentstroke}{rgb}{0.000000,0.000000,0.000000}%
\pgfsetstrokecolor{currentstroke}%
\pgfsetdash{}{0pt}%
\pgfpathmoveto{\pgfqpoint{0.365404in}{0.180665in}}%
\pgfpathlineto{\pgfqpoint{2.630906in}{0.180665in}}%
\pgfusepath{stroke}%
\end{pgfscope}%
\begin{pgfscope}%
\pgfsetbuttcap%
\pgfsetroundjoin%
\pgfsetlinewidth{1.003750pt}%
\definecolor{currentstroke}{rgb}{0.000000,0.000000,0.000000}%
\pgfsetstrokecolor{currentstroke}%
\pgfsetdash{}{0pt}%
\pgfpathmoveto{\pgfqpoint{0.365404in}{0.180665in}}%
\pgfpathlineto{\pgfqpoint{0.365404in}{1.625989in}}%
\pgfusepath{stroke}%
\end{pgfscope}%
\begin{pgfscope}%
\pgftext[x=1.498155in,y=1.481457in,left,top]{{\rmfamily\fontsize{12.000000}{14.400000}\selectfont \(\displaystyle q=8\)}}%
\end{pgfscope}%
\begin{pgfscope}%
\pgftext[x=1.498155in,y=1.336924in,left,top]{{\rmfamily\fontsize{12.000000}{14.400000}\selectfont \(\displaystyle p=0.3\)}}%
\end{pgfscope}%
\begin{pgfscope}%
\pgftext[x=1.498155in,y=1.192392in,left,top]{{\rmfamily\fontsize{12.000000}{14.400000}\selectfont \(\displaystyle L=128, d=2\) }}%
\end{pgfscope}%
\begin{pgfscope}%
\pgfpathrectangle{\pgfqpoint{0.365404in}{0.180665in}}{\pgfqpoint{2.265502in}{1.445324in}} %
\pgfusepath{clip}%
\pgfsetbuttcap%
\pgfsetroundjoin%
\pgfsetlinewidth{1.003750pt}%
\definecolor{currentstroke}{rgb}{0.000000,0.000000,0.000000}%
\pgfsetstrokecolor{currentstroke}%
\pgfsetdash{}{0pt}%
\pgfpathmoveto{\pgfqpoint{0.377995in}{0.283903in}}%
\pgfpathlineto{\pgfqpoint{0.377995in}{0.287890in}}%
\pgfpathlineto{\pgfqpoint{0.392356in}{0.287890in}}%
\pgfpathlineto{\pgfqpoint{0.392356in}{0.285896in}}%
\pgfpathlineto{\pgfqpoint{0.421076in}{0.285896in}}%
\pgfpathlineto{\pgfqpoint{0.421076in}{0.289883in}}%
\pgfpathlineto{\pgfqpoint{0.435436in}{0.289883in}}%
\pgfpathlineto{\pgfqpoint{0.435436in}{0.291877in}}%
\pgfpathlineto{\pgfqpoint{0.449796in}{0.291877in}}%
\pgfpathlineto{\pgfqpoint{0.449796in}{0.305832in}}%
\pgfpathlineto{\pgfqpoint{0.464156in}{0.305832in}}%
\pgfpathlineto{\pgfqpoint{0.464156in}{0.297858in}}%
\pgfpathlineto{\pgfqpoint{0.478516in}{0.297858in}}%
\pgfpathlineto{\pgfqpoint{0.478516in}{0.335734in}}%
\pgfpathlineto{\pgfqpoint{0.492876in}{0.335734in}}%
\pgfpathlineto{\pgfqpoint{0.492876in}{0.327760in}}%
\pgfpathlineto{\pgfqpoint{0.507237in}{0.327760in}}%
\pgfpathlineto{\pgfqpoint{0.507237in}{0.343709in}}%
\pgfpathlineto{\pgfqpoint{0.521597in}{0.343709in}}%
\pgfpathlineto{\pgfqpoint{0.521597in}{0.377598in}}%
\pgfpathlineto{\pgfqpoint{0.535957in}{0.377598in}}%
\pgfpathlineto{\pgfqpoint{0.535957in}{0.393547in}}%
\pgfpathlineto{\pgfqpoint{0.550317in}{0.393547in}}%
\pgfpathlineto{\pgfqpoint{0.550317in}{0.485249in}}%
\pgfpathlineto{\pgfqpoint{0.564677in}{0.485249in}}%
\pgfpathlineto{\pgfqpoint{0.564677in}{0.471294in}}%
\pgfpathlineto{\pgfqpoint{0.579037in}{0.471294in}}%
\pgfpathlineto{\pgfqpoint{0.579037in}{0.527113in}}%
\pgfpathlineto{\pgfqpoint{0.593397in}{0.527113in}}%
\pgfpathlineto{\pgfqpoint{0.593397in}{0.572964in}}%
\pgfpathlineto{\pgfqpoint{0.607757in}{0.572964in}}%
\pgfpathlineto{\pgfqpoint{0.607757in}{0.700549in}}%
\pgfpathlineto{\pgfqpoint{0.622117in}{0.700549in}}%
\pgfpathlineto{\pgfqpoint{0.622117in}{0.694569in}}%
\pgfpathlineto{\pgfqpoint{0.636478in}{0.694569in}}%
\pgfpathlineto{\pgfqpoint{0.636478in}{0.724471in}}%
\pgfpathlineto{\pgfqpoint{0.650838in}{0.724471in}}%
\pgfpathlineto{\pgfqpoint{0.650838in}{0.802219in}}%
\pgfpathlineto{\pgfqpoint{0.665198in}{0.802219in}}%
\pgfpathlineto{\pgfqpoint{0.665198in}{0.762348in}}%
\pgfpathlineto{\pgfqpoint{0.679558in}{0.762348in}}%
\pgfpathlineto{\pgfqpoint{0.679558in}{0.907876in}}%
\pgfpathlineto{\pgfqpoint{0.693918in}{0.907876in}}%
\pgfpathlineto{\pgfqpoint{0.693918in}{1.019513in}}%
\pgfpathlineto{\pgfqpoint{0.708278in}{1.019513in}}%
\pgfpathlineto{\pgfqpoint{0.708278in}{0.979642in}}%
\pgfpathlineto{\pgfqpoint{0.722638in}{0.979642in}}%
\pgfpathlineto{\pgfqpoint{0.722638in}{1.071345in}}%
\pgfpathlineto{\pgfqpoint{0.736998in}{1.071345in}}%
\pgfpathlineto{\pgfqpoint{0.736998in}{1.077325in}}%
\pgfpathlineto{\pgfqpoint{0.751359in}{1.077325in}}%
\pgfpathlineto{\pgfqpoint{0.751359in}{1.151085in}}%
\pgfpathlineto{\pgfqpoint{0.765719in}{1.151085in}}%
\pgfpathlineto{\pgfqpoint{0.765719in}{1.242788in}}%
\pgfpathlineto{\pgfqpoint{0.780079in}{1.242788in}}%
\pgfpathlineto{\pgfqpoint{0.780079in}{1.178995in}}%
\pgfpathlineto{\pgfqpoint{0.794439in}{1.178995in}}%
\pgfpathlineto{\pgfqpoint{0.794439in}{1.216872in}}%
\pgfpathlineto{\pgfqpoint{0.808799in}{1.216872in}}%
\pgfpathlineto{\pgfqpoint{0.808799in}{1.284652in}}%
\pgfpathlineto{\pgfqpoint{0.823159in}{1.284652in}}%
\pgfpathlineto{\pgfqpoint{0.823159in}{1.282658in}}%
\pgfpathlineto{\pgfqpoint{0.837519in}{1.282658in}}%
\pgfpathlineto{\pgfqpoint{0.837519in}{1.314554in}}%
\pgfpathlineto{\pgfqpoint{0.851879in}{1.314554in}}%
\pgfpathlineto{\pgfqpoint{0.851879in}{1.252755in}}%
\pgfpathlineto{\pgfqpoint{0.866240in}{1.252755in}}%
\pgfpathlineto{\pgfqpoint{0.866240in}{1.316548in}}%
\pgfpathlineto{\pgfqpoint{0.880600in}{1.316548in}}%
\pgfpathlineto{\pgfqpoint{0.880600in}{1.296613in}}%
\pgfpathlineto{\pgfqpoint{0.894960in}{1.296613in}}%
\pgfpathlineto{\pgfqpoint{0.894960in}{1.212885in}}%
\pgfpathlineto{\pgfqpoint{0.909320in}{1.212885in}}%
\pgfpathlineto{\pgfqpoint{0.909320in}{1.240794in}}%
\pgfpathlineto{\pgfqpoint{0.923680in}{1.240794in}}%
\pgfpathlineto{\pgfqpoint{0.923680in}{1.129157in}}%
\pgfpathlineto{\pgfqpoint{0.938040in}{1.129157in}}%
\pgfpathlineto{\pgfqpoint{0.938040in}{1.139124in}}%
\pgfpathlineto{\pgfqpoint{0.952400in}{1.139124in}}%
\pgfpathlineto{\pgfqpoint{0.952400in}{1.163047in}}%
\pgfpathlineto{\pgfqpoint{0.966760in}{1.163047in}}%
\pgfpathlineto{\pgfqpoint{0.966760in}{1.200924in}}%
\pgfpathlineto{\pgfqpoint{0.981121in}{1.200924in}}%
\pgfpathlineto{\pgfqpoint{0.981121in}{1.079319in}}%
\pgfpathlineto{\pgfqpoint{0.995481in}{1.079319in}}%
\pgfpathlineto{\pgfqpoint{0.995481in}{1.099254in}}%
\pgfpathlineto{\pgfqpoint{1.009841in}{1.099254in}}%
\pgfpathlineto{\pgfqpoint{1.009841in}{1.011539in}}%
\pgfpathlineto{\pgfqpoint{1.024201in}{1.011539in}}%
\pgfpathlineto{\pgfqpoint{1.024201in}{1.009545in}}%
\pgfpathlineto{\pgfqpoint{1.038561in}{1.009545in}}%
\pgfpathlineto{\pgfqpoint{1.038561in}{0.931798in}}%
\pgfpathlineto{\pgfqpoint{1.052921in}{0.931798in}}%
\pgfpathlineto{\pgfqpoint{1.052921in}{0.971668in}}%
\pgfpathlineto{\pgfqpoint{1.067281in}{0.971668in}}%
\pgfpathlineto{\pgfqpoint{1.067281in}{0.869999in}}%
\pgfpathlineto{\pgfqpoint{1.081641in}{0.869999in}}%
\pgfpathlineto{\pgfqpoint{1.081641in}{0.879966in}}%
\pgfpathlineto{\pgfqpoint{1.096002in}{0.879966in}}%
\pgfpathlineto{\pgfqpoint{1.096002in}{0.903889in}}%
\pgfpathlineto{\pgfqpoint{1.110362in}{0.903889in}}%
\pgfpathlineto{\pgfqpoint{1.110362in}{0.869999in}}%
\pgfpathlineto{\pgfqpoint{1.124722in}{0.869999in}}%
\pgfpathlineto{\pgfqpoint{1.124722in}{0.826141in}}%
\pgfpathlineto{\pgfqpoint{1.139082in}{0.826141in}}%
\pgfpathlineto{\pgfqpoint{1.139082in}{0.760355in}}%
\pgfpathlineto{\pgfqpoint{1.153442in}{0.760355in}}%
\pgfpathlineto{\pgfqpoint{1.153442in}{0.786271in}}%
\pgfpathlineto{\pgfqpoint{1.167802in}{0.786271in}}%
\pgfpathlineto{\pgfqpoint{1.167802in}{0.716497in}}%
\pgfpathlineto{\pgfqpoint{1.182162in}{0.716497in}}%
\pgfpathlineto{\pgfqpoint{1.182162in}{0.700549in}}%
\pgfpathlineto{\pgfqpoint{1.196522in}{0.700549in}}%
\pgfpathlineto{\pgfqpoint{1.196522in}{0.668653in}}%
\pgfpathlineto{\pgfqpoint{1.210883in}{0.668653in}}%
\pgfpathlineto{\pgfqpoint{1.210883in}{0.634763in}}%
\pgfpathlineto{\pgfqpoint{1.225243in}{0.634763in}}%
\pgfpathlineto{\pgfqpoint{1.225243in}{0.678620in}}%
\pgfpathlineto{\pgfqpoint{1.239603in}{0.678620in}}%
\pgfpathlineto{\pgfqpoint{1.239603in}{0.624795in}}%
\pgfpathlineto{\pgfqpoint{1.253963in}{0.624795in}}%
\pgfpathlineto{\pgfqpoint{1.253963in}{0.612834in}}%
\pgfpathlineto{\pgfqpoint{1.268323in}{0.612834in}}%
\pgfpathlineto{\pgfqpoint{1.268323in}{0.515152in}}%
\pgfpathlineto{\pgfqpoint{1.282683in}{0.515152in}}%
\pgfpathlineto{\pgfqpoint{1.282683in}{0.517145in}}%
\pgfpathlineto{\pgfqpoint{1.297043in}{0.517145in}}%
\pgfpathlineto{\pgfqpoint{1.297043in}{0.531100in}}%
\pgfpathlineto{\pgfqpoint{1.311403in}{0.531100in}}%
\pgfpathlineto{\pgfqpoint{1.311403in}{0.539074in}}%
\pgfpathlineto{\pgfqpoint{1.325764in}{0.539074in}}%
\pgfpathlineto{\pgfqpoint{1.325764in}{0.501197in}}%
\pgfpathlineto{\pgfqpoint{1.340124in}{0.501197in}}%
\pgfpathlineto{\pgfqpoint{1.340124in}{0.515152in}}%
\pgfpathlineto{\pgfqpoint{1.354484in}{0.515152in}}%
\pgfpathlineto{\pgfqpoint{1.354484in}{0.477275in}}%
\pgfpathlineto{\pgfqpoint{1.368844in}{0.477275in}}%
\pgfpathlineto{\pgfqpoint{1.368844in}{0.475281in}}%
\pgfpathlineto{\pgfqpoint{1.383204in}{0.475281in}}%
\pgfpathlineto{\pgfqpoint{1.383204in}{0.425443in}}%
\pgfpathlineto{\pgfqpoint{1.397564in}{0.425443in}}%
\pgfpathlineto{\pgfqpoint{1.397564in}{0.473288in}}%
\pgfpathlineto{\pgfqpoint{1.411924in}{0.473288in}}%
\pgfpathlineto{\pgfqpoint{1.411924in}{0.457339in}}%
\pgfpathlineto{\pgfqpoint{1.426284in}{0.457339in}}%
\pgfpathlineto{\pgfqpoint{1.426284in}{0.431424in}}%
\pgfpathlineto{\pgfqpoint{1.440645in}{0.431424in}}%
\pgfpathlineto{\pgfqpoint{1.440645in}{0.427437in}}%
\pgfpathlineto{\pgfqpoint{1.455005in}{0.427437in}}%
\pgfpathlineto{\pgfqpoint{1.455005in}{0.421456in}}%
\pgfpathlineto{\pgfqpoint{1.469365in}{0.421456in}}%
\pgfpathlineto{\pgfqpoint{1.469365in}{0.413482in}}%
\pgfpathlineto{\pgfqpoint{1.483725in}{0.413482in}}%
\pgfpathlineto{\pgfqpoint{1.483725in}{0.383579in}}%
\pgfpathlineto{\pgfqpoint{1.498085in}{0.383579in}}%
\pgfpathlineto{\pgfqpoint{1.498085in}{0.381585in}}%
\pgfpathlineto{\pgfqpoint{1.512445in}{0.381585in}}%
\pgfpathlineto{\pgfqpoint{1.512445in}{0.363644in}}%
\pgfpathlineto{\pgfqpoint{1.526805in}{0.363644in}}%
\pgfpathlineto{\pgfqpoint{1.526805in}{0.383579in}}%
\pgfpathlineto{\pgfqpoint{1.541165in}{0.383579in}}%
\pgfpathlineto{\pgfqpoint{1.541165in}{0.389560in}}%
\pgfpathlineto{\pgfqpoint{1.555526in}{0.389560in}}%
\pgfpathlineto{\pgfqpoint{1.555526in}{0.345702in}}%
\pgfpathlineto{\pgfqpoint{1.569886in}{0.345702in}}%
\pgfpathlineto{\pgfqpoint{1.569886in}{0.357663in}}%
\pgfpathlineto{\pgfqpoint{1.584246in}{0.357663in}}%
\pgfpathlineto{\pgfqpoint{1.584246in}{0.383579in}}%
\pgfpathlineto{\pgfqpoint{1.598606in}{0.383579in}}%
\pgfpathlineto{\pgfqpoint{1.598606in}{0.363644in}}%
\pgfpathlineto{\pgfqpoint{1.612966in}{0.363644in}}%
\pgfpathlineto{\pgfqpoint{1.612966in}{0.351683in}}%
\pgfpathlineto{\pgfqpoint{1.627326in}{0.351683in}}%
\pgfpathlineto{\pgfqpoint{1.627326in}{0.339722in}}%
\pgfpathlineto{\pgfqpoint{1.641686in}{0.339722in}}%
\pgfpathlineto{\pgfqpoint{1.641686in}{0.347696in}}%
\pgfpathlineto{\pgfqpoint{1.656046in}{0.347696in}}%
\pgfpathlineto{\pgfqpoint{1.656046in}{0.321780in}}%
\pgfpathlineto{\pgfqpoint{1.670406in}{0.321780in}}%
\pgfpathlineto{\pgfqpoint{1.670406in}{0.331747in}}%
\pgfpathlineto{\pgfqpoint{1.699127in}{0.331747in}}%
\pgfpathlineto{\pgfqpoint{1.699127in}{0.315799in}}%
\pgfpathlineto{\pgfqpoint{1.713487in}{0.315799in}}%
\pgfpathlineto{\pgfqpoint{1.713487in}{0.333741in}}%
\pgfpathlineto{\pgfqpoint{1.727847in}{0.333741in}}%
\pgfpathlineto{\pgfqpoint{1.727847in}{0.315799in}}%
\pgfpathlineto{\pgfqpoint{1.742207in}{0.315799in}}%
\pgfpathlineto{\pgfqpoint{1.742207in}{0.311812in}}%
\pgfpathlineto{\pgfqpoint{1.756567in}{0.311812in}}%
\pgfpathlineto{\pgfqpoint{1.756567in}{0.325767in}}%
\pgfpathlineto{\pgfqpoint{1.770927in}{0.325767in}}%
\pgfpathlineto{\pgfqpoint{1.770927in}{0.311812in}}%
\pgfpathlineto{\pgfqpoint{1.785287in}{0.311812in}}%
\pgfpathlineto{\pgfqpoint{1.785287in}{0.303838in}}%
\pgfpathlineto{\pgfqpoint{1.799648in}{0.303838in}}%
\pgfpathlineto{\pgfqpoint{1.799648in}{0.311812in}}%
\pgfpathlineto{\pgfqpoint{1.814008in}{0.311812in}}%
\pgfpathlineto{\pgfqpoint{1.814008in}{0.303838in}}%
\pgfpathlineto{\pgfqpoint{1.842728in}{0.303838in}}%
\pgfpathlineto{\pgfqpoint{1.842728in}{0.307825in}}%
\pgfpathlineto{\pgfqpoint{1.857088in}{0.307825in}}%
\pgfpathlineto{\pgfqpoint{1.857088in}{0.305832in}}%
\pgfpathlineto{\pgfqpoint{1.871448in}{0.305832in}}%
\pgfpathlineto{\pgfqpoint{1.871448in}{0.303838in}}%
\pgfpathlineto{\pgfqpoint{1.900168in}{0.303838in}}%
\pgfpathlineto{\pgfqpoint{1.900168in}{0.299851in}}%
\pgfpathlineto{\pgfqpoint{1.914529in}{0.299851in}}%
\pgfpathlineto{\pgfqpoint{1.914529in}{0.293870in}}%
\pgfpathlineto{\pgfqpoint{1.928889in}{0.293870in}}%
\pgfpathlineto{\pgfqpoint{1.928889in}{0.301845in}}%
\pgfpathlineto{\pgfqpoint{1.943249in}{0.301845in}}%
\pgfpathlineto{\pgfqpoint{1.943249in}{0.291877in}}%
\pgfpathlineto{\pgfqpoint{1.957609in}{0.291877in}}%
\pgfpathlineto{\pgfqpoint{1.957609in}{0.297858in}}%
\pgfpathlineto{\pgfqpoint{2.000689in}{0.297858in}}%
\pgfpathlineto{\pgfqpoint{2.000689in}{0.293870in}}%
\pgfpathlineto{\pgfqpoint{2.015049in}{0.293870in}}%
\pgfpathlineto{\pgfqpoint{2.015049in}{0.291877in}}%
\pgfpathlineto{\pgfqpoint{2.029410in}{0.291877in}}%
\pgfpathlineto{\pgfqpoint{2.029410in}{0.287890in}}%
\pgfpathlineto{\pgfqpoint{2.043770in}{0.287890in}}%
\pgfpathlineto{\pgfqpoint{2.043770in}{0.291877in}}%
\pgfpathlineto{\pgfqpoint{2.058130in}{0.291877in}}%
\pgfpathlineto{\pgfqpoint{2.058130in}{0.297858in}}%
\pgfpathlineto{\pgfqpoint{2.072490in}{0.297858in}}%
\pgfpathlineto{\pgfqpoint{2.072490in}{0.285896in}}%
\pgfpathlineto{\pgfqpoint{2.086850in}{0.285896in}}%
\pgfpathlineto{\pgfqpoint{2.086850in}{0.287890in}}%
\pgfpathlineto{\pgfqpoint{2.115570in}{0.287890in}}%
\pgfpathlineto{\pgfqpoint{2.115570in}{0.293870in}}%
\pgfpathlineto{\pgfqpoint{2.129930in}{0.293870in}}%
\pgfpathlineto{\pgfqpoint{2.129930in}{0.283903in}}%
\pgfpathlineto{\pgfqpoint{2.144291in}{0.283903in}}%
\pgfpathlineto{\pgfqpoint{2.144291in}{0.285896in}}%
\pgfpathlineto{\pgfqpoint{2.158651in}{0.285896in}}%
\pgfpathlineto{\pgfqpoint{2.158651in}{0.289883in}}%
\pgfpathlineto{\pgfqpoint{2.173011in}{0.289883in}}%
\pgfpathlineto{\pgfqpoint{2.173011in}{0.291877in}}%
\pgfpathlineto{\pgfqpoint{2.187371in}{0.291877in}}%
\pgfpathlineto{\pgfqpoint{2.187371in}{0.287890in}}%
\pgfpathlineto{\pgfqpoint{2.201731in}{0.287890in}}%
\pgfpathlineto{\pgfqpoint{2.201731in}{0.293870in}}%
\pgfpathlineto{\pgfqpoint{2.216091in}{0.293870in}}%
\pgfpathlineto{\pgfqpoint{2.216091in}{0.289883in}}%
\pgfpathlineto{\pgfqpoint{2.230451in}{0.289883in}}%
\pgfpathlineto{\pgfqpoint{2.230451in}{0.285896in}}%
\pgfpathlineto{\pgfqpoint{2.259172in}{0.285896in}}%
\pgfpathlineto{\pgfqpoint{2.259172in}{0.283903in}}%
\pgfpathlineto{\pgfqpoint{2.273532in}{0.283903in}}%
\pgfpathlineto{\pgfqpoint{2.273532in}{0.285896in}}%
\pgfpathlineto{\pgfqpoint{2.302252in}{0.285896in}}%
\pgfpathlineto{\pgfqpoint{2.302252in}{0.289883in}}%
\pgfpathlineto{\pgfqpoint{2.316612in}{0.289883in}}%
\pgfpathlineto{\pgfqpoint{2.316612in}{0.285896in}}%
\pgfpathlineto{\pgfqpoint{2.330972in}{0.285896in}}%
\pgfpathlineto{\pgfqpoint{2.330972in}{0.283903in}}%
\pgfpathlineto{\pgfqpoint{2.345332in}{0.283903in}}%
\pgfpathlineto{\pgfqpoint{2.345332in}{0.285896in}}%
\pgfpathlineto{\pgfqpoint{2.359692in}{0.285896in}}%
\pgfpathlineto{\pgfqpoint{2.359692in}{0.283903in}}%
\pgfpathlineto{\pgfqpoint{2.374053in}{0.283903in}}%
\pgfpathlineto{\pgfqpoint{2.374053in}{0.285896in}}%
\pgfpathlineto{\pgfqpoint{2.402773in}{0.285896in}}%
\pgfpathlineto{\pgfqpoint{2.402773in}{0.283903in}}%
\pgfpathlineto{\pgfqpoint{2.417133in}{0.283903in}}%
\pgfpathlineto{\pgfqpoint{2.417133in}{0.285896in}}%
\pgfpathlineto{\pgfqpoint{2.431493in}{0.285896in}}%
\pgfpathlineto{\pgfqpoint{2.431493in}{0.283903in}}%
\pgfpathlineto{\pgfqpoint{2.445853in}{0.283903in}}%
\pgfpathlineto{\pgfqpoint{2.445853in}{0.285896in}}%
\pgfpathlineto{\pgfqpoint{2.474573in}{0.285896in}}%
\pgfpathlineto{\pgfqpoint{2.474573in}{0.283903in}}%
\pgfpathlineto{\pgfqpoint{2.503294in}{0.283903in}}%
\pgfpathlineto{\pgfqpoint{2.503294in}{0.285896in}}%
\pgfpathlineto{\pgfqpoint{2.517654in}{0.285896in}}%
\pgfpathlineto{\pgfqpoint{2.517654in}{0.283903in}}%
\pgfpathlineto{\pgfqpoint{2.575094in}{0.283903in}}%
\pgfpathlineto{\pgfqpoint{2.575094in}{0.287890in}}%
\pgfpathlineto{\pgfqpoint{2.589454in}{0.287890in}}%
\pgfpathlineto{\pgfqpoint{2.589454in}{0.283903in}}%
\pgfpathlineto{\pgfqpoint{2.640906in}{0.283903in}}%
\pgfpathlineto{\pgfqpoint{2.640906in}{0.283903in}}%
\pgfusepath{stroke}%
\end{pgfscope}%
\end{pgfpicture}%
\makeatother%
\endgroup%

%% file: gumbel_off_q1p5_3d.pgf
\begingroup%
\makeatletter%
\begin{pgfpicture}%
\pgfpathrectangle{\pgfpointorigin}{\pgfqpoint{2.923228in}{1.806655in}}%
\pgfusepath{use as bounding box}%
\begin{pgfscope}%
\pgfsetbuttcap%
\pgfsetroundjoin%
\definecolor{currentfill}{rgb}{1.000000,1.000000,1.000000}%
\pgfsetfillcolor{currentfill}%
\pgfsetlinewidth{0.000000pt}%
\definecolor{currentstroke}{rgb}{1.000000,1.000000,1.000000}%
\pgfsetstrokecolor{currentstroke}%
\pgfsetdash{}{0pt}%
\pgfpathmoveto{\pgfqpoint{0.000000in}{0.000000in}}%
\pgfpathlineto{\pgfqpoint{2.923228in}{0.000000in}}%
\pgfpathlineto{\pgfqpoint{2.923228in}{1.806655in}}%
\pgfpathlineto{\pgfqpoint{0.000000in}{1.806655in}}%
\pgfpathclose%
\pgfusepath{fill}%
\end{pgfscope}%
\begin{pgfscope}%
\pgfsetbuttcap%
\pgfsetroundjoin%
\definecolor{currentfill}{rgb}{1.000000,1.000000,1.000000}%
\pgfsetfillcolor{currentfill}%
\pgfsetlinewidth{0.000000pt}%
\definecolor{currentstroke}{rgb}{0.000000,0.000000,0.000000}%
\pgfsetstrokecolor{currentstroke}%
\pgfsetstrokeopacity{0.000000}%
\pgfsetdash{}{0pt}%
\pgfpathmoveto{\pgfqpoint{0.365404in}{0.180665in}}%
\pgfpathlineto{\pgfqpoint{2.630906in}{0.180665in}}%
\pgfpathlineto{\pgfqpoint{2.630906in}{1.625989in}}%
\pgfpathlineto{\pgfqpoint{0.365404in}{1.625989in}}%
\pgfpathclose%
\pgfusepath{fill}%
\end{pgfscope}%
\begin{pgfscope}%
\pgfpathrectangle{\pgfqpoint{0.365404in}{0.180665in}}{\pgfqpoint{2.265502in}{1.445324in}} %
\pgfusepath{clip}%
\pgfsetrectcap%
\pgfsetroundjoin%
\pgfsetlinewidth{0.501875pt}%
\definecolor{currentstroke}{rgb}{0.000000,0.500000,0.000000}%
\pgfsetstrokecolor{currentstroke}%
\pgfsetdash{}{0pt}%
\pgfpathmoveto{\pgfqpoint{0.365402in}{0.284527in}}%
\pgfpathlineto{\pgfqpoint{0.399712in}{0.286652in}}%
\pgfpathlineto{\pgfqpoint{0.422035in}{0.290183in}}%
\pgfpathlineto{\pgfqpoint{0.439719in}{0.295182in}}%
\pgfpathlineto{\pgfqpoint{0.454999in}{0.301789in}}%
\pgfpathlineto{\pgfqpoint{0.468932in}{0.310240in}}%
\pgfpathlineto{\pgfqpoint{0.482146in}{0.320882in}}%
\pgfpathlineto{\pgfqpoint{0.495058in}{0.334186in}}%
\pgfpathlineto{\pgfqpoint{0.507981in}{0.350766in}}%
\pgfpathlineto{\pgfqpoint{0.521170in}{0.371402in}}%
\pgfpathlineto{\pgfqpoint{0.534856in}{0.397070in}}%
\pgfpathlineto{\pgfqpoint{0.549283in}{0.429009in}}%
\pgfpathlineto{\pgfqpoint{0.564755in}{0.468847in}}%
\pgfpathlineto{\pgfqpoint{0.581712in}{0.518886in}}%
\pgfpathlineto{\pgfqpoint{0.600945in}{0.582870in}}%
\pgfpathlineto{\pgfqpoint{0.624316in}{0.668805in}}%
\pgfpathlineto{\pgfqpoint{0.662244in}{0.818241in}}%
\pgfpathlineto{\pgfqpoint{0.699574in}{0.962080in}}%
\pgfpathlineto{\pgfqpoint{0.722418in}{1.041430in}}%
\pgfpathlineto{\pgfqpoint{0.741422in}{1.099736in}}%
\pgfpathlineto{\pgfqpoint{0.758159in}{1.144192in}}%
\pgfpathlineto{\pgfqpoint{0.773261in}{1.178245in}}%
\pgfpathlineto{\pgfqpoint{0.787065in}{1.204101in}}%
\pgfpathlineto{\pgfqpoint{0.799782in}{1.223385in}}%
\pgfpathlineto{\pgfqpoint{0.811572in}{1.237384in}}%
\pgfpathlineto{\pgfqpoint{0.822585in}{1.247133in}}%
\pgfpathlineto{\pgfqpoint{0.832978in}{1.253456in}}%
\pgfpathlineto{\pgfqpoint{0.842946in}{1.256976in}}%
\pgfpathlineto{\pgfqpoint{0.852738in}{1.258100in}}%
\pgfpathlineto{\pgfqpoint{0.862630in}{1.256988in}}%
\pgfpathlineto{\pgfqpoint{0.872919in}{1.253558in}}%
\pgfpathlineto{\pgfqpoint{0.883898in}{1.247492in}}%
\pgfpathlineto{\pgfqpoint{0.895855in}{1.238264in}}%
\pgfpathlineto{\pgfqpoint{0.909084in}{1.225145in}}%
\pgfpathlineto{\pgfqpoint{0.923932in}{1.207169in}}%
\pgfpathlineto{\pgfqpoint{0.940864in}{1.183037in}}%
\pgfpathlineto{\pgfqpoint{0.960598in}{1.150869in}}%
\pgfpathlineto{\pgfqpoint{0.984497in}{1.107444in}}%
\pgfpathlineto{\pgfqpoint{1.016063in}{1.045131in}}%
\pgfpathlineto{\pgfqpoint{1.145560in}{0.784401in}}%
\pgfpathlineto{\pgfqpoint{1.178974in}{0.724888in}}%
\pgfpathlineto{\pgfqpoint{1.210027in}{0.674035in}}%
\pgfpathlineto{\pgfqpoint{1.239764in}{0.629536in}}%
\pgfpathlineto{\pgfqpoint{1.268744in}{0.590114in}}%
\pgfpathlineto{\pgfqpoint{1.297338in}{0.554928in}}%
\pgfpathlineto{\pgfqpoint{1.325824in}{0.523363in}}%
\pgfpathlineto{\pgfqpoint{1.354436in}{0.494950in}}%
\pgfpathlineto{\pgfqpoint{1.383398in}{0.469301in}}%
\pgfpathlineto{\pgfqpoint{1.412932in}{0.446103in}}%
\pgfpathlineto{\pgfqpoint{1.443273in}{0.425090in}}%
\pgfpathlineto{\pgfqpoint{1.474677in}{0.406046in}}%
\pgfpathlineto{\pgfqpoint{1.507434in}{0.388786in}}%
\pgfpathlineto{\pgfqpoint{1.541884in}{0.373158in}}%
\pgfpathlineto{\pgfqpoint{1.578436in}{0.359033in}}%
\pgfpathlineto{\pgfqpoint{1.617580in}{0.346312in}}%
\pgfpathlineto{\pgfqpoint{1.659945in}{0.334909in}}%
\pgfpathlineto{\pgfqpoint{1.706325in}{0.324761in}}%
\pgfpathlineto{\pgfqpoint{1.757796in}{0.315815in}}%
\pgfpathlineto{\pgfqpoint{1.815827in}{0.308033in}}%
\pgfpathlineto{\pgfqpoint{1.882562in}{0.301384in}}%
\pgfpathlineto{\pgfqpoint{1.961279in}{0.295844in}}%
\pgfpathlineto{\pgfqpoint{2.057411in}{0.291393in}}%
\pgfpathlineto{\pgfqpoint{2.180934in}{0.288013in}}%
\pgfpathlineto{\pgfqpoint{2.353317in}{0.285680in}}%
\pgfpathlineto{\pgfqpoint{2.630907in}{0.284363in}}%
\pgfpathlineto{\pgfqpoint{2.630907in}{0.284363in}}%
\pgfusepath{stroke}%
\end{pgfscope}%
\begin{pgfscope}%
\pgfsetbuttcap%
\pgfsetroundjoin%
\definecolor{currentfill}{rgb}{0.000000,0.000000,0.000000}%
\pgfsetfillcolor{currentfill}%
\pgfsetlinewidth{0.501875pt}%
\definecolor{currentstroke}{rgb}{0.000000,0.000000,0.000000}%
\pgfsetstrokecolor{currentstroke}%
\pgfsetdash{}{0pt}%
\pgfsys@defobject{currentmarker}{\pgfqpoint{0.000000in}{0.000000in}}{\pgfqpoint{0.000000in}{0.055556in}}{%
\pgfpathmoveto{\pgfqpoint{0.000000in}{0.000000in}}%
\pgfpathlineto{\pgfqpoint{0.000000in}{0.055556in}}%
\pgfusepath{stroke,fill}%
}%
\begin{pgfscope}%
\pgfsys@transformshift{0.444433in}{0.180665in}%
\pgfsys@useobject{currentmarker}{}%
\end{pgfscope}%
\end{pgfscope}%
\begin{pgfscope}%
\pgfsetbuttcap%
\pgfsetroundjoin%
\definecolor{currentfill}{rgb}{0.000000,0.000000,0.000000}%
\pgfsetfillcolor{currentfill}%
\pgfsetlinewidth{0.501875pt}%
\definecolor{currentstroke}{rgb}{0.000000,0.000000,0.000000}%
\pgfsetstrokecolor{currentstroke}%
\pgfsetdash{}{0pt}%
\pgfsys@defobject{currentmarker}{\pgfqpoint{0.000000in}{-0.055556in}}{\pgfqpoint{0.000000in}{0.000000in}}{%
\pgfpathmoveto{\pgfqpoint{0.000000in}{0.000000in}}%
\pgfpathlineto{\pgfqpoint{0.000000in}{-0.055556in}}%
\pgfusepath{stroke,fill}%
}%
\begin{pgfscope}%
\pgfsys@transformshift{0.444433in}{1.625989in}%
\pgfsys@useobject{currentmarker}{}%
\end{pgfscope}%
\end{pgfscope}%
\begin{pgfscope}%
\pgftext[x=0.444433in,y=0.125110in,,top]{{\rmfamily\fontsize{8.000000}{9.600000}\selectfont \(\displaystyle -2\)}}%
\end{pgfscope}%
\begin{pgfscope}%
\pgfsetbuttcap%
\pgfsetroundjoin%
\definecolor{currentfill}{rgb}{0.000000,0.000000,0.000000}%
\pgfsetfillcolor{currentfill}%
\pgfsetlinewidth{0.501875pt}%
\definecolor{currentstroke}{rgb}{0.000000,0.000000,0.000000}%
\pgfsetstrokecolor{currentstroke}%
\pgfsetdash{}{0pt}%
\pgfsys@defobject{currentmarker}{\pgfqpoint{0.000000in}{0.000000in}}{\pgfqpoint{0.000000in}{0.055556in}}{%
\pgfpathmoveto{\pgfqpoint{0.000000in}{0.000000in}}%
\pgfpathlineto{\pgfqpoint{0.000000in}{0.055556in}}%
\pgfusepath{stroke,fill}%
}%
\begin{pgfscope}%
\pgfsys@transformshift{0.707863in}{0.180665in}%
\pgfsys@useobject{currentmarker}{}%
\end{pgfscope}%
\end{pgfscope}%
\begin{pgfscope}%
\pgfsetbuttcap%
\pgfsetroundjoin%
\definecolor{currentfill}{rgb}{0.000000,0.000000,0.000000}%
\pgfsetfillcolor{currentfill}%
\pgfsetlinewidth{0.501875pt}%
\definecolor{currentstroke}{rgb}{0.000000,0.000000,0.000000}%
\pgfsetstrokecolor{currentstroke}%
\pgfsetdash{}{0pt}%
\pgfsys@defobject{currentmarker}{\pgfqpoint{0.000000in}{-0.055556in}}{\pgfqpoint{0.000000in}{0.000000in}}{%
\pgfpathmoveto{\pgfqpoint{0.000000in}{0.000000in}}%
\pgfpathlineto{\pgfqpoint{0.000000in}{-0.055556in}}%
\pgfusepath{stroke,fill}%
}%
\begin{pgfscope}%
\pgfsys@transformshift{0.707863in}{1.625989in}%
\pgfsys@useobject{currentmarker}{}%
\end{pgfscope}%
\end{pgfscope}%
\begin{pgfscope}%
\pgftext[x=0.707863in,y=0.125110in,,top]{{\rmfamily\fontsize{8.000000}{9.600000}\selectfont \(\displaystyle -1\)}}%
\end{pgfscope}%
\begin{pgfscope}%
\pgfsetbuttcap%
\pgfsetroundjoin%
\definecolor{currentfill}{rgb}{0.000000,0.000000,0.000000}%
\pgfsetfillcolor{currentfill}%
\pgfsetlinewidth{0.501875pt}%
\definecolor{currentstroke}{rgb}{0.000000,0.000000,0.000000}%
\pgfsetstrokecolor{currentstroke}%
\pgfsetdash{}{0pt}%
\pgfsys@defobject{currentmarker}{\pgfqpoint{0.000000in}{0.000000in}}{\pgfqpoint{0.000000in}{0.055556in}}{%
\pgfpathmoveto{\pgfqpoint{0.000000in}{0.000000in}}%
\pgfpathlineto{\pgfqpoint{0.000000in}{0.055556in}}%
\pgfusepath{stroke,fill}%
}%
\begin{pgfscope}%
\pgfsys@transformshift{0.971294in}{0.180665in}%
\pgfsys@useobject{currentmarker}{}%
\end{pgfscope}%
\end{pgfscope}%
\begin{pgfscope}%
\pgfsetbuttcap%
\pgfsetroundjoin%
\definecolor{currentfill}{rgb}{0.000000,0.000000,0.000000}%
\pgfsetfillcolor{currentfill}%
\pgfsetlinewidth{0.501875pt}%
\definecolor{currentstroke}{rgb}{0.000000,0.000000,0.000000}%
\pgfsetstrokecolor{currentstroke}%
\pgfsetdash{}{0pt}%
\pgfsys@defobject{currentmarker}{\pgfqpoint{0.000000in}{-0.055556in}}{\pgfqpoint{0.000000in}{0.000000in}}{%
\pgfpathmoveto{\pgfqpoint{0.000000in}{0.000000in}}%
\pgfpathlineto{\pgfqpoint{0.000000in}{-0.055556in}}%
\pgfusepath{stroke,fill}%
}%
\begin{pgfscope}%
\pgfsys@transformshift{0.971294in}{1.625989in}%
\pgfsys@useobject{currentmarker}{}%
\end{pgfscope}%
\end{pgfscope}%
\begin{pgfscope}%
\pgftext[x=0.971294in,y=0.125110in,,top]{{\rmfamily\fontsize{8.000000}{9.600000}\selectfont \(\displaystyle 0\)}}%
\end{pgfscope}%
\begin{pgfscope}%
\pgfsetbuttcap%
\pgfsetroundjoin%
\definecolor{currentfill}{rgb}{0.000000,0.000000,0.000000}%
\pgfsetfillcolor{currentfill}%
\pgfsetlinewidth{0.501875pt}%
\definecolor{currentstroke}{rgb}{0.000000,0.000000,0.000000}%
\pgfsetstrokecolor{currentstroke}%
\pgfsetdash{}{0pt}%
\pgfsys@defobject{currentmarker}{\pgfqpoint{0.000000in}{0.000000in}}{\pgfqpoint{0.000000in}{0.055556in}}{%
\pgfpathmoveto{\pgfqpoint{0.000000in}{0.000000in}}%
\pgfpathlineto{\pgfqpoint{0.000000in}{0.055556in}}%
\pgfusepath{stroke,fill}%
}%
\begin{pgfscope}%
\pgfsys@transformshift{1.234724in}{0.180665in}%
\pgfsys@useobject{currentmarker}{}%
\end{pgfscope}%
\end{pgfscope}%
\begin{pgfscope}%
\pgfsetbuttcap%
\pgfsetroundjoin%
\definecolor{currentfill}{rgb}{0.000000,0.000000,0.000000}%
\pgfsetfillcolor{currentfill}%
\pgfsetlinewidth{0.501875pt}%
\definecolor{currentstroke}{rgb}{0.000000,0.000000,0.000000}%
\pgfsetstrokecolor{currentstroke}%
\pgfsetdash{}{0pt}%
\pgfsys@defobject{currentmarker}{\pgfqpoint{0.000000in}{-0.055556in}}{\pgfqpoint{0.000000in}{0.000000in}}{%
\pgfpathmoveto{\pgfqpoint{0.000000in}{0.000000in}}%
\pgfpathlineto{\pgfqpoint{0.000000in}{-0.055556in}}%
\pgfusepath{stroke,fill}%
}%
\begin{pgfscope}%
\pgfsys@transformshift{1.234724in}{1.625989in}%
\pgfsys@useobject{currentmarker}{}%
\end{pgfscope}%
\end{pgfscope}%
\begin{pgfscope}%
\pgftext[x=1.234724in,y=0.125110in,,top]{{\rmfamily\fontsize{8.000000}{9.600000}\selectfont \(\displaystyle 1\)}}%
\end{pgfscope}%
\begin{pgfscope}%
\pgfsetbuttcap%
\pgfsetroundjoin%
\definecolor{currentfill}{rgb}{0.000000,0.000000,0.000000}%
\pgfsetfillcolor{currentfill}%
\pgfsetlinewidth{0.501875pt}%
\definecolor{currentstroke}{rgb}{0.000000,0.000000,0.000000}%
\pgfsetstrokecolor{currentstroke}%
\pgfsetdash{}{0pt}%
\pgfsys@defobject{currentmarker}{\pgfqpoint{0.000000in}{0.000000in}}{\pgfqpoint{0.000000in}{0.055556in}}{%
\pgfpathmoveto{\pgfqpoint{0.000000in}{0.000000in}}%
\pgfpathlineto{\pgfqpoint{0.000000in}{0.055556in}}%
\pgfusepath{stroke,fill}%
}%
\begin{pgfscope}%
\pgfsys@transformshift{1.498155in}{0.180665in}%
\pgfsys@useobject{currentmarker}{}%
\end{pgfscope}%
\end{pgfscope}%
\begin{pgfscope}%
\pgfsetbuttcap%
\pgfsetroundjoin%
\definecolor{currentfill}{rgb}{0.000000,0.000000,0.000000}%
\pgfsetfillcolor{currentfill}%
\pgfsetlinewidth{0.501875pt}%
\definecolor{currentstroke}{rgb}{0.000000,0.000000,0.000000}%
\pgfsetstrokecolor{currentstroke}%
\pgfsetdash{}{0pt}%
\pgfsys@defobject{currentmarker}{\pgfqpoint{0.000000in}{-0.055556in}}{\pgfqpoint{0.000000in}{0.000000in}}{%
\pgfpathmoveto{\pgfqpoint{0.000000in}{0.000000in}}%
\pgfpathlineto{\pgfqpoint{0.000000in}{-0.055556in}}%
\pgfusepath{stroke,fill}%
}%
\begin{pgfscope}%
\pgfsys@transformshift{1.498155in}{1.625989in}%
\pgfsys@useobject{currentmarker}{}%
\end{pgfscope}%
\end{pgfscope}%
\begin{pgfscope}%
\pgftext[x=1.498155in,y=0.125110in,,top]{{\rmfamily\fontsize{8.000000}{9.600000}\selectfont \(\displaystyle 2\)}}%
\end{pgfscope}%
\begin{pgfscope}%
\pgfsetbuttcap%
\pgfsetroundjoin%
\definecolor{currentfill}{rgb}{0.000000,0.000000,0.000000}%
\pgfsetfillcolor{currentfill}%
\pgfsetlinewidth{0.501875pt}%
\definecolor{currentstroke}{rgb}{0.000000,0.000000,0.000000}%
\pgfsetstrokecolor{currentstroke}%
\pgfsetdash{}{0pt}%
\pgfsys@defobject{currentmarker}{\pgfqpoint{0.000000in}{0.000000in}}{\pgfqpoint{0.000000in}{0.055556in}}{%
\pgfpathmoveto{\pgfqpoint{0.000000in}{0.000000in}}%
\pgfpathlineto{\pgfqpoint{0.000000in}{0.055556in}}%
\pgfusepath{stroke,fill}%
}%
\begin{pgfscope}%
\pgfsys@transformshift{1.761585in}{0.180665in}%
\pgfsys@useobject{currentmarker}{}%
\end{pgfscope}%
\end{pgfscope}%
\begin{pgfscope}%
\pgfsetbuttcap%
\pgfsetroundjoin%
\definecolor{currentfill}{rgb}{0.000000,0.000000,0.000000}%
\pgfsetfillcolor{currentfill}%
\pgfsetlinewidth{0.501875pt}%
\definecolor{currentstroke}{rgb}{0.000000,0.000000,0.000000}%
\pgfsetstrokecolor{currentstroke}%
\pgfsetdash{}{0pt}%
\pgfsys@defobject{currentmarker}{\pgfqpoint{0.000000in}{-0.055556in}}{\pgfqpoint{0.000000in}{0.000000in}}{%
\pgfpathmoveto{\pgfqpoint{0.000000in}{0.000000in}}%
\pgfpathlineto{\pgfqpoint{0.000000in}{-0.055556in}}%
\pgfusepath{stroke,fill}%
}%
\begin{pgfscope}%
\pgfsys@transformshift{1.761585in}{1.625989in}%
\pgfsys@useobject{currentmarker}{}%
\end{pgfscope}%
\end{pgfscope}%
\begin{pgfscope}%
\pgftext[x=1.761585in,y=0.125110in,,top]{{\rmfamily\fontsize{8.000000}{9.600000}\selectfont \(\displaystyle 3\)}}%
\end{pgfscope}%
\begin{pgfscope}%
\pgfsetbuttcap%
\pgfsetroundjoin%
\definecolor{currentfill}{rgb}{0.000000,0.000000,0.000000}%
\pgfsetfillcolor{currentfill}%
\pgfsetlinewidth{0.501875pt}%
\definecolor{currentstroke}{rgb}{0.000000,0.000000,0.000000}%
\pgfsetstrokecolor{currentstroke}%
\pgfsetdash{}{0pt}%
\pgfsys@defobject{currentmarker}{\pgfqpoint{0.000000in}{0.000000in}}{\pgfqpoint{0.000000in}{0.055556in}}{%
\pgfpathmoveto{\pgfqpoint{0.000000in}{0.000000in}}%
\pgfpathlineto{\pgfqpoint{0.000000in}{0.055556in}}%
\pgfusepath{stroke,fill}%
}%
\begin{pgfscope}%
\pgfsys@transformshift{2.025016in}{0.180665in}%
\pgfsys@useobject{currentmarker}{}%
\end{pgfscope}%
\end{pgfscope}%
\begin{pgfscope}%
\pgfsetbuttcap%
\pgfsetroundjoin%
\definecolor{currentfill}{rgb}{0.000000,0.000000,0.000000}%
\pgfsetfillcolor{currentfill}%
\pgfsetlinewidth{0.501875pt}%
\definecolor{currentstroke}{rgb}{0.000000,0.000000,0.000000}%
\pgfsetstrokecolor{currentstroke}%
\pgfsetdash{}{0pt}%
\pgfsys@defobject{currentmarker}{\pgfqpoint{0.000000in}{-0.055556in}}{\pgfqpoint{0.000000in}{0.000000in}}{%
\pgfpathmoveto{\pgfqpoint{0.000000in}{0.000000in}}%
\pgfpathlineto{\pgfqpoint{0.000000in}{-0.055556in}}%
\pgfusepath{stroke,fill}%
}%
\begin{pgfscope}%
\pgfsys@transformshift{2.025016in}{1.625989in}%
\pgfsys@useobject{currentmarker}{}%
\end{pgfscope}%
\end{pgfscope}%
\begin{pgfscope}%
\pgftext[x=2.025016in,y=0.125110in,,top]{{\rmfamily\fontsize{8.000000}{9.600000}\selectfont \(\displaystyle 4\)}}%
\end{pgfscope}%
\begin{pgfscope}%
\pgfsetbuttcap%
\pgfsetroundjoin%
\definecolor{currentfill}{rgb}{0.000000,0.000000,0.000000}%
\pgfsetfillcolor{currentfill}%
\pgfsetlinewidth{0.501875pt}%
\definecolor{currentstroke}{rgb}{0.000000,0.000000,0.000000}%
\pgfsetstrokecolor{currentstroke}%
\pgfsetdash{}{0pt}%
\pgfsys@defobject{currentmarker}{\pgfqpoint{0.000000in}{0.000000in}}{\pgfqpoint{0.000000in}{0.055556in}}{%
\pgfpathmoveto{\pgfqpoint{0.000000in}{0.000000in}}%
\pgfpathlineto{\pgfqpoint{0.000000in}{0.055556in}}%
\pgfusepath{stroke,fill}%
}%
\begin{pgfscope}%
\pgfsys@transformshift{2.288446in}{0.180665in}%
\pgfsys@useobject{currentmarker}{}%
\end{pgfscope}%
\end{pgfscope}%
\begin{pgfscope}%
\pgfsetbuttcap%
\pgfsetroundjoin%
\definecolor{currentfill}{rgb}{0.000000,0.000000,0.000000}%
\pgfsetfillcolor{currentfill}%
\pgfsetlinewidth{0.501875pt}%
\definecolor{currentstroke}{rgb}{0.000000,0.000000,0.000000}%
\pgfsetstrokecolor{currentstroke}%
\pgfsetdash{}{0pt}%
\pgfsys@defobject{currentmarker}{\pgfqpoint{0.000000in}{-0.055556in}}{\pgfqpoint{0.000000in}{0.000000in}}{%
\pgfpathmoveto{\pgfqpoint{0.000000in}{0.000000in}}%
\pgfpathlineto{\pgfqpoint{0.000000in}{-0.055556in}}%
\pgfusepath{stroke,fill}%
}%
\begin{pgfscope}%
\pgfsys@transformshift{2.288446in}{1.625989in}%
\pgfsys@useobject{currentmarker}{}%
\end{pgfscope}%
\end{pgfscope}%
\begin{pgfscope}%
\pgftext[x=2.288446in,y=0.125110in,,top]{{\rmfamily\fontsize{8.000000}{9.600000}\selectfont \(\displaystyle 5\)}}%
\end{pgfscope}%
\begin{pgfscope}%
\pgfsetbuttcap%
\pgfsetroundjoin%
\definecolor{currentfill}{rgb}{0.000000,0.000000,0.000000}%
\pgfsetfillcolor{currentfill}%
\pgfsetlinewidth{0.501875pt}%
\definecolor{currentstroke}{rgb}{0.000000,0.000000,0.000000}%
\pgfsetstrokecolor{currentstroke}%
\pgfsetdash{}{0pt}%
\pgfsys@defobject{currentmarker}{\pgfqpoint{0.000000in}{0.000000in}}{\pgfqpoint{0.000000in}{0.055556in}}{%
\pgfpathmoveto{\pgfqpoint{0.000000in}{0.000000in}}%
\pgfpathlineto{\pgfqpoint{0.000000in}{0.055556in}}%
\pgfusepath{stroke,fill}%
}%
\begin{pgfscope}%
\pgfsys@transformshift{2.551876in}{0.180665in}%
\pgfsys@useobject{currentmarker}{}%
\end{pgfscope}%
\end{pgfscope}%
\begin{pgfscope}%
\pgfsetbuttcap%
\pgfsetroundjoin%
\definecolor{currentfill}{rgb}{0.000000,0.000000,0.000000}%
\pgfsetfillcolor{currentfill}%
\pgfsetlinewidth{0.501875pt}%
\definecolor{currentstroke}{rgb}{0.000000,0.000000,0.000000}%
\pgfsetstrokecolor{currentstroke}%
\pgfsetdash{}{0pt}%
\pgfsys@defobject{currentmarker}{\pgfqpoint{0.000000in}{-0.055556in}}{\pgfqpoint{0.000000in}{0.000000in}}{%
\pgfpathmoveto{\pgfqpoint{0.000000in}{0.000000in}}%
\pgfpathlineto{\pgfqpoint{0.000000in}{-0.055556in}}%
\pgfusepath{stroke,fill}%
}%
\begin{pgfscope}%
\pgfsys@transformshift{2.551876in}{1.625989in}%
\pgfsys@useobject{currentmarker}{}%
\end{pgfscope}%
\end{pgfscope}%
\begin{pgfscope}%
\pgftext[x=2.551876in,y=0.125110in,,top]{{\rmfamily\fontsize{8.000000}{9.600000}\selectfont \(\displaystyle 6\)}}%
\end{pgfscope}%
\begin{pgfscope}%
\pgftext[x=1.498155in,y=-0.042459in,,top]{{\rmfamily\fontsize{10.000000}{12.000000}\selectfont \(\displaystyle s\)}}%
\end{pgfscope}%
\begin{pgfscope}%
\pgfsetbuttcap%
\pgfsetroundjoin%
\definecolor{currentfill}{rgb}{0.000000,0.000000,0.000000}%
\pgfsetfillcolor{currentfill}%
\pgfsetlinewidth{0.501875pt}%
\definecolor{currentstroke}{rgb}{0.000000,0.000000,0.000000}%
\pgfsetstrokecolor{currentstroke}%
\pgfsetdash{}{0pt}%
\pgfsys@defobject{currentmarker}{\pgfqpoint{0.000000in}{0.000000in}}{\pgfqpoint{0.055556in}{0.000000in}}{%
\pgfpathmoveto{\pgfqpoint{0.000000in}{0.000000in}}%
\pgfpathlineto{\pgfqpoint{0.055556in}{0.000000in}}%
\pgfusepath{stroke,fill}%
}%
\begin{pgfscope}%
\pgfsys@transformshift{0.365404in}{0.283903in}%
\pgfsys@useobject{currentmarker}{}%
\end{pgfscope}%
\end{pgfscope}%
\begin{pgfscope}%
\pgfsetbuttcap%
\pgfsetroundjoin%
\definecolor{currentfill}{rgb}{0.000000,0.000000,0.000000}%
\pgfsetfillcolor{currentfill}%
\pgfsetlinewidth{0.501875pt}%
\definecolor{currentstroke}{rgb}{0.000000,0.000000,0.000000}%
\pgfsetstrokecolor{currentstroke}%
\pgfsetdash{}{0pt}%
\pgfsys@defobject{currentmarker}{\pgfqpoint{-0.055556in}{0.000000in}}{\pgfqpoint{0.000000in}{0.000000in}}{%
\pgfpathmoveto{\pgfqpoint{0.000000in}{0.000000in}}%
\pgfpathlineto{\pgfqpoint{-0.055556in}{0.000000in}}%
\pgfusepath{stroke,fill}%
}%
\begin{pgfscope}%
\pgfsys@transformshift{2.630906in}{0.283903in}%
\pgfsys@useobject{currentmarker}{}%
\end{pgfscope}%
\end{pgfscope}%
\begin{pgfscope}%
\pgftext[x=0.309848in,y=0.283903in,right,]{{\rmfamily\fontsize{8.000000}{9.600000}\selectfont \(\displaystyle 0.0\)}}%
\end{pgfscope}%
\begin{pgfscope}%
\pgfsetbuttcap%
\pgfsetroundjoin%
\definecolor{currentfill}{rgb}{0.000000,0.000000,0.000000}%
\pgfsetfillcolor{currentfill}%
\pgfsetlinewidth{0.501875pt}%
\definecolor{currentstroke}{rgb}{0.000000,0.000000,0.000000}%
\pgfsetstrokecolor{currentstroke}%
\pgfsetdash{}{0pt}%
\pgfsys@defobject{currentmarker}{\pgfqpoint{0.000000in}{0.000000in}}{\pgfqpoint{0.055556in}{0.000000in}}{%
\pgfpathmoveto{\pgfqpoint{0.000000in}{0.000000in}}%
\pgfpathlineto{\pgfqpoint{0.055556in}{0.000000in}}%
\pgfusepath{stroke,fill}%
}%
\begin{pgfscope}%
\pgfsys@transformshift{0.365404in}{0.490378in}%
\pgfsys@useobject{currentmarker}{}%
\end{pgfscope}%
\end{pgfscope}%
\begin{pgfscope}%
\pgfsetbuttcap%
\pgfsetroundjoin%
\definecolor{currentfill}{rgb}{0.000000,0.000000,0.000000}%
\pgfsetfillcolor{currentfill}%
\pgfsetlinewidth{0.501875pt}%
\definecolor{currentstroke}{rgb}{0.000000,0.000000,0.000000}%
\pgfsetstrokecolor{currentstroke}%
\pgfsetdash{}{0pt}%
\pgfsys@defobject{currentmarker}{\pgfqpoint{-0.055556in}{0.000000in}}{\pgfqpoint{0.000000in}{0.000000in}}{%
\pgfpathmoveto{\pgfqpoint{0.000000in}{0.000000in}}%
\pgfpathlineto{\pgfqpoint{-0.055556in}{0.000000in}}%
\pgfusepath{stroke,fill}%
}%
\begin{pgfscope}%
\pgfsys@transformshift{2.630906in}{0.490378in}%
\pgfsys@useobject{currentmarker}{}%
\end{pgfscope}%
\end{pgfscope}%
\begin{pgfscope}%
\pgftext[x=0.309848in,y=0.490378in,right,]{{\rmfamily\fontsize{8.000000}{9.600000}\selectfont \(\displaystyle 0.1\)}}%
\end{pgfscope}%
\begin{pgfscope}%
\pgfsetbuttcap%
\pgfsetroundjoin%
\definecolor{currentfill}{rgb}{0.000000,0.000000,0.000000}%
\pgfsetfillcolor{currentfill}%
\pgfsetlinewidth{0.501875pt}%
\definecolor{currentstroke}{rgb}{0.000000,0.000000,0.000000}%
\pgfsetstrokecolor{currentstroke}%
\pgfsetdash{}{0pt}%
\pgfsys@defobject{currentmarker}{\pgfqpoint{0.000000in}{0.000000in}}{\pgfqpoint{0.055556in}{0.000000in}}{%
\pgfpathmoveto{\pgfqpoint{0.000000in}{0.000000in}}%
\pgfpathlineto{\pgfqpoint{0.055556in}{0.000000in}}%
\pgfusepath{stroke,fill}%
}%
\begin{pgfscope}%
\pgfsys@transformshift{0.365404in}{0.696852in}%
\pgfsys@useobject{currentmarker}{}%
\end{pgfscope}%
\end{pgfscope}%
\begin{pgfscope}%
\pgfsetbuttcap%
\pgfsetroundjoin%
\definecolor{currentfill}{rgb}{0.000000,0.000000,0.000000}%
\pgfsetfillcolor{currentfill}%
\pgfsetlinewidth{0.501875pt}%
\definecolor{currentstroke}{rgb}{0.000000,0.000000,0.000000}%
\pgfsetstrokecolor{currentstroke}%
\pgfsetdash{}{0pt}%
\pgfsys@defobject{currentmarker}{\pgfqpoint{-0.055556in}{0.000000in}}{\pgfqpoint{0.000000in}{0.000000in}}{%
\pgfpathmoveto{\pgfqpoint{0.000000in}{0.000000in}}%
\pgfpathlineto{\pgfqpoint{-0.055556in}{0.000000in}}%
\pgfusepath{stroke,fill}%
}%
\begin{pgfscope}%
\pgfsys@transformshift{2.630906in}{0.696852in}%
\pgfsys@useobject{currentmarker}{}%
\end{pgfscope}%
\end{pgfscope}%
\begin{pgfscope}%
\pgftext[x=0.309848in,y=0.696852in,right,]{{\rmfamily\fontsize{8.000000}{9.600000}\selectfont \(\displaystyle 0.2\)}}%
\end{pgfscope}%
\begin{pgfscope}%
\pgfsetbuttcap%
\pgfsetroundjoin%
\definecolor{currentfill}{rgb}{0.000000,0.000000,0.000000}%
\pgfsetfillcolor{currentfill}%
\pgfsetlinewidth{0.501875pt}%
\definecolor{currentstroke}{rgb}{0.000000,0.000000,0.000000}%
\pgfsetstrokecolor{currentstroke}%
\pgfsetdash{}{0pt}%
\pgfsys@defobject{currentmarker}{\pgfqpoint{0.000000in}{0.000000in}}{\pgfqpoint{0.055556in}{0.000000in}}{%
\pgfpathmoveto{\pgfqpoint{0.000000in}{0.000000in}}%
\pgfpathlineto{\pgfqpoint{0.055556in}{0.000000in}}%
\pgfusepath{stroke,fill}%
}%
\begin{pgfscope}%
\pgfsys@transformshift{0.365404in}{0.903327in}%
\pgfsys@useobject{currentmarker}{}%
\end{pgfscope}%
\end{pgfscope}%
\begin{pgfscope}%
\pgfsetbuttcap%
\pgfsetroundjoin%
\definecolor{currentfill}{rgb}{0.000000,0.000000,0.000000}%
\pgfsetfillcolor{currentfill}%
\pgfsetlinewidth{0.501875pt}%
\definecolor{currentstroke}{rgb}{0.000000,0.000000,0.000000}%
\pgfsetstrokecolor{currentstroke}%
\pgfsetdash{}{0pt}%
\pgfsys@defobject{currentmarker}{\pgfqpoint{-0.055556in}{0.000000in}}{\pgfqpoint{0.000000in}{0.000000in}}{%
\pgfpathmoveto{\pgfqpoint{0.000000in}{0.000000in}}%
\pgfpathlineto{\pgfqpoint{-0.055556in}{0.000000in}}%
\pgfusepath{stroke,fill}%
}%
\begin{pgfscope}%
\pgfsys@transformshift{2.630906in}{0.903327in}%
\pgfsys@useobject{currentmarker}{}%
\end{pgfscope}%
\end{pgfscope}%
\begin{pgfscope}%
\pgftext[x=0.309848in,y=0.903327in,right,]{{\rmfamily\fontsize{8.000000}{9.600000}\selectfont \(\displaystyle 0.3\)}}%
\end{pgfscope}%
\begin{pgfscope}%
\pgfsetbuttcap%
\pgfsetroundjoin%
\definecolor{currentfill}{rgb}{0.000000,0.000000,0.000000}%
\pgfsetfillcolor{currentfill}%
\pgfsetlinewidth{0.501875pt}%
\definecolor{currentstroke}{rgb}{0.000000,0.000000,0.000000}%
\pgfsetstrokecolor{currentstroke}%
\pgfsetdash{}{0pt}%
\pgfsys@defobject{currentmarker}{\pgfqpoint{0.000000in}{0.000000in}}{\pgfqpoint{0.055556in}{0.000000in}}{%
\pgfpathmoveto{\pgfqpoint{0.000000in}{0.000000in}}%
\pgfpathlineto{\pgfqpoint{0.055556in}{0.000000in}}%
\pgfusepath{stroke,fill}%
}%
\begin{pgfscope}%
\pgfsys@transformshift{0.365404in}{1.109802in}%
\pgfsys@useobject{currentmarker}{}%
\end{pgfscope}%
\end{pgfscope}%
\begin{pgfscope}%
\pgfsetbuttcap%
\pgfsetroundjoin%
\definecolor{currentfill}{rgb}{0.000000,0.000000,0.000000}%
\pgfsetfillcolor{currentfill}%
\pgfsetlinewidth{0.501875pt}%
\definecolor{currentstroke}{rgb}{0.000000,0.000000,0.000000}%
\pgfsetstrokecolor{currentstroke}%
\pgfsetdash{}{0pt}%
\pgfsys@defobject{currentmarker}{\pgfqpoint{-0.055556in}{0.000000in}}{\pgfqpoint{0.000000in}{0.000000in}}{%
\pgfpathmoveto{\pgfqpoint{0.000000in}{0.000000in}}%
\pgfpathlineto{\pgfqpoint{-0.055556in}{0.000000in}}%
\pgfusepath{stroke,fill}%
}%
\begin{pgfscope}%
\pgfsys@transformshift{2.630906in}{1.109802in}%
\pgfsys@useobject{currentmarker}{}%
\end{pgfscope}%
\end{pgfscope}%
\begin{pgfscope}%
\pgftext[x=0.309848in,y=1.109802in,right,]{{\rmfamily\fontsize{8.000000}{9.600000}\selectfont \(\displaystyle 0.4\)}}%
\end{pgfscope}%
\begin{pgfscope}%
\pgfsetbuttcap%
\pgfsetroundjoin%
\definecolor{currentfill}{rgb}{0.000000,0.000000,0.000000}%
\pgfsetfillcolor{currentfill}%
\pgfsetlinewidth{0.501875pt}%
\definecolor{currentstroke}{rgb}{0.000000,0.000000,0.000000}%
\pgfsetstrokecolor{currentstroke}%
\pgfsetdash{}{0pt}%
\pgfsys@defobject{currentmarker}{\pgfqpoint{0.000000in}{0.000000in}}{\pgfqpoint{0.055556in}{0.000000in}}{%
\pgfpathmoveto{\pgfqpoint{0.000000in}{0.000000in}}%
\pgfpathlineto{\pgfqpoint{0.055556in}{0.000000in}}%
\pgfusepath{stroke,fill}%
}%
\begin{pgfscope}%
\pgfsys@transformshift{0.365404in}{1.316277in}%
\pgfsys@useobject{currentmarker}{}%
\end{pgfscope}%
\end{pgfscope}%
\begin{pgfscope}%
\pgfsetbuttcap%
\pgfsetroundjoin%
\definecolor{currentfill}{rgb}{0.000000,0.000000,0.000000}%
\pgfsetfillcolor{currentfill}%
\pgfsetlinewidth{0.501875pt}%
\definecolor{currentstroke}{rgb}{0.000000,0.000000,0.000000}%
\pgfsetstrokecolor{currentstroke}%
\pgfsetdash{}{0pt}%
\pgfsys@defobject{currentmarker}{\pgfqpoint{-0.055556in}{0.000000in}}{\pgfqpoint{0.000000in}{0.000000in}}{%
\pgfpathmoveto{\pgfqpoint{0.000000in}{0.000000in}}%
\pgfpathlineto{\pgfqpoint{-0.055556in}{0.000000in}}%
\pgfusepath{stroke,fill}%
}%
\begin{pgfscope}%
\pgfsys@transformshift{2.630906in}{1.316277in}%
\pgfsys@useobject{currentmarker}{}%
\end{pgfscope}%
\end{pgfscope}%
\begin{pgfscope}%
\pgftext[x=0.309848in,y=1.316277in,right,]{{\rmfamily\fontsize{8.000000}{9.600000}\selectfont \(\displaystyle 0.5\)}}%
\end{pgfscope}%
\begin{pgfscope}%
\pgfsetbuttcap%
\pgfsetroundjoin%
\definecolor{currentfill}{rgb}{0.000000,0.000000,0.000000}%
\pgfsetfillcolor{currentfill}%
\pgfsetlinewidth{0.501875pt}%
\definecolor{currentstroke}{rgb}{0.000000,0.000000,0.000000}%
\pgfsetstrokecolor{currentstroke}%
\pgfsetdash{}{0pt}%
\pgfsys@defobject{currentmarker}{\pgfqpoint{0.000000in}{0.000000in}}{\pgfqpoint{0.055556in}{0.000000in}}{%
\pgfpathmoveto{\pgfqpoint{0.000000in}{0.000000in}}%
\pgfpathlineto{\pgfqpoint{0.055556in}{0.000000in}}%
\pgfusepath{stroke,fill}%
}%
\begin{pgfscope}%
\pgfsys@transformshift{0.365404in}{1.522752in}%
\pgfsys@useobject{currentmarker}{}%
\end{pgfscope}%
\end{pgfscope}%
\begin{pgfscope}%
\pgfsetbuttcap%
\pgfsetroundjoin%
\definecolor{currentfill}{rgb}{0.000000,0.000000,0.000000}%
\pgfsetfillcolor{currentfill}%
\pgfsetlinewidth{0.501875pt}%
\definecolor{currentstroke}{rgb}{0.000000,0.000000,0.000000}%
\pgfsetstrokecolor{currentstroke}%
\pgfsetdash{}{0pt}%
\pgfsys@defobject{currentmarker}{\pgfqpoint{-0.055556in}{0.000000in}}{\pgfqpoint{0.000000in}{0.000000in}}{%
\pgfpathmoveto{\pgfqpoint{0.000000in}{0.000000in}}%
\pgfpathlineto{\pgfqpoint{-0.055556in}{0.000000in}}%
\pgfusepath{stroke,fill}%
}%
\begin{pgfscope}%
\pgfsys@transformshift{2.630906in}{1.522752in}%
\pgfsys@useobject{currentmarker}{}%
\end{pgfscope}%
\end{pgfscope}%
\begin{pgfscope}%
\pgftext[x=0.309848in,y=1.522752in,right,]{{\rmfamily\fontsize{8.000000}{9.600000}\selectfont \(\displaystyle 0.6\)}}%
\end{pgfscope}%
\begin{pgfscope}%
\pgftext[x=0.089553in,y=0.903327in,,bottom,rotate=90.000000]{{\rmfamily\fontsize{10.000000}{12.000000}\selectfont \(\displaystyle p(s)\)}}%
\end{pgfscope}%
\begin{pgfscope}%
\pgfsetbuttcap%
\pgfsetroundjoin%
\pgfsetlinewidth{1.003750pt}%
\definecolor{currentstroke}{rgb}{0.000000,0.000000,0.000000}%
\pgfsetstrokecolor{currentstroke}%
\pgfsetdash{}{0pt}%
\pgfpathmoveto{\pgfqpoint{0.365404in}{1.625989in}}%
\pgfpathlineto{\pgfqpoint{2.630906in}{1.625989in}}%
\pgfusepath{stroke}%
\end{pgfscope}%
\begin{pgfscope}%
\pgfsetbuttcap%
\pgfsetroundjoin%
\pgfsetlinewidth{1.003750pt}%
\definecolor{currentstroke}{rgb}{0.000000,0.000000,0.000000}%
\pgfsetstrokecolor{currentstroke}%
\pgfsetdash{}{0pt}%
\pgfpathmoveto{\pgfqpoint{2.630906in}{0.180665in}}%
\pgfpathlineto{\pgfqpoint{2.630906in}{1.625989in}}%
\pgfusepath{stroke}%
\end{pgfscope}%
\begin{pgfscope}%
\pgfsetbuttcap%
\pgfsetroundjoin%
\pgfsetlinewidth{1.003750pt}%
\definecolor{currentstroke}{rgb}{0.000000,0.000000,0.000000}%
\pgfsetstrokecolor{currentstroke}%
\pgfsetdash{}{0pt}%
\pgfpathmoveto{\pgfqpoint{0.365404in}{0.180665in}}%
\pgfpathlineto{\pgfqpoint{2.630906in}{0.180665in}}%
\pgfusepath{stroke}%
\end{pgfscope}%
\begin{pgfscope}%
\pgfsetbuttcap%
\pgfsetroundjoin%
\pgfsetlinewidth{1.003750pt}%
\definecolor{currentstroke}{rgb}{0.000000,0.000000,0.000000}%
\pgfsetstrokecolor{currentstroke}%
\pgfsetdash{}{0pt}%
\pgfpathmoveto{\pgfqpoint{0.365404in}{0.180665in}}%
\pgfpathlineto{\pgfqpoint{0.365404in}{1.625989in}}%
\pgfusepath{stroke}%
\end{pgfscope}%
\begin{pgfscope}%
\pgftext[x=1.498155in,y=1.481457in,left,top]{{\rmfamily\fontsize{12.000000}{14.400000}\selectfont \(\displaystyle q=1.5\)}}%
\end{pgfscope}%
\begin{pgfscope}%
\pgftext[x=1.498155in,y=1.336924in,left,top]{{\rmfamily\fontsize{12.000000}{14.400000}\selectfont \(\displaystyle p=0.2\)}}%
\end{pgfscope}%
\begin{pgfscope}%
\pgftext[x=1.498155in,y=1.192392in,left,top]{{\rmfamily\fontsize{12.000000}{14.400000}\selectfont \(\displaystyle L=64, d=3\)}}%
\end{pgfscope}%
\begin{pgfscope}%
\pgfpathrectangle{\pgfqpoint{0.365404in}{0.180665in}}{\pgfqpoint{2.265502in}{1.445324in}} %
\pgfusepath{clip}%
\pgfsetbuttcap%
\pgfsetroundjoin%
\pgfsetlinewidth{1.003750pt}%
\definecolor{currentstroke}{rgb}{0.000000,0.000000,0.000000}%
\pgfsetstrokecolor{currentstroke}%
\pgfsetdash{}{0pt}%
\pgfpathmoveto{\pgfqpoint{0.378034in}{0.283903in}}%
\pgfpathlineto{\pgfqpoint{0.378034in}{0.287924in}}%
\pgfpathlineto{\pgfqpoint{0.414928in}{0.287924in}}%
\pgfpathlineto{\pgfqpoint{0.414928in}{0.283903in}}%
\pgfpathlineto{\pgfqpoint{0.427226in}{0.283903in}}%
\pgfpathlineto{\pgfqpoint{0.427226in}{0.295965in}}%
\pgfpathlineto{\pgfqpoint{0.439524in}{0.295965in}}%
\pgfpathlineto{\pgfqpoint{0.439524in}{0.299986in}}%
\pgfpathlineto{\pgfqpoint{0.451823in}{0.299986in}}%
\pgfpathlineto{\pgfqpoint{0.451823in}{0.308027in}}%
\pgfpathlineto{\pgfqpoint{0.464121in}{0.308027in}}%
\pgfpathlineto{\pgfqpoint{0.464121in}{0.295965in}}%
\pgfpathlineto{\pgfqpoint{0.476419in}{0.295965in}}%
\pgfpathlineto{\pgfqpoint{0.476419in}{0.316068in}}%
\pgfpathlineto{\pgfqpoint{0.488717in}{0.316068in}}%
\pgfpathlineto{\pgfqpoint{0.488717in}{0.348234in}}%
\pgfpathlineto{\pgfqpoint{0.501015in}{0.348234in}}%
\pgfpathlineto{\pgfqpoint{0.501015in}{0.364317in}}%
\pgfpathlineto{\pgfqpoint{0.513313in}{0.364317in}}%
\pgfpathlineto{\pgfqpoint{0.513313in}{0.372358in}}%
\pgfpathlineto{\pgfqpoint{0.525612in}{0.372358in}}%
\pgfpathlineto{\pgfqpoint{0.525612in}{0.380399in}}%
\pgfpathlineto{\pgfqpoint{0.537910in}{0.380399in}}%
\pgfpathlineto{\pgfqpoint{0.537910in}{0.412565in}}%
\pgfpathlineto{\pgfqpoint{0.562506in}{0.412565in}}%
\pgfpathlineto{\pgfqpoint{0.562506in}{0.480916in}}%
\pgfpathlineto{\pgfqpoint{0.574804in}{0.480916in}}%
\pgfpathlineto{\pgfqpoint{0.574804in}{0.581434in}}%
\pgfpathlineto{\pgfqpoint{0.587102in}{0.581434in}}%
\pgfpathlineto{\pgfqpoint{0.587102in}{0.553289in}}%
\pgfpathlineto{\pgfqpoint{0.599401in}{0.553289in}}%
\pgfpathlineto{\pgfqpoint{0.599401in}{0.561330in}}%
\pgfpathlineto{\pgfqpoint{0.611699in}{0.561330in}}%
\pgfpathlineto{\pgfqpoint{0.611699in}{0.585454in}}%
\pgfpathlineto{\pgfqpoint{0.623997in}{0.585454in}}%
\pgfpathlineto{\pgfqpoint{0.623997in}{0.714116in}}%
\pgfpathlineto{\pgfqpoint{0.636295in}{0.714116in}}%
\pgfpathlineto{\pgfqpoint{0.636295in}{0.794530in}}%
\pgfpathlineto{\pgfqpoint{0.648593in}{0.794530in}}%
\pgfpathlineto{\pgfqpoint{0.648593in}{0.766385in}}%
\pgfpathlineto{\pgfqpoint{0.660891in}{0.766385in}}%
\pgfpathlineto{\pgfqpoint{0.660891in}{0.850819in}}%
\pgfpathlineto{\pgfqpoint{0.673190in}{0.850819in}}%
\pgfpathlineto{\pgfqpoint{0.673190in}{0.882985in}}%
\pgfpathlineto{\pgfqpoint{0.685488in}{0.882985in}}%
\pgfpathlineto{\pgfqpoint{0.685488in}{0.854840in}}%
\pgfpathlineto{\pgfqpoint{0.697786in}{0.854840in}}%
\pgfpathlineto{\pgfqpoint{0.697786in}{1.092060in}}%
\pgfpathlineto{\pgfqpoint{0.710084in}{1.092060in}}%
\pgfpathlineto{\pgfqpoint{0.710084in}{1.031750in}}%
\pgfpathlineto{\pgfqpoint{0.722382in}{1.031750in}}%
\pgfpathlineto{\pgfqpoint{0.722382in}{1.088040in}}%
\pgfpathlineto{\pgfqpoint{0.734680in}{1.088040in}}%
\pgfpathlineto{\pgfqpoint{0.734680in}{1.140309in}}%
\pgfpathlineto{\pgfqpoint{0.746979in}{1.140309in}}%
\pgfpathlineto{\pgfqpoint{0.746979in}{1.228764in}}%
\pgfpathlineto{\pgfqpoint{0.759277in}{1.228764in}}%
\pgfpathlineto{\pgfqpoint{0.759277in}{1.192578in}}%
\pgfpathlineto{\pgfqpoint{0.771575in}{1.192578in}}%
\pgfpathlineto{\pgfqpoint{0.771575in}{1.120205in}}%
\pgfpathlineto{\pgfqpoint{0.783873in}{1.120205in}}%
\pgfpathlineto{\pgfqpoint{0.783873in}{1.212681in}}%
\pgfpathlineto{\pgfqpoint{0.796171in}{1.212681in}}%
\pgfpathlineto{\pgfqpoint{0.796171in}{1.208660in}}%
\pgfpathlineto{\pgfqpoint{0.808469in}{1.208660in}}%
\pgfpathlineto{\pgfqpoint{0.808469in}{1.309177in}}%
\pgfpathlineto{\pgfqpoint{0.820768in}{1.309177in}}%
\pgfpathlineto{\pgfqpoint{0.820768in}{1.281033in}}%
\pgfpathlineto{\pgfqpoint{0.833066in}{1.281033in}}%
\pgfpathlineto{\pgfqpoint{0.833066in}{1.333302in}}%
\pgfpathlineto{\pgfqpoint{0.845364in}{1.333302in}}%
\pgfpathlineto{\pgfqpoint{0.845364in}{1.293095in}}%
\pgfpathlineto{\pgfqpoint{0.857662in}{1.293095in}}%
\pgfpathlineto{\pgfqpoint{0.857662in}{1.268971in}}%
\pgfpathlineto{\pgfqpoint{0.869960in}{1.268971in}}%
\pgfpathlineto{\pgfqpoint{0.869960in}{1.277012in}}%
\pgfpathlineto{\pgfqpoint{0.882258in}{1.277012in}}%
\pgfpathlineto{\pgfqpoint{0.882258in}{1.204640in}}%
\pgfpathlineto{\pgfqpoint{0.894557in}{1.204640in}}%
\pgfpathlineto{\pgfqpoint{0.894557in}{1.168453in}}%
\pgfpathlineto{\pgfqpoint{0.906855in}{1.168453in}}%
\pgfpathlineto{\pgfqpoint{0.906855in}{1.084019in}}%
\pgfpathlineto{\pgfqpoint{0.919153in}{1.084019in}}%
\pgfpathlineto{\pgfqpoint{0.919153in}{1.236805in}}%
\pgfpathlineto{\pgfqpoint{0.931451in}{1.236805in}}%
\pgfpathlineto{\pgfqpoint{0.931451in}{1.128247in}}%
\pgfpathlineto{\pgfqpoint{0.943749in}{1.128247in}}%
\pgfpathlineto{\pgfqpoint{0.943749in}{1.055874in}}%
\pgfpathlineto{\pgfqpoint{0.956047in}{1.055874in}}%
\pgfpathlineto{\pgfqpoint{0.956047in}{1.007626in}}%
\pgfpathlineto{\pgfqpoint{0.968346in}{1.007626in}}%
\pgfpathlineto{\pgfqpoint{0.968346in}{1.039792in}}%
\pgfpathlineto{\pgfqpoint{0.980644in}{1.039792in}}%
\pgfpathlineto{\pgfqpoint{0.980644in}{1.160412in}}%
\pgfpathlineto{\pgfqpoint{0.992942in}{1.160412in}}%
\pgfpathlineto{\pgfqpoint{0.992942in}{1.019688in}}%
\pgfpathlineto{\pgfqpoint{1.005240in}{1.019688in}}%
\pgfpathlineto{\pgfqpoint{1.005240in}{0.963399in}}%
\pgfpathlineto{\pgfqpoint{1.017538in}{0.963399in}}%
\pgfpathlineto{\pgfqpoint{1.017538in}{0.987523in}}%
\pgfpathlineto{\pgfqpoint{1.029836in}{0.987523in}}%
\pgfpathlineto{\pgfqpoint{1.029836in}{0.971440in}}%
\pgfpathlineto{\pgfqpoint{1.042135in}{0.971440in}}%
\pgfpathlineto{\pgfqpoint{1.042135in}{1.023709in}}%
\pgfpathlineto{\pgfqpoint{1.054433in}{1.023709in}}%
\pgfpathlineto{\pgfqpoint{1.054433in}{0.955357in}}%
\pgfpathlineto{\pgfqpoint{1.079029in}{0.955357in}}%
\pgfpathlineto{\pgfqpoint{1.079029in}{0.882985in}}%
\pgfpathlineto{\pgfqpoint{1.091327in}{0.882985in}}%
\pgfpathlineto{\pgfqpoint{1.091327in}{0.870923in}}%
\pgfpathlineto{\pgfqpoint{1.103625in}{0.870923in}}%
\pgfpathlineto{\pgfqpoint{1.103625in}{0.887006in}}%
\pgfpathlineto{\pgfqpoint{1.115924in}{0.887006in}}%
\pgfpathlineto{\pgfqpoint{1.115924in}{0.874943in}}%
\pgfpathlineto{\pgfqpoint{1.128222in}{0.874943in}}%
\pgfpathlineto{\pgfqpoint{1.128222in}{0.826695in}}%
\pgfpathlineto{\pgfqpoint{1.140520in}{0.826695in}}%
\pgfpathlineto{\pgfqpoint{1.140520in}{0.798550in}}%
\pgfpathlineto{\pgfqpoint{1.165116in}{0.798550in}}%
\pgfpathlineto{\pgfqpoint{1.165116in}{0.754323in}}%
\pgfpathlineto{\pgfqpoint{1.177414in}{0.754323in}}%
\pgfpathlineto{\pgfqpoint{1.177414in}{0.798550in}}%
\pgfpathlineto{\pgfqpoint{1.189713in}{0.798550in}}%
\pgfpathlineto{\pgfqpoint{1.189713in}{0.778447in}}%
\pgfpathlineto{\pgfqpoint{1.202011in}{0.778447in}}%
\pgfpathlineto{\pgfqpoint{1.202011in}{0.561330in}}%
\pgfpathlineto{\pgfqpoint{1.214309in}{0.561330in}}%
\pgfpathlineto{\pgfqpoint{1.214309in}{0.637723in}}%
\pgfpathlineto{\pgfqpoint{1.226607in}{0.637723in}}%
\pgfpathlineto{\pgfqpoint{1.226607in}{0.593496in}}%
\pgfpathlineto{\pgfqpoint{1.238905in}{0.593496in}}%
\pgfpathlineto{\pgfqpoint{1.238905in}{0.726178in}}%
\pgfpathlineto{\pgfqpoint{1.251203in}{0.726178in}}%
\pgfpathlineto{\pgfqpoint{1.251203in}{0.625661in}}%
\pgfpathlineto{\pgfqpoint{1.263502in}{0.625661in}}%
\pgfpathlineto{\pgfqpoint{1.263502in}{0.577413in}}%
\pgfpathlineto{\pgfqpoint{1.275800in}{0.577413in}}%
\pgfpathlineto{\pgfqpoint{1.275800in}{0.605558in}}%
\pgfpathlineto{\pgfqpoint{1.288098in}{0.605558in}}%
\pgfpathlineto{\pgfqpoint{1.288098in}{0.589475in}}%
\pgfpathlineto{\pgfqpoint{1.300396in}{0.589475in}}%
\pgfpathlineto{\pgfqpoint{1.300396in}{0.533185in}}%
\pgfpathlineto{\pgfqpoint{1.324992in}{0.533185in}}%
\pgfpathlineto{\pgfqpoint{1.324992in}{0.492978in}}%
\pgfpathlineto{\pgfqpoint{1.337291in}{0.492978in}}%
\pgfpathlineto{\pgfqpoint{1.337291in}{0.456792in}}%
\pgfpathlineto{\pgfqpoint{1.349589in}{0.456792in}}%
\pgfpathlineto{\pgfqpoint{1.349589in}{0.496999in}}%
\pgfpathlineto{\pgfqpoint{1.361887in}{0.496999in}}%
\pgfpathlineto{\pgfqpoint{1.361887in}{0.472875in}}%
\pgfpathlineto{\pgfqpoint{1.374185in}{0.472875in}}%
\pgfpathlineto{\pgfqpoint{1.374185in}{0.476896in}}%
\pgfpathlineto{\pgfqpoint{1.386483in}{0.476896in}}%
\pgfpathlineto{\pgfqpoint{1.386483in}{0.492978in}}%
\pgfpathlineto{\pgfqpoint{1.398781in}{0.492978in}}%
\pgfpathlineto{\pgfqpoint{1.398781in}{0.464834in}}%
\pgfpathlineto{\pgfqpoint{1.411080in}{0.464834in}}%
\pgfpathlineto{\pgfqpoint{1.411080in}{0.440710in}}%
\pgfpathlineto{\pgfqpoint{1.423378in}{0.440710in}}%
\pgfpathlineto{\pgfqpoint{1.423378in}{0.416585in}}%
\pgfpathlineto{\pgfqpoint{1.435676in}{0.416585in}}%
\pgfpathlineto{\pgfqpoint{1.435676in}{0.464834in}}%
\pgfpathlineto{\pgfqpoint{1.447974in}{0.464834in}}%
\pgfpathlineto{\pgfqpoint{1.447974in}{0.396482in}}%
\pgfpathlineto{\pgfqpoint{1.460272in}{0.396482in}}%
\pgfpathlineto{\pgfqpoint{1.460272in}{0.424627in}}%
\pgfpathlineto{\pgfqpoint{1.472570in}{0.424627in}}%
\pgfpathlineto{\pgfqpoint{1.472570in}{0.404523in}}%
\pgfpathlineto{\pgfqpoint{1.484869in}{0.404523in}}%
\pgfpathlineto{\pgfqpoint{1.484869in}{0.416585in}}%
\pgfpathlineto{\pgfqpoint{1.497167in}{0.416585in}}%
\pgfpathlineto{\pgfqpoint{1.497167in}{0.384420in}}%
\pgfpathlineto{\pgfqpoint{1.509465in}{0.384420in}}%
\pgfpathlineto{\pgfqpoint{1.509465in}{0.404523in}}%
\pgfpathlineto{\pgfqpoint{1.521763in}{0.404523in}}%
\pgfpathlineto{\pgfqpoint{1.521763in}{0.356275in}}%
\pgfpathlineto{\pgfqpoint{1.534061in}{0.356275in}}%
\pgfpathlineto{\pgfqpoint{1.534061in}{0.380399in}}%
\pgfpathlineto{\pgfqpoint{1.546359in}{0.380399in}}%
\pgfpathlineto{\pgfqpoint{1.546359in}{0.364317in}}%
\pgfpathlineto{\pgfqpoint{1.558658in}{0.364317in}}%
\pgfpathlineto{\pgfqpoint{1.558658in}{0.340192in}}%
\pgfpathlineto{\pgfqpoint{1.570956in}{0.340192in}}%
\pgfpathlineto{\pgfqpoint{1.570956in}{0.356275in}}%
\pgfpathlineto{\pgfqpoint{1.583254in}{0.356275in}}%
\pgfpathlineto{\pgfqpoint{1.583254in}{0.348234in}}%
\pgfpathlineto{\pgfqpoint{1.595552in}{0.348234in}}%
\pgfpathlineto{\pgfqpoint{1.595552in}{0.372358in}}%
\pgfpathlineto{\pgfqpoint{1.607850in}{0.372358in}}%
\pgfpathlineto{\pgfqpoint{1.607850in}{0.356275in}}%
\pgfpathlineto{\pgfqpoint{1.620148in}{0.356275in}}%
\pgfpathlineto{\pgfqpoint{1.620148in}{0.328130in}}%
\pgfpathlineto{\pgfqpoint{1.632447in}{0.328130in}}%
\pgfpathlineto{\pgfqpoint{1.632447in}{0.340192in}}%
\pgfpathlineto{\pgfqpoint{1.644745in}{0.340192in}}%
\pgfpathlineto{\pgfqpoint{1.644745in}{0.336172in}}%
\pgfpathlineto{\pgfqpoint{1.657043in}{0.336172in}}%
\pgfpathlineto{\pgfqpoint{1.657043in}{0.299986in}}%
\pgfpathlineto{\pgfqpoint{1.669341in}{0.299986in}}%
\pgfpathlineto{\pgfqpoint{1.669341in}{0.384420in}}%
\pgfpathlineto{\pgfqpoint{1.681639in}{0.384420in}}%
\pgfpathlineto{\pgfqpoint{1.681639in}{0.336172in}}%
\pgfpathlineto{\pgfqpoint{1.693937in}{0.336172in}}%
\pgfpathlineto{\pgfqpoint{1.693937in}{0.348234in}}%
\pgfpathlineto{\pgfqpoint{1.706236in}{0.348234in}}%
\pgfpathlineto{\pgfqpoint{1.706236in}{0.328130in}}%
\pgfpathlineto{\pgfqpoint{1.718534in}{0.328130in}}%
\pgfpathlineto{\pgfqpoint{1.718534in}{0.308027in}}%
\pgfpathlineto{\pgfqpoint{1.730832in}{0.308027in}}%
\pgfpathlineto{\pgfqpoint{1.730832in}{0.312048in}}%
\pgfpathlineto{\pgfqpoint{1.743130in}{0.312048in}}%
\pgfpathlineto{\pgfqpoint{1.743130in}{0.320089in}}%
\pgfpathlineto{\pgfqpoint{1.767726in}{0.320089in}}%
\pgfpathlineto{\pgfqpoint{1.767726in}{0.304006in}}%
\pgfpathlineto{\pgfqpoint{1.780025in}{0.304006in}}%
\pgfpathlineto{\pgfqpoint{1.780025in}{0.316068in}}%
\pgfpathlineto{\pgfqpoint{1.792323in}{0.316068in}}%
\pgfpathlineto{\pgfqpoint{1.792323in}{0.304006in}}%
\pgfpathlineto{\pgfqpoint{1.816919in}{0.304006in}}%
\pgfpathlineto{\pgfqpoint{1.816919in}{0.312048in}}%
\pgfpathlineto{\pgfqpoint{1.829217in}{0.312048in}}%
\pgfpathlineto{\pgfqpoint{1.829217in}{0.308027in}}%
\pgfpathlineto{\pgfqpoint{1.841515in}{0.308027in}}%
\pgfpathlineto{\pgfqpoint{1.841515in}{0.291944in}}%
\pgfpathlineto{\pgfqpoint{1.853814in}{0.291944in}}%
\pgfpathlineto{\pgfqpoint{1.853814in}{0.308027in}}%
\pgfpathlineto{\pgfqpoint{1.866112in}{0.308027in}}%
\pgfpathlineto{\pgfqpoint{1.866112in}{0.320089in}}%
\pgfpathlineto{\pgfqpoint{1.878410in}{0.320089in}}%
\pgfpathlineto{\pgfqpoint{1.878410in}{0.299986in}}%
\pgfpathlineto{\pgfqpoint{1.890708in}{0.299986in}}%
\pgfpathlineto{\pgfqpoint{1.890708in}{0.291944in}}%
\pgfpathlineto{\pgfqpoint{1.927603in}{0.291944in}}%
\pgfpathlineto{\pgfqpoint{1.927603in}{0.312048in}}%
\pgfpathlineto{\pgfqpoint{1.939901in}{0.312048in}}%
\pgfpathlineto{\pgfqpoint{1.939901in}{0.287924in}}%
\pgfpathlineto{\pgfqpoint{1.952199in}{0.287924in}}%
\pgfpathlineto{\pgfqpoint{1.952199in}{0.291944in}}%
\pgfpathlineto{\pgfqpoint{1.964497in}{0.291944in}}%
\pgfpathlineto{\pgfqpoint{1.964497in}{0.312048in}}%
\pgfpathlineto{\pgfqpoint{1.976795in}{0.312048in}}%
\pgfpathlineto{\pgfqpoint{1.976795in}{0.304006in}}%
\pgfpathlineto{\pgfqpoint{1.989093in}{0.304006in}}%
\pgfpathlineto{\pgfqpoint{1.989093in}{0.283903in}}%
\pgfpathlineto{\pgfqpoint{2.001392in}{0.283903in}}%
\pgfpathlineto{\pgfqpoint{2.001392in}{0.299986in}}%
\pgfpathlineto{\pgfqpoint{2.013690in}{0.299986in}}%
\pgfpathlineto{\pgfqpoint{2.013690in}{0.295965in}}%
\pgfpathlineto{\pgfqpoint{2.050584in}{0.295965in}}%
\pgfpathlineto{\pgfqpoint{2.050584in}{0.287924in}}%
\pgfpathlineto{\pgfqpoint{2.062882in}{0.287924in}}%
\pgfpathlineto{\pgfqpoint{2.062882in}{0.299986in}}%
\pgfpathlineto{\pgfqpoint{2.075181in}{0.299986in}}%
\pgfpathlineto{\pgfqpoint{2.075181in}{0.287924in}}%
\pgfpathlineto{\pgfqpoint{2.087479in}{0.287924in}}%
\pgfpathlineto{\pgfqpoint{2.087479in}{0.291944in}}%
\pgfpathlineto{\pgfqpoint{2.099777in}{0.291944in}}%
\pgfpathlineto{\pgfqpoint{2.099777in}{0.283903in}}%
\pgfpathlineto{\pgfqpoint{2.124373in}{0.283903in}}%
\pgfpathlineto{\pgfqpoint{2.124373in}{0.291944in}}%
\pgfpathlineto{\pgfqpoint{2.136671in}{0.291944in}}%
\pgfpathlineto{\pgfqpoint{2.136671in}{0.287924in}}%
\pgfpathlineto{\pgfqpoint{2.148970in}{0.287924in}}%
\pgfpathlineto{\pgfqpoint{2.148970in}{0.291944in}}%
\pgfpathlineto{\pgfqpoint{2.161268in}{0.291944in}}%
\pgfpathlineto{\pgfqpoint{2.161268in}{0.287924in}}%
\pgfpathlineto{\pgfqpoint{2.173566in}{0.287924in}}%
\pgfpathlineto{\pgfqpoint{2.173566in}{0.283903in}}%
\pgfpathlineto{\pgfqpoint{2.185864in}{0.283903in}}%
\pgfpathlineto{\pgfqpoint{2.185864in}{0.291944in}}%
\pgfpathlineto{\pgfqpoint{2.198162in}{0.291944in}}%
\pgfpathlineto{\pgfqpoint{2.198162in}{0.283903in}}%
\pgfpathlineto{\pgfqpoint{2.247355in}{0.283903in}}%
\pgfpathlineto{\pgfqpoint{2.247355in}{0.287924in}}%
\pgfpathlineto{\pgfqpoint{2.259653in}{0.287924in}}%
\pgfpathlineto{\pgfqpoint{2.259653in}{0.283903in}}%
\pgfpathlineto{\pgfqpoint{2.284249in}{0.283903in}}%
\pgfpathlineto{\pgfqpoint{2.284249in}{0.287924in}}%
\pgfpathlineto{\pgfqpoint{2.296548in}{0.287924in}}%
\pgfpathlineto{\pgfqpoint{2.296548in}{0.283903in}}%
\pgfpathlineto{\pgfqpoint{2.358038in}{0.283903in}}%
\pgfpathlineto{\pgfqpoint{2.358038in}{0.287924in}}%
\pgfpathlineto{\pgfqpoint{2.370337in}{0.287924in}}%
\pgfpathlineto{\pgfqpoint{2.370337in}{0.283903in}}%
\pgfpathlineto{\pgfqpoint{2.407231in}{0.283903in}}%
\pgfpathlineto{\pgfqpoint{2.407231in}{0.287924in}}%
\pgfpathlineto{\pgfqpoint{2.419529in}{0.287924in}}%
\pgfpathlineto{\pgfqpoint{2.419529in}{0.283903in}}%
\pgfpathlineto{\pgfqpoint{2.431827in}{0.283903in}}%
\pgfpathlineto{\pgfqpoint{2.431827in}{0.287924in}}%
\pgfpathlineto{\pgfqpoint{2.444126in}{0.287924in}}%
\pgfpathlineto{\pgfqpoint{2.444126in}{0.283903in}}%
\pgfpathlineto{\pgfqpoint{2.468722in}{0.283903in}}%
\pgfpathlineto{\pgfqpoint{2.468722in}{0.287924in}}%
\pgfpathlineto{\pgfqpoint{2.481020in}{0.287924in}}%
\pgfpathlineto{\pgfqpoint{2.481020in}{0.283903in}}%
\pgfpathlineto{\pgfqpoint{2.579405in}{0.283903in}}%
\pgfpathlineto{\pgfqpoint{2.579405in}{0.287924in}}%
\pgfpathlineto{\pgfqpoint{2.591704in}{0.287924in}}%
\pgfpathlineto{\pgfqpoint{2.591704in}{0.283903in}}%
\pgfpathlineto{\pgfqpoint{2.640906in}{0.283903in}}%
\pgfpathlineto{\pgfqpoint{2.640906in}{0.283903in}}%
\pgfusepath{stroke}%
\end{pgfscope}%
\end{pgfpicture}%
\makeatother%
\endgroup%

%% file: gumbel_critical_q1p5_3d.pgf
\begingroup%
\makeatletter%
\begin{pgfpicture}%
\pgfpathrectangle{\pgfpointorigin}{\pgfqpoint{2.923228in}{1.806655in}}%
\pgfusepath{use as bounding box}%
\begin{pgfscope}%
\pgfsetbuttcap%
\pgfsetroundjoin%
\definecolor{currentfill}{rgb}{1.000000,1.000000,1.000000}%
\pgfsetfillcolor{currentfill}%
\pgfsetlinewidth{0.000000pt}%
\definecolor{currentstroke}{rgb}{1.000000,1.000000,1.000000}%
\pgfsetstrokecolor{currentstroke}%
\pgfsetdash{}{0pt}%
\pgfpathmoveto{\pgfqpoint{0.000000in}{0.000000in}}%
\pgfpathlineto{\pgfqpoint{2.923228in}{0.000000in}}%
\pgfpathlineto{\pgfqpoint{2.923228in}{1.806655in}}%
\pgfpathlineto{\pgfqpoint{0.000000in}{1.806655in}}%
\pgfpathclose%
\pgfusepath{fill}%
\end{pgfscope}%
\begin{pgfscope}%
\pgfsetbuttcap%
\pgfsetroundjoin%
\definecolor{currentfill}{rgb}{1.000000,1.000000,1.000000}%
\pgfsetfillcolor{currentfill}%
\pgfsetlinewidth{0.000000pt}%
\definecolor{currentstroke}{rgb}{0.000000,0.000000,0.000000}%
\pgfsetstrokecolor{currentstroke}%
\pgfsetstrokeopacity{0.000000}%
\pgfsetdash{}{0pt}%
\pgfpathmoveto{\pgfqpoint{0.365404in}{0.180665in}}%
\pgfpathlineto{\pgfqpoint{2.630906in}{0.180665in}}%
\pgfpathlineto{\pgfqpoint{2.630906in}{1.625989in}}%
\pgfpathlineto{\pgfqpoint{0.365404in}{1.625989in}}%
\pgfpathclose%
\pgfusepath{fill}%
\end{pgfscope}%
\begin{pgfscope}%
\pgfpathrectangle{\pgfqpoint{0.365404in}{0.180665in}}{\pgfqpoint{2.265502in}{1.445324in}} %
\pgfusepath{clip}%
\pgfsetrectcap%
\pgfsetroundjoin%
\pgfsetlinewidth{0.501875pt}%
\definecolor{currentstroke}{rgb}{0.000000,0.500000,0.000000}%
\pgfsetstrokecolor{currentstroke}%
\pgfsetdash{}{0pt}%
\pgfpathmoveto{\pgfqpoint{0.365402in}{0.284527in}}%
\pgfpathlineto{\pgfqpoint{0.399712in}{0.286652in}}%
\pgfpathlineto{\pgfqpoint{0.422035in}{0.290183in}}%
\pgfpathlineto{\pgfqpoint{0.439719in}{0.295182in}}%
\pgfpathlineto{\pgfqpoint{0.454999in}{0.301789in}}%
\pgfpathlineto{\pgfqpoint{0.468932in}{0.310240in}}%
\pgfpathlineto{\pgfqpoint{0.482146in}{0.320882in}}%
\pgfpathlineto{\pgfqpoint{0.495058in}{0.334186in}}%
\pgfpathlineto{\pgfqpoint{0.507981in}{0.350766in}}%
\pgfpathlineto{\pgfqpoint{0.521170in}{0.371402in}}%
\pgfpathlineto{\pgfqpoint{0.534856in}{0.397070in}}%
\pgfpathlineto{\pgfqpoint{0.549283in}{0.429009in}}%
\pgfpathlineto{\pgfqpoint{0.564755in}{0.468847in}}%
\pgfpathlineto{\pgfqpoint{0.581712in}{0.518886in}}%
\pgfpathlineto{\pgfqpoint{0.600945in}{0.582870in}}%
\pgfpathlineto{\pgfqpoint{0.624316in}{0.668805in}}%
\pgfpathlineto{\pgfqpoint{0.662244in}{0.818241in}}%
\pgfpathlineto{\pgfqpoint{0.699574in}{0.962080in}}%
\pgfpathlineto{\pgfqpoint{0.722418in}{1.041430in}}%
\pgfpathlineto{\pgfqpoint{0.741422in}{1.099736in}}%
\pgfpathlineto{\pgfqpoint{0.758159in}{1.144192in}}%
\pgfpathlineto{\pgfqpoint{0.773261in}{1.178245in}}%
\pgfpathlineto{\pgfqpoint{0.787065in}{1.204101in}}%
\pgfpathlineto{\pgfqpoint{0.799782in}{1.223385in}}%
\pgfpathlineto{\pgfqpoint{0.811572in}{1.237384in}}%
\pgfpathlineto{\pgfqpoint{0.822585in}{1.247133in}}%
\pgfpathlineto{\pgfqpoint{0.832978in}{1.253456in}}%
\pgfpathlineto{\pgfqpoint{0.842946in}{1.256976in}}%
\pgfpathlineto{\pgfqpoint{0.852738in}{1.258100in}}%
\pgfpathlineto{\pgfqpoint{0.862630in}{1.256988in}}%
\pgfpathlineto{\pgfqpoint{0.872919in}{1.253558in}}%
\pgfpathlineto{\pgfqpoint{0.883898in}{1.247492in}}%
\pgfpathlineto{\pgfqpoint{0.895855in}{1.238264in}}%
\pgfpathlineto{\pgfqpoint{0.909084in}{1.225145in}}%
\pgfpathlineto{\pgfqpoint{0.923932in}{1.207169in}}%
\pgfpathlineto{\pgfqpoint{0.940864in}{1.183037in}}%
\pgfpathlineto{\pgfqpoint{0.960598in}{1.150869in}}%
\pgfpathlineto{\pgfqpoint{0.984497in}{1.107444in}}%
\pgfpathlineto{\pgfqpoint{1.016063in}{1.045131in}}%
\pgfpathlineto{\pgfqpoint{1.145560in}{0.784401in}}%
\pgfpathlineto{\pgfqpoint{1.178974in}{0.724888in}}%
\pgfpathlineto{\pgfqpoint{1.210027in}{0.674035in}}%
\pgfpathlineto{\pgfqpoint{1.239764in}{0.629536in}}%
\pgfpathlineto{\pgfqpoint{1.268744in}{0.590114in}}%
\pgfpathlineto{\pgfqpoint{1.297338in}{0.554928in}}%
\pgfpathlineto{\pgfqpoint{1.325824in}{0.523363in}}%
\pgfpathlineto{\pgfqpoint{1.354436in}{0.494950in}}%
\pgfpathlineto{\pgfqpoint{1.383398in}{0.469301in}}%
\pgfpathlineto{\pgfqpoint{1.412932in}{0.446103in}}%
\pgfpathlineto{\pgfqpoint{1.443273in}{0.425090in}}%
\pgfpathlineto{\pgfqpoint{1.474677in}{0.406046in}}%
\pgfpathlineto{\pgfqpoint{1.507434in}{0.388786in}}%
\pgfpathlineto{\pgfqpoint{1.541884in}{0.373158in}}%
\pgfpathlineto{\pgfqpoint{1.578436in}{0.359033in}}%
\pgfpathlineto{\pgfqpoint{1.617580in}{0.346312in}}%
\pgfpathlineto{\pgfqpoint{1.659945in}{0.334909in}}%
\pgfpathlineto{\pgfqpoint{1.706325in}{0.324761in}}%
\pgfpathlineto{\pgfqpoint{1.757796in}{0.315815in}}%
\pgfpathlineto{\pgfqpoint{1.815827in}{0.308033in}}%
\pgfpathlineto{\pgfqpoint{1.882562in}{0.301384in}}%
\pgfpathlineto{\pgfqpoint{1.961279in}{0.295844in}}%
\pgfpathlineto{\pgfqpoint{2.057411in}{0.291393in}}%
\pgfpathlineto{\pgfqpoint{2.180934in}{0.288013in}}%
\pgfpathlineto{\pgfqpoint{2.353317in}{0.285680in}}%
\pgfpathlineto{\pgfqpoint{2.630907in}{0.284363in}}%
\pgfpathlineto{\pgfqpoint{2.630907in}{0.284363in}}%
\pgfusepath{stroke}%
\end{pgfscope}%
\begin{pgfscope}%
\pgfsetbuttcap%
\pgfsetroundjoin%
\definecolor{currentfill}{rgb}{0.000000,0.000000,0.000000}%
\pgfsetfillcolor{currentfill}%
\pgfsetlinewidth{0.501875pt}%
\definecolor{currentstroke}{rgb}{0.000000,0.000000,0.000000}%
\pgfsetstrokecolor{currentstroke}%
\pgfsetdash{}{0pt}%
\pgfsys@defobject{currentmarker}{\pgfqpoint{0.000000in}{0.000000in}}{\pgfqpoint{0.000000in}{0.055556in}}{%
\pgfpathmoveto{\pgfqpoint{0.000000in}{0.000000in}}%
\pgfpathlineto{\pgfqpoint{0.000000in}{0.055556in}}%
\pgfusepath{stroke,fill}%
}%
\begin{pgfscope}%
\pgfsys@transformshift{0.444433in}{0.180665in}%
\pgfsys@useobject{currentmarker}{}%
\end{pgfscope}%
\end{pgfscope}%
\begin{pgfscope}%
\pgfsetbuttcap%
\pgfsetroundjoin%
\definecolor{currentfill}{rgb}{0.000000,0.000000,0.000000}%
\pgfsetfillcolor{currentfill}%
\pgfsetlinewidth{0.501875pt}%
\definecolor{currentstroke}{rgb}{0.000000,0.000000,0.000000}%
\pgfsetstrokecolor{currentstroke}%
\pgfsetdash{}{0pt}%
\pgfsys@defobject{currentmarker}{\pgfqpoint{0.000000in}{-0.055556in}}{\pgfqpoint{0.000000in}{0.000000in}}{%
\pgfpathmoveto{\pgfqpoint{0.000000in}{0.000000in}}%
\pgfpathlineto{\pgfqpoint{0.000000in}{-0.055556in}}%
\pgfusepath{stroke,fill}%
}%
\begin{pgfscope}%
\pgfsys@transformshift{0.444433in}{1.625989in}%
\pgfsys@useobject{currentmarker}{}%
\end{pgfscope}%
\end{pgfscope}%
\begin{pgfscope}%
\pgftext[x=0.444433in,y=0.125110in,,top]{{\rmfamily\fontsize{8.000000}{9.600000}\selectfont \(\displaystyle -2\)}}%
\end{pgfscope}%
\begin{pgfscope}%
\pgfsetbuttcap%
\pgfsetroundjoin%
\definecolor{currentfill}{rgb}{0.000000,0.000000,0.000000}%
\pgfsetfillcolor{currentfill}%
\pgfsetlinewidth{0.501875pt}%
\definecolor{currentstroke}{rgb}{0.000000,0.000000,0.000000}%
\pgfsetstrokecolor{currentstroke}%
\pgfsetdash{}{0pt}%
\pgfsys@defobject{currentmarker}{\pgfqpoint{0.000000in}{0.000000in}}{\pgfqpoint{0.000000in}{0.055556in}}{%
\pgfpathmoveto{\pgfqpoint{0.000000in}{0.000000in}}%
\pgfpathlineto{\pgfqpoint{0.000000in}{0.055556in}}%
\pgfusepath{stroke,fill}%
}%
\begin{pgfscope}%
\pgfsys@transformshift{0.707863in}{0.180665in}%
\pgfsys@useobject{currentmarker}{}%
\end{pgfscope}%
\end{pgfscope}%
\begin{pgfscope}%
\pgfsetbuttcap%
\pgfsetroundjoin%
\definecolor{currentfill}{rgb}{0.000000,0.000000,0.000000}%
\pgfsetfillcolor{currentfill}%
\pgfsetlinewidth{0.501875pt}%
\definecolor{currentstroke}{rgb}{0.000000,0.000000,0.000000}%
\pgfsetstrokecolor{currentstroke}%
\pgfsetdash{}{0pt}%
\pgfsys@defobject{currentmarker}{\pgfqpoint{0.000000in}{-0.055556in}}{\pgfqpoint{0.000000in}{0.000000in}}{%
\pgfpathmoveto{\pgfqpoint{0.000000in}{0.000000in}}%
\pgfpathlineto{\pgfqpoint{0.000000in}{-0.055556in}}%
\pgfusepath{stroke,fill}%
}%
\begin{pgfscope}%
\pgfsys@transformshift{0.707863in}{1.625989in}%
\pgfsys@useobject{currentmarker}{}%
\end{pgfscope}%
\end{pgfscope}%
\begin{pgfscope}%
\pgftext[x=0.707863in,y=0.125110in,,top]{{\rmfamily\fontsize{8.000000}{9.600000}\selectfont \(\displaystyle -1\)}}%
\end{pgfscope}%
\begin{pgfscope}%
\pgfsetbuttcap%
\pgfsetroundjoin%
\definecolor{currentfill}{rgb}{0.000000,0.000000,0.000000}%
\pgfsetfillcolor{currentfill}%
\pgfsetlinewidth{0.501875pt}%
\definecolor{currentstroke}{rgb}{0.000000,0.000000,0.000000}%
\pgfsetstrokecolor{currentstroke}%
\pgfsetdash{}{0pt}%
\pgfsys@defobject{currentmarker}{\pgfqpoint{0.000000in}{0.000000in}}{\pgfqpoint{0.000000in}{0.055556in}}{%
\pgfpathmoveto{\pgfqpoint{0.000000in}{0.000000in}}%
\pgfpathlineto{\pgfqpoint{0.000000in}{0.055556in}}%
\pgfusepath{stroke,fill}%
}%
\begin{pgfscope}%
\pgfsys@transformshift{0.971294in}{0.180665in}%
\pgfsys@useobject{currentmarker}{}%
\end{pgfscope}%
\end{pgfscope}%
\begin{pgfscope}%
\pgfsetbuttcap%
\pgfsetroundjoin%
\definecolor{currentfill}{rgb}{0.000000,0.000000,0.000000}%
\pgfsetfillcolor{currentfill}%
\pgfsetlinewidth{0.501875pt}%
\definecolor{currentstroke}{rgb}{0.000000,0.000000,0.000000}%
\pgfsetstrokecolor{currentstroke}%
\pgfsetdash{}{0pt}%
\pgfsys@defobject{currentmarker}{\pgfqpoint{0.000000in}{-0.055556in}}{\pgfqpoint{0.000000in}{0.000000in}}{%
\pgfpathmoveto{\pgfqpoint{0.000000in}{0.000000in}}%
\pgfpathlineto{\pgfqpoint{0.000000in}{-0.055556in}}%
\pgfusepath{stroke,fill}%
}%
\begin{pgfscope}%
\pgfsys@transformshift{0.971294in}{1.625989in}%
\pgfsys@useobject{currentmarker}{}%
\end{pgfscope}%
\end{pgfscope}%
\begin{pgfscope}%
\pgftext[x=0.971294in,y=0.125110in,,top]{{\rmfamily\fontsize{8.000000}{9.600000}\selectfont \(\displaystyle 0\)}}%
\end{pgfscope}%
\begin{pgfscope}%
\pgfsetbuttcap%
\pgfsetroundjoin%
\definecolor{currentfill}{rgb}{0.000000,0.000000,0.000000}%
\pgfsetfillcolor{currentfill}%
\pgfsetlinewidth{0.501875pt}%
\definecolor{currentstroke}{rgb}{0.000000,0.000000,0.000000}%
\pgfsetstrokecolor{currentstroke}%
\pgfsetdash{}{0pt}%
\pgfsys@defobject{currentmarker}{\pgfqpoint{0.000000in}{0.000000in}}{\pgfqpoint{0.000000in}{0.055556in}}{%
\pgfpathmoveto{\pgfqpoint{0.000000in}{0.000000in}}%
\pgfpathlineto{\pgfqpoint{0.000000in}{0.055556in}}%
\pgfusepath{stroke,fill}%
}%
\begin{pgfscope}%
\pgfsys@transformshift{1.234724in}{0.180665in}%
\pgfsys@useobject{currentmarker}{}%
\end{pgfscope}%
\end{pgfscope}%
\begin{pgfscope}%
\pgfsetbuttcap%
\pgfsetroundjoin%
\definecolor{currentfill}{rgb}{0.000000,0.000000,0.000000}%
\pgfsetfillcolor{currentfill}%
\pgfsetlinewidth{0.501875pt}%
\definecolor{currentstroke}{rgb}{0.000000,0.000000,0.000000}%
\pgfsetstrokecolor{currentstroke}%
\pgfsetdash{}{0pt}%
\pgfsys@defobject{currentmarker}{\pgfqpoint{0.000000in}{-0.055556in}}{\pgfqpoint{0.000000in}{0.000000in}}{%
\pgfpathmoveto{\pgfqpoint{0.000000in}{0.000000in}}%
\pgfpathlineto{\pgfqpoint{0.000000in}{-0.055556in}}%
\pgfusepath{stroke,fill}%
}%
\begin{pgfscope}%
\pgfsys@transformshift{1.234724in}{1.625989in}%
\pgfsys@useobject{currentmarker}{}%
\end{pgfscope}%
\end{pgfscope}%
\begin{pgfscope}%
\pgftext[x=1.234724in,y=0.125110in,,top]{{\rmfamily\fontsize{8.000000}{9.600000}\selectfont \(\displaystyle 1\)}}%
\end{pgfscope}%
\begin{pgfscope}%
\pgfsetbuttcap%
\pgfsetroundjoin%
\definecolor{currentfill}{rgb}{0.000000,0.000000,0.000000}%
\pgfsetfillcolor{currentfill}%
\pgfsetlinewidth{0.501875pt}%
\definecolor{currentstroke}{rgb}{0.000000,0.000000,0.000000}%
\pgfsetstrokecolor{currentstroke}%
\pgfsetdash{}{0pt}%
\pgfsys@defobject{currentmarker}{\pgfqpoint{0.000000in}{0.000000in}}{\pgfqpoint{0.000000in}{0.055556in}}{%
\pgfpathmoveto{\pgfqpoint{0.000000in}{0.000000in}}%
\pgfpathlineto{\pgfqpoint{0.000000in}{0.055556in}}%
\pgfusepath{stroke,fill}%
}%
\begin{pgfscope}%
\pgfsys@transformshift{1.498155in}{0.180665in}%
\pgfsys@useobject{currentmarker}{}%
\end{pgfscope}%
\end{pgfscope}%
\begin{pgfscope}%
\pgfsetbuttcap%
\pgfsetroundjoin%
\definecolor{currentfill}{rgb}{0.000000,0.000000,0.000000}%
\pgfsetfillcolor{currentfill}%
\pgfsetlinewidth{0.501875pt}%
\definecolor{currentstroke}{rgb}{0.000000,0.000000,0.000000}%
\pgfsetstrokecolor{currentstroke}%
\pgfsetdash{}{0pt}%
\pgfsys@defobject{currentmarker}{\pgfqpoint{0.000000in}{-0.055556in}}{\pgfqpoint{0.000000in}{0.000000in}}{%
\pgfpathmoveto{\pgfqpoint{0.000000in}{0.000000in}}%
\pgfpathlineto{\pgfqpoint{0.000000in}{-0.055556in}}%
\pgfusepath{stroke,fill}%
}%
\begin{pgfscope}%
\pgfsys@transformshift{1.498155in}{1.625989in}%
\pgfsys@useobject{currentmarker}{}%
\end{pgfscope}%
\end{pgfscope}%
\begin{pgfscope}%
\pgftext[x=1.498155in,y=0.125110in,,top]{{\rmfamily\fontsize{8.000000}{9.600000}\selectfont \(\displaystyle 2\)}}%
\end{pgfscope}%
\begin{pgfscope}%
\pgfsetbuttcap%
\pgfsetroundjoin%
\definecolor{currentfill}{rgb}{0.000000,0.000000,0.000000}%
\pgfsetfillcolor{currentfill}%
\pgfsetlinewidth{0.501875pt}%
\definecolor{currentstroke}{rgb}{0.000000,0.000000,0.000000}%
\pgfsetstrokecolor{currentstroke}%
\pgfsetdash{}{0pt}%
\pgfsys@defobject{currentmarker}{\pgfqpoint{0.000000in}{0.000000in}}{\pgfqpoint{0.000000in}{0.055556in}}{%
\pgfpathmoveto{\pgfqpoint{0.000000in}{0.000000in}}%
\pgfpathlineto{\pgfqpoint{0.000000in}{0.055556in}}%
\pgfusepath{stroke,fill}%
}%
\begin{pgfscope}%
\pgfsys@transformshift{1.761585in}{0.180665in}%
\pgfsys@useobject{currentmarker}{}%
\end{pgfscope}%
\end{pgfscope}%
\begin{pgfscope}%
\pgfsetbuttcap%
\pgfsetroundjoin%
\definecolor{currentfill}{rgb}{0.000000,0.000000,0.000000}%
\pgfsetfillcolor{currentfill}%
\pgfsetlinewidth{0.501875pt}%
\definecolor{currentstroke}{rgb}{0.000000,0.000000,0.000000}%
\pgfsetstrokecolor{currentstroke}%
\pgfsetdash{}{0pt}%
\pgfsys@defobject{currentmarker}{\pgfqpoint{0.000000in}{-0.055556in}}{\pgfqpoint{0.000000in}{0.000000in}}{%
\pgfpathmoveto{\pgfqpoint{0.000000in}{0.000000in}}%
\pgfpathlineto{\pgfqpoint{0.000000in}{-0.055556in}}%
\pgfusepath{stroke,fill}%
}%
\begin{pgfscope}%
\pgfsys@transformshift{1.761585in}{1.625989in}%
\pgfsys@useobject{currentmarker}{}%
\end{pgfscope}%
\end{pgfscope}%
\begin{pgfscope}%
\pgftext[x=1.761585in,y=0.125110in,,top]{{\rmfamily\fontsize{8.000000}{9.600000}\selectfont \(\displaystyle 3\)}}%
\end{pgfscope}%
\begin{pgfscope}%
\pgfsetbuttcap%
\pgfsetroundjoin%
\definecolor{currentfill}{rgb}{0.000000,0.000000,0.000000}%
\pgfsetfillcolor{currentfill}%
\pgfsetlinewidth{0.501875pt}%
\definecolor{currentstroke}{rgb}{0.000000,0.000000,0.000000}%
\pgfsetstrokecolor{currentstroke}%
\pgfsetdash{}{0pt}%
\pgfsys@defobject{currentmarker}{\pgfqpoint{0.000000in}{0.000000in}}{\pgfqpoint{0.000000in}{0.055556in}}{%
\pgfpathmoveto{\pgfqpoint{0.000000in}{0.000000in}}%
\pgfpathlineto{\pgfqpoint{0.000000in}{0.055556in}}%
\pgfusepath{stroke,fill}%
}%
\begin{pgfscope}%
\pgfsys@transformshift{2.025016in}{0.180665in}%
\pgfsys@useobject{currentmarker}{}%
\end{pgfscope}%
\end{pgfscope}%
\begin{pgfscope}%
\pgfsetbuttcap%
\pgfsetroundjoin%
\definecolor{currentfill}{rgb}{0.000000,0.000000,0.000000}%
\pgfsetfillcolor{currentfill}%
\pgfsetlinewidth{0.501875pt}%
\definecolor{currentstroke}{rgb}{0.000000,0.000000,0.000000}%
\pgfsetstrokecolor{currentstroke}%
\pgfsetdash{}{0pt}%
\pgfsys@defobject{currentmarker}{\pgfqpoint{0.000000in}{-0.055556in}}{\pgfqpoint{0.000000in}{0.000000in}}{%
\pgfpathmoveto{\pgfqpoint{0.000000in}{0.000000in}}%
\pgfpathlineto{\pgfqpoint{0.000000in}{-0.055556in}}%
\pgfusepath{stroke,fill}%
}%
\begin{pgfscope}%
\pgfsys@transformshift{2.025016in}{1.625989in}%
\pgfsys@useobject{currentmarker}{}%
\end{pgfscope}%
\end{pgfscope}%
\begin{pgfscope}%
\pgftext[x=2.025016in,y=0.125110in,,top]{{\rmfamily\fontsize{8.000000}{9.600000}\selectfont \(\displaystyle 4\)}}%
\end{pgfscope}%
\begin{pgfscope}%
\pgfsetbuttcap%
\pgfsetroundjoin%
\definecolor{currentfill}{rgb}{0.000000,0.000000,0.000000}%
\pgfsetfillcolor{currentfill}%
\pgfsetlinewidth{0.501875pt}%
\definecolor{currentstroke}{rgb}{0.000000,0.000000,0.000000}%
\pgfsetstrokecolor{currentstroke}%
\pgfsetdash{}{0pt}%
\pgfsys@defobject{currentmarker}{\pgfqpoint{0.000000in}{0.000000in}}{\pgfqpoint{0.000000in}{0.055556in}}{%
\pgfpathmoveto{\pgfqpoint{0.000000in}{0.000000in}}%
\pgfpathlineto{\pgfqpoint{0.000000in}{0.055556in}}%
\pgfusepath{stroke,fill}%
}%
\begin{pgfscope}%
\pgfsys@transformshift{2.288446in}{0.180665in}%
\pgfsys@useobject{currentmarker}{}%
\end{pgfscope}%
\end{pgfscope}%
\begin{pgfscope}%
\pgfsetbuttcap%
\pgfsetroundjoin%
\definecolor{currentfill}{rgb}{0.000000,0.000000,0.000000}%
\pgfsetfillcolor{currentfill}%
\pgfsetlinewidth{0.501875pt}%
\definecolor{currentstroke}{rgb}{0.000000,0.000000,0.000000}%
\pgfsetstrokecolor{currentstroke}%
\pgfsetdash{}{0pt}%
\pgfsys@defobject{currentmarker}{\pgfqpoint{0.000000in}{-0.055556in}}{\pgfqpoint{0.000000in}{0.000000in}}{%
\pgfpathmoveto{\pgfqpoint{0.000000in}{0.000000in}}%
\pgfpathlineto{\pgfqpoint{0.000000in}{-0.055556in}}%
\pgfusepath{stroke,fill}%
}%
\begin{pgfscope}%
\pgfsys@transformshift{2.288446in}{1.625989in}%
\pgfsys@useobject{currentmarker}{}%
\end{pgfscope}%
\end{pgfscope}%
\begin{pgfscope}%
\pgftext[x=2.288446in,y=0.125110in,,top]{{\rmfamily\fontsize{8.000000}{9.600000}\selectfont \(\displaystyle 5\)}}%
\end{pgfscope}%
\begin{pgfscope}%
\pgfsetbuttcap%
\pgfsetroundjoin%
\definecolor{currentfill}{rgb}{0.000000,0.000000,0.000000}%
\pgfsetfillcolor{currentfill}%
\pgfsetlinewidth{0.501875pt}%
\definecolor{currentstroke}{rgb}{0.000000,0.000000,0.000000}%
\pgfsetstrokecolor{currentstroke}%
\pgfsetdash{}{0pt}%
\pgfsys@defobject{currentmarker}{\pgfqpoint{0.000000in}{0.000000in}}{\pgfqpoint{0.000000in}{0.055556in}}{%
\pgfpathmoveto{\pgfqpoint{0.000000in}{0.000000in}}%
\pgfpathlineto{\pgfqpoint{0.000000in}{0.055556in}}%
\pgfusepath{stroke,fill}%
}%
\begin{pgfscope}%
\pgfsys@transformshift{2.551876in}{0.180665in}%
\pgfsys@useobject{currentmarker}{}%
\end{pgfscope}%
\end{pgfscope}%
\begin{pgfscope}%
\pgfsetbuttcap%
\pgfsetroundjoin%
\definecolor{currentfill}{rgb}{0.000000,0.000000,0.000000}%
\pgfsetfillcolor{currentfill}%
\pgfsetlinewidth{0.501875pt}%
\definecolor{currentstroke}{rgb}{0.000000,0.000000,0.000000}%
\pgfsetstrokecolor{currentstroke}%
\pgfsetdash{}{0pt}%
\pgfsys@defobject{currentmarker}{\pgfqpoint{0.000000in}{-0.055556in}}{\pgfqpoint{0.000000in}{0.000000in}}{%
\pgfpathmoveto{\pgfqpoint{0.000000in}{0.000000in}}%
\pgfpathlineto{\pgfqpoint{0.000000in}{-0.055556in}}%
\pgfusepath{stroke,fill}%
}%
\begin{pgfscope}%
\pgfsys@transformshift{2.551876in}{1.625989in}%
\pgfsys@useobject{currentmarker}{}%
\end{pgfscope}%
\end{pgfscope}%
\begin{pgfscope}%
\pgftext[x=2.551876in,y=0.125110in,,top]{{\rmfamily\fontsize{8.000000}{9.600000}\selectfont \(\displaystyle 6\)}}%
\end{pgfscope}%
\begin{pgfscope}%
\pgftext[x=1.498155in,y=-0.042459in,,top]{{\rmfamily\fontsize{10.000000}{12.000000}\selectfont \(\displaystyle s\)}}%
\end{pgfscope}%
\begin{pgfscope}%
\pgfsetbuttcap%
\pgfsetroundjoin%
\definecolor{currentfill}{rgb}{0.000000,0.000000,0.000000}%
\pgfsetfillcolor{currentfill}%
\pgfsetlinewidth{0.501875pt}%
\definecolor{currentstroke}{rgb}{0.000000,0.000000,0.000000}%
\pgfsetstrokecolor{currentstroke}%
\pgfsetdash{}{0pt}%
\pgfsys@defobject{currentmarker}{\pgfqpoint{0.000000in}{0.000000in}}{\pgfqpoint{0.055556in}{0.000000in}}{%
\pgfpathmoveto{\pgfqpoint{0.000000in}{0.000000in}}%
\pgfpathlineto{\pgfqpoint{0.055556in}{0.000000in}}%
\pgfusepath{stroke,fill}%
}%
\begin{pgfscope}%
\pgfsys@transformshift{0.365404in}{0.283903in}%
\pgfsys@useobject{currentmarker}{}%
\end{pgfscope}%
\end{pgfscope}%
\begin{pgfscope}%
\pgfsetbuttcap%
\pgfsetroundjoin%
\definecolor{currentfill}{rgb}{0.000000,0.000000,0.000000}%
\pgfsetfillcolor{currentfill}%
\pgfsetlinewidth{0.501875pt}%
\definecolor{currentstroke}{rgb}{0.000000,0.000000,0.000000}%
\pgfsetstrokecolor{currentstroke}%
\pgfsetdash{}{0pt}%
\pgfsys@defobject{currentmarker}{\pgfqpoint{-0.055556in}{0.000000in}}{\pgfqpoint{0.000000in}{0.000000in}}{%
\pgfpathmoveto{\pgfqpoint{0.000000in}{0.000000in}}%
\pgfpathlineto{\pgfqpoint{-0.055556in}{0.000000in}}%
\pgfusepath{stroke,fill}%
}%
\begin{pgfscope}%
\pgfsys@transformshift{2.630906in}{0.283903in}%
\pgfsys@useobject{currentmarker}{}%
\end{pgfscope}%
\end{pgfscope}%
\begin{pgfscope}%
\pgftext[x=0.309848in,y=0.283903in,right,]{{\rmfamily\fontsize{8.000000}{9.600000}\selectfont \(\displaystyle 0.0\)}}%
\end{pgfscope}%
\begin{pgfscope}%
\pgfsetbuttcap%
\pgfsetroundjoin%
\definecolor{currentfill}{rgb}{0.000000,0.000000,0.000000}%
\pgfsetfillcolor{currentfill}%
\pgfsetlinewidth{0.501875pt}%
\definecolor{currentstroke}{rgb}{0.000000,0.000000,0.000000}%
\pgfsetstrokecolor{currentstroke}%
\pgfsetdash{}{0pt}%
\pgfsys@defobject{currentmarker}{\pgfqpoint{0.000000in}{0.000000in}}{\pgfqpoint{0.055556in}{0.000000in}}{%
\pgfpathmoveto{\pgfqpoint{0.000000in}{0.000000in}}%
\pgfpathlineto{\pgfqpoint{0.055556in}{0.000000in}}%
\pgfusepath{stroke,fill}%
}%
\begin{pgfscope}%
\pgfsys@transformshift{0.365404in}{0.490378in}%
\pgfsys@useobject{currentmarker}{}%
\end{pgfscope}%
\end{pgfscope}%
\begin{pgfscope}%
\pgfsetbuttcap%
\pgfsetroundjoin%
\definecolor{currentfill}{rgb}{0.000000,0.000000,0.000000}%
\pgfsetfillcolor{currentfill}%
\pgfsetlinewidth{0.501875pt}%
\definecolor{currentstroke}{rgb}{0.000000,0.000000,0.000000}%
\pgfsetstrokecolor{currentstroke}%
\pgfsetdash{}{0pt}%
\pgfsys@defobject{currentmarker}{\pgfqpoint{-0.055556in}{0.000000in}}{\pgfqpoint{0.000000in}{0.000000in}}{%
\pgfpathmoveto{\pgfqpoint{0.000000in}{0.000000in}}%
\pgfpathlineto{\pgfqpoint{-0.055556in}{0.000000in}}%
\pgfusepath{stroke,fill}%
}%
\begin{pgfscope}%
\pgfsys@transformshift{2.630906in}{0.490378in}%
\pgfsys@useobject{currentmarker}{}%
\end{pgfscope}%
\end{pgfscope}%
\begin{pgfscope}%
\pgftext[x=0.309848in,y=0.490378in,right,]{{\rmfamily\fontsize{8.000000}{9.600000}\selectfont \(\displaystyle 0.1\)}}%
\end{pgfscope}%
\begin{pgfscope}%
\pgfsetbuttcap%
\pgfsetroundjoin%
\definecolor{currentfill}{rgb}{0.000000,0.000000,0.000000}%
\pgfsetfillcolor{currentfill}%
\pgfsetlinewidth{0.501875pt}%
\definecolor{currentstroke}{rgb}{0.000000,0.000000,0.000000}%
\pgfsetstrokecolor{currentstroke}%
\pgfsetdash{}{0pt}%
\pgfsys@defobject{currentmarker}{\pgfqpoint{0.000000in}{0.000000in}}{\pgfqpoint{0.055556in}{0.000000in}}{%
\pgfpathmoveto{\pgfqpoint{0.000000in}{0.000000in}}%
\pgfpathlineto{\pgfqpoint{0.055556in}{0.000000in}}%
\pgfusepath{stroke,fill}%
}%
\begin{pgfscope}%
\pgfsys@transformshift{0.365404in}{0.696852in}%
\pgfsys@useobject{currentmarker}{}%
\end{pgfscope}%
\end{pgfscope}%
\begin{pgfscope}%
\pgfsetbuttcap%
\pgfsetroundjoin%
\definecolor{currentfill}{rgb}{0.000000,0.000000,0.000000}%
\pgfsetfillcolor{currentfill}%
\pgfsetlinewidth{0.501875pt}%
\definecolor{currentstroke}{rgb}{0.000000,0.000000,0.000000}%
\pgfsetstrokecolor{currentstroke}%
\pgfsetdash{}{0pt}%
\pgfsys@defobject{currentmarker}{\pgfqpoint{-0.055556in}{0.000000in}}{\pgfqpoint{0.000000in}{0.000000in}}{%
\pgfpathmoveto{\pgfqpoint{0.000000in}{0.000000in}}%
\pgfpathlineto{\pgfqpoint{-0.055556in}{0.000000in}}%
\pgfusepath{stroke,fill}%
}%
\begin{pgfscope}%
\pgfsys@transformshift{2.630906in}{0.696852in}%
\pgfsys@useobject{currentmarker}{}%
\end{pgfscope}%
\end{pgfscope}%
\begin{pgfscope}%
\pgftext[x=0.309848in,y=0.696852in,right,]{{\rmfamily\fontsize{8.000000}{9.600000}\selectfont \(\displaystyle 0.2\)}}%
\end{pgfscope}%
\begin{pgfscope}%
\pgfsetbuttcap%
\pgfsetroundjoin%
\definecolor{currentfill}{rgb}{0.000000,0.000000,0.000000}%
\pgfsetfillcolor{currentfill}%
\pgfsetlinewidth{0.501875pt}%
\definecolor{currentstroke}{rgb}{0.000000,0.000000,0.000000}%
\pgfsetstrokecolor{currentstroke}%
\pgfsetdash{}{0pt}%
\pgfsys@defobject{currentmarker}{\pgfqpoint{0.000000in}{0.000000in}}{\pgfqpoint{0.055556in}{0.000000in}}{%
\pgfpathmoveto{\pgfqpoint{0.000000in}{0.000000in}}%
\pgfpathlineto{\pgfqpoint{0.055556in}{0.000000in}}%
\pgfusepath{stroke,fill}%
}%
\begin{pgfscope}%
\pgfsys@transformshift{0.365404in}{0.903327in}%
\pgfsys@useobject{currentmarker}{}%
\end{pgfscope}%
\end{pgfscope}%
\begin{pgfscope}%
\pgfsetbuttcap%
\pgfsetroundjoin%
\definecolor{currentfill}{rgb}{0.000000,0.000000,0.000000}%
\pgfsetfillcolor{currentfill}%
\pgfsetlinewidth{0.501875pt}%
\definecolor{currentstroke}{rgb}{0.000000,0.000000,0.000000}%
\pgfsetstrokecolor{currentstroke}%
\pgfsetdash{}{0pt}%
\pgfsys@defobject{currentmarker}{\pgfqpoint{-0.055556in}{0.000000in}}{\pgfqpoint{0.000000in}{0.000000in}}{%
\pgfpathmoveto{\pgfqpoint{0.000000in}{0.000000in}}%
\pgfpathlineto{\pgfqpoint{-0.055556in}{0.000000in}}%
\pgfusepath{stroke,fill}%
}%
\begin{pgfscope}%
\pgfsys@transformshift{2.630906in}{0.903327in}%
\pgfsys@useobject{currentmarker}{}%
\end{pgfscope}%
\end{pgfscope}%
\begin{pgfscope}%
\pgftext[x=0.309848in,y=0.903327in,right,]{{\rmfamily\fontsize{8.000000}{9.600000}\selectfont \(\displaystyle 0.3\)}}%
\end{pgfscope}%
\begin{pgfscope}%
\pgfsetbuttcap%
\pgfsetroundjoin%
\definecolor{currentfill}{rgb}{0.000000,0.000000,0.000000}%
\pgfsetfillcolor{currentfill}%
\pgfsetlinewidth{0.501875pt}%
\definecolor{currentstroke}{rgb}{0.000000,0.000000,0.000000}%
\pgfsetstrokecolor{currentstroke}%
\pgfsetdash{}{0pt}%
\pgfsys@defobject{currentmarker}{\pgfqpoint{0.000000in}{0.000000in}}{\pgfqpoint{0.055556in}{0.000000in}}{%
\pgfpathmoveto{\pgfqpoint{0.000000in}{0.000000in}}%
\pgfpathlineto{\pgfqpoint{0.055556in}{0.000000in}}%
\pgfusepath{stroke,fill}%
}%
\begin{pgfscope}%
\pgfsys@transformshift{0.365404in}{1.109802in}%
\pgfsys@useobject{currentmarker}{}%
\end{pgfscope}%
\end{pgfscope}%
\begin{pgfscope}%
\pgfsetbuttcap%
\pgfsetroundjoin%
\definecolor{currentfill}{rgb}{0.000000,0.000000,0.000000}%
\pgfsetfillcolor{currentfill}%
\pgfsetlinewidth{0.501875pt}%
\definecolor{currentstroke}{rgb}{0.000000,0.000000,0.000000}%
\pgfsetstrokecolor{currentstroke}%
\pgfsetdash{}{0pt}%
\pgfsys@defobject{currentmarker}{\pgfqpoint{-0.055556in}{0.000000in}}{\pgfqpoint{0.000000in}{0.000000in}}{%
\pgfpathmoveto{\pgfqpoint{0.000000in}{0.000000in}}%
\pgfpathlineto{\pgfqpoint{-0.055556in}{0.000000in}}%
\pgfusepath{stroke,fill}%
}%
\begin{pgfscope}%
\pgfsys@transformshift{2.630906in}{1.109802in}%
\pgfsys@useobject{currentmarker}{}%
\end{pgfscope}%
\end{pgfscope}%
\begin{pgfscope}%
\pgftext[x=0.309848in,y=1.109802in,right,]{{\rmfamily\fontsize{8.000000}{9.600000}\selectfont \(\displaystyle 0.4\)}}%
\end{pgfscope}%
\begin{pgfscope}%
\pgfsetbuttcap%
\pgfsetroundjoin%
\definecolor{currentfill}{rgb}{0.000000,0.000000,0.000000}%
\pgfsetfillcolor{currentfill}%
\pgfsetlinewidth{0.501875pt}%
\definecolor{currentstroke}{rgb}{0.000000,0.000000,0.000000}%
\pgfsetstrokecolor{currentstroke}%
\pgfsetdash{}{0pt}%
\pgfsys@defobject{currentmarker}{\pgfqpoint{0.000000in}{0.000000in}}{\pgfqpoint{0.055556in}{0.000000in}}{%
\pgfpathmoveto{\pgfqpoint{0.000000in}{0.000000in}}%
\pgfpathlineto{\pgfqpoint{0.055556in}{0.000000in}}%
\pgfusepath{stroke,fill}%
}%
\begin{pgfscope}%
\pgfsys@transformshift{0.365404in}{1.316277in}%
\pgfsys@useobject{currentmarker}{}%
\end{pgfscope}%
\end{pgfscope}%
\begin{pgfscope}%
\pgfsetbuttcap%
\pgfsetroundjoin%
\definecolor{currentfill}{rgb}{0.000000,0.000000,0.000000}%
\pgfsetfillcolor{currentfill}%
\pgfsetlinewidth{0.501875pt}%
\definecolor{currentstroke}{rgb}{0.000000,0.000000,0.000000}%
\pgfsetstrokecolor{currentstroke}%
\pgfsetdash{}{0pt}%
\pgfsys@defobject{currentmarker}{\pgfqpoint{-0.055556in}{0.000000in}}{\pgfqpoint{0.000000in}{0.000000in}}{%
\pgfpathmoveto{\pgfqpoint{0.000000in}{0.000000in}}%
\pgfpathlineto{\pgfqpoint{-0.055556in}{0.000000in}}%
\pgfusepath{stroke,fill}%
}%
\begin{pgfscope}%
\pgfsys@transformshift{2.630906in}{1.316277in}%
\pgfsys@useobject{currentmarker}{}%
\end{pgfscope}%
\end{pgfscope}%
\begin{pgfscope}%
\pgftext[x=0.309848in,y=1.316277in,right,]{{\rmfamily\fontsize{8.000000}{9.600000}\selectfont \(\displaystyle 0.5\)}}%
\end{pgfscope}%
\begin{pgfscope}%
\pgfsetbuttcap%
\pgfsetroundjoin%
\definecolor{currentfill}{rgb}{0.000000,0.000000,0.000000}%
\pgfsetfillcolor{currentfill}%
\pgfsetlinewidth{0.501875pt}%
\definecolor{currentstroke}{rgb}{0.000000,0.000000,0.000000}%
\pgfsetstrokecolor{currentstroke}%
\pgfsetdash{}{0pt}%
\pgfsys@defobject{currentmarker}{\pgfqpoint{0.000000in}{0.000000in}}{\pgfqpoint{0.055556in}{0.000000in}}{%
\pgfpathmoveto{\pgfqpoint{0.000000in}{0.000000in}}%
\pgfpathlineto{\pgfqpoint{0.055556in}{0.000000in}}%
\pgfusepath{stroke,fill}%
}%
\begin{pgfscope}%
\pgfsys@transformshift{0.365404in}{1.522752in}%
\pgfsys@useobject{currentmarker}{}%
\end{pgfscope}%
\end{pgfscope}%
\begin{pgfscope}%
\pgfsetbuttcap%
\pgfsetroundjoin%
\definecolor{currentfill}{rgb}{0.000000,0.000000,0.000000}%
\pgfsetfillcolor{currentfill}%
\pgfsetlinewidth{0.501875pt}%
\definecolor{currentstroke}{rgb}{0.000000,0.000000,0.000000}%
\pgfsetstrokecolor{currentstroke}%
\pgfsetdash{}{0pt}%
\pgfsys@defobject{currentmarker}{\pgfqpoint{-0.055556in}{0.000000in}}{\pgfqpoint{0.000000in}{0.000000in}}{%
\pgfpathmoveto{\pgfqpoint{0.000000in}{0.000000in}}%
\pgfpathlineto{\pgfqpoint{-0.055556in}{0.000000in}}%
\pgfusepath{stroke,fill}%
}%
\begin{pgfscope}%
\pgfsys@transformshift{2.630906in}{1.522752in}%
\pgfsys@useobject{currentmarker}{}%
\end{pgfscope}%
\end{pgfscope}%
\begin{pgfscope}%
\pgftext[x=0.309848in,y=1.522752in,right,]{{\rmfamily\fontsize{8.000000}{9.600000}\selectfont \(\displaystyle 0.6\)}}%
\end{pgfscope}%
\begin{pgfscope}%
\pgftext[x=0.089553in,y=0.903327in,,bottom,rotate=90.000000]{{\rmfamily\fontsize{10.000000}{12.000000}\selectfont \(\displaystyle p(s)\)}}%
\end{pgfscope}%
\begin{pgfscope}%
\pgfsetbuttcap%
\pgfsetroundjoin%
\pgfsetlinewidth{1.003750pt}%
\definecolor{currentstroke}{rgb}{0.000000,0.000000,0.000000}%
\pgfsetstrokecolor{currentstroke}%
\pgfsetdash{}{0pt}%
\pgfpathmoveto{\pgfqpoint{0.365404in}{1.625989in}}%
\pgfpathlineto{\pgfqpoint{2.630906in}{1.625989in}}%
\pgfusepath{stroke}%
\end{pgfscope}%
\begin{pgfscope}%
\pgfsetbuttcap%
\pgfsetroundjoin%
\pgfsetlinewidth{1.003750pt}%
\definecolor{currentstroke}{rgb}{0.000000,0.000000,0.000000}%
\pgfsetstrokecolor{currentstroke}%
\pgfsetdash{}{0pt}%
\pgfpathmoveto{\pgfqpoint{2.630906in}{0.180665in}}%
\pgfpathlineto{\pgfqpoint{2.630906in}{1.625989in}}%
\pgfusepath{stroke}%
\end{pgfscope}%
\begin{pgfscope}%
\pgfsetbuttcap%
\pgfsetroundjoin%
\pgfsetlinewidth{1.003750pt}%
\definecolor{currentstroke}{rgb}{0.000000,0.000000,0.000000}%
\pgfsetstrokecolor{currentstroke}%
\pgfsetdash{}{0pt}%
\pgfpathmoveto{\pgfqpoint{0.365404in}{0.180665in}}%
\pgfpathlineto{\pgfqpoint{2.630906in}{0.180665in}}%
\pgfusepath{stroke}%
\end{pgfscope}%
\begin{pgfscope}%
\pgfsetbuttcap%
\pgfsetroundjoin%
\pgfsetlinewidth{1.003750pt}%
\definecolor{currentstroke}{rgb}{0.000000,0.000000,0.000000}%
\pgfsetstrokecolor{currentstroke}%
\pgfsetdash{}{0pt}%
\pgfpathmoveto{\pgfqpoint{0.365404in}{0.180665in}}%
\pgfpathlineto{\pgfqpoint{0.365404in}{1.625989in}}%
\pgfusepath{stroke}%
\end{pgfscope}%
\begin{pgfscope}%
\pgftext[x=1.498155in,y=1.481457in,left,top]{{\rmfamily\fontsize{12.000000}{14.400000}\selectfont \(\displaystyle q=1.5\)}}%
\end{pgfscope}%
\begin{pgfscope}%
\pgftext[x=1.498155in,y=1.336924in,left,top]{{\rmfamily\fontsize{12.000000}{14.400000}\selectfont \(\displaystyle p=0.311575\)}}%
\end{pgfscope}%
\begin{pgfscope}%
\pgftext[x=1.488155in,y=1.162392in,left,top]{{\rmfamily\fontsize{12.000000}{14.400000}\selectfont \(\displaystyle L=64, d=3\) }}%
\end{pgfscope}%
\begin{pgfscope}%
\pgfpathrectangle{\pgfqpoint{0.365404in}{0.180665in}}{\pgfqpoint{2.265502in}{1.445324in}} %
\pgfusepath{clip}%
\pgfsetbuttcap%
\pgfsetroundjoin%
\pgfsetlinewidth{1.003750pt}%
\definecolor{currentstroke}{rgb}{0.000000,0.000000,0.000000}%
\pgfsetstrokecolor{currentstroke}%
\pgfsetdash{}{0pt}%
\pgfpathmoveto{\pgfqpoint{0.358904in}{0.283903in}}%
\pgfpathlineto{\pgfqpoint{0.358904in}{0.288155in}}%
\pgfpathlineto{\pgfqpoint{0.371695in}{0.288155in}}%
\pgfpathlineto{\pgfqpoint{0.371695in}{0.283903in}}%
\pgfpathlineto{\pgfqpoint{0.410068in}{0.283903in}}%
\pgfpathlineto{\pgfqpoint{0.410068in}{0.292408in}}%
\pgfpathlineto{\pgfqpoint{0.422858in}{0.292408in}}%
\pgfpathlineto{\pgfqpoint{0.422858in}{0.288155in}}%
\pgfpathlineto{\pgfqpoint{0.435649in}{0.288155in}}%
\pgfpathlineto{\pgfqpoint{0.435649in}{0.292408in}}%
\pgfpathlineto{\pgfqpoint{0.448440in}{0.292408in}}%
\pgfpathlineto{\pgfqpoint{0.448440in}{0.313670in}}%
\pgfpathlineto{\pgfqpoint{0.474022in}{0.313670in}}%
\pgfpathlineto{\pgfqpoint{0.474022in}{0.322174in}}%
\pgfpathlineto{\pgfqpoint{0.486813in}{0.322174in}}%
\pgfpathlineto{\pgfqpoint{0.486813in}{0.347689in}}%
\pgfpathlineto{\pgfqpoint{0.512395in}{0.347689in}}%
\pgfpathlineto{\pgfqpoint{0.512395in}{0.356193in}}%
\pgfpathlineto{\pgfqpoint{0.525185in}{0.356193in}}%
\pgfpathlineto{\pgfqpoint{0.525185in}{0.368951in}}%
\pgfpathlineto{\pgfqpoint{0.537976in}{0.368951in}}%
\pgfpathlineto{\pgfqpoint{0.537976in}{0.432736in}}%
\pgfpathlineto{\pgfqpoint{0.550767in}{0.432736in}}%
\pgfpathlineto{\pgfqpoint{0.550767in}{0.441241in}}%
\pgfpathlineto{\pgfqpoint{0.563558in}{0.441241in}}%
\pgfpathlineto{\pgfqpoint{0.563558in}{0.492270in}}%
\pgfpathlineto{\pgfqpoint{0.576349in}{0.492270in}}%
\pgfpathlineto{\pgfqpoint{0.576349in}{0.530541in}}%
\pgfpathlineto{\pgfqpoint{0.601931in}{0.530541in}}%
\pgfpathlineto{\pgfqpoint{0.601931in}{0.564561in}}%
\pgfpathlineto{\pgfqpoint{0.614722in}{0.564561in}}%
\pgfpathlineto{\pgfqpoint{0.614722in}{0.645356in}}%
\pgfpathlineto{\pgfqpoint{0.627512in}{0.645356in}}%
\pgfpathlineto{\pgfqpoint{0.627512in}{0.747413in}}%
\pgfpathlineto{\pgfqpoint{0.640303in}{0.747413in}}%
\pgfpathlineto{\pgfqpoint{0.640303in}{0.734656in}}%
\pgfpathlineto{\pgfqpoint{0.653094in}{0.734656in}}%
\pgfpathlineto{\pgfqpoint{0.653094in}{0.819704in}}%
\pgfpathlineto{\pgfqpoint{0.665885in}{0.819704in}}%
\pgfpathlineto{\pgfqpoint{0.665885in}{0.985547in}}%
\pgfpathlineto{\pgfqpoint{0.678676in}{0.985547in}}%
\pgfpathlineto{\pgfqpoint{0.678676in}{0.943023in}}%
\pgfpathlineto{\pgfqpoint{0.691467in}{0.943023in}}%
\pgfpathlineto{\pgfqpoint{0.691467in}{0.874985in}}%
\pgfpathlineto{\pgfqpoint{0.704258in}{0.874985in}}%
\pgfpathlineto{\pgfqpoint{0.704258in}{1.006809in}}%
\pgfpathlineto{\pgfqpoint{0.717048in}{1.006809in}}%
\pgfpathlineto{\pgfqpoint{0.717048in}{1.087604in}}%
\pgfpathlineto{\pgfqpoint{0.729839in}{1.087604in}}%
\pgfpathlineto{\pgfqpoint{0.729839in}{1.019566in}}%
\pgfpathlineto{\pgfqpoint{0.742630in}{1.019566in}}%
\pgfpathlineto{\pgfqpoint{0.742630in}{1.181157in}}%
\pgfpathlineto{\pgfqpoint{0.755421in}{1.181157in}}%
\pgfpathlineto{\pgfqpoint{0.755421in}{1.164147in}}%
\pgfpathlineto{\pgfqpoint{0.768212in}{1.164147in}}%
\pgfpathlineto{\pgfqpoint{0.768212in}{1.104614in}}%
\pgfpathlineto{\pgfqpoint{0.781003in}{1.104614in}}%
\pgfpathlineto{\pgfqpoint{0.781003in}{1.172652in}}%
\pgfpathlineto{\pgfqpoint{0.793794in}{1.172652in}}%
\pgfpathlineto{\pgfqpoint{0.793794in}{1.202419in}}%
\pgfpathlineto{\pgfqpoint{0.806585in}{1.202419in}}%
\pgfpathlineto{\pgfqpoint{0.806585in}{1.176905in}}%
\pgfpathlineto{\pgfqpoint{0.819375in}{1.176905in}}%
\pgfpathlineto{\pgfqpoint{0.819375in}{1.215176in}}%
\pgfpathlineto{\pgfqpoint{0.832166in}{1.215176in}}%
\pgfpathlineto{\pgfqpoint{0.832166in}{1.240690in}}%
\pgfpathlineto{\pgfqpoint{0.844957in}{1.240690in}}%
\pgfpathlineto{\pgfqpoint{0.844957in}{1.308729in}}%
\pgfpathlineto{\pgfqpoint{0.857748in}{1.308729in}}%
\pgfpathlineto{\pgfqpoint{0.857748in}{1.176905in}}%
\pgfpathlineto{\pgfqpoint{0.870539in}{1.176905in}}%
\pgfpathlineto{\pgfqpoint{0.870539in}{1.181157in}}%
\pgfpathlineto{\pgfqpoint{0.883330in}{1.181157in}}%
\pgfpathlineto{\pgfqpoint{0.883330in}{1.274710in}}%
\pgfpathlineto{\pgfqpoint{0.896121in}{1.274710in}}%
\pgfpathlineto{\pgfqpoint{0.896121in}{1.257700in}}%
\pgfpathlineto{\pgfqpoint{0.908912in}{1.257700in}}%
\pgfpathlineto{\pgfqpoint{0.908912in}{1.219428in}}%
\pgfpathlineto{\pgfqpoint{0.921702in}{1.219428in}}%
\pgfpathlineto{\pgfqpoint{0.921702in}{1.261952in}}%
\pgfpathlineto{\pgfqpoint{0.934493in}{1.261952in}}%
\pgfpathlineto{\pgfqpoint{0.934493in}{1.176905in}}%
\pgfpathlineto{\pgfqpoint{0.947284in}{1.176905in}}%
\pgfpathlineto{\pgfqpoint{0.947284in}{1.219428in}}%
\pgfpathlineto{\pgfqpoint{0.960075in}{1.219428in}}%
\pgfpathlineto{\pgfqpoint{0.960075in}{1.108866in}}%
\pgfpathlineto{\pgfqpoint{0.972866in}{1.108866in}}%
\pgfpathlineto{\pgfqpoint{0.972866in}{1.130128in}}%
\pgfpathlineto{\pgfqpoint{0.985657in}{1.130128in}}%
\pgfpathlineto{\pgfqpoint{0.985657in}{1.100362in}}%
\pgfpathlineto{\pgfqpoint{0.998448in}{1.100362in}}%
\pgfpathlineto{\pgfqpoint{0.998448in}{1.117371in}}%
\pgfpathlineto{\pgfqpoint{1.011239in}{1.117371in}}%
\pgfpathlineto{\pgfqpoint{1.011239in}{1.002557in}}%
\pgfpathlineto{\pgfqpoint{1.024029in}{1.002557in}}%
\pgfpathlineto{\pgfqpoint{1.024029in}{1.066342in}}%
\pgfpathlineto{\pgfqpoint{1.036820in}{1.066342in}}%
\pgfpathlineto{\pgfqpoint{1.036820in}{0.998304in}}%
\pgfpathlineto{\pgfqpoint{1.049611in}{0.998304in}}%
\pgfpathlineto{\pgfqpoint{1.049611in}{0.891994in}}%
\pgfpathlineto{\pgfqpoint{1.062402in}{0.891994in}}%
\pgfpathlineto{\pgfqpoint{1.062402in}{0.951528in}}%
\pgfpathlineto{\pgfqpoint{1.075193in}{0.951528in}}%
\pgfpathlineto{\pgfqpoint{1.075193in}{0.960033in}}%
\pgfpathlineto{\pgfqpoint{1.087984in}{0.960033in}}%
\pgfpathlineto{\pgfqpoint{1.087984in}{0.947276in}}%
\pgfpathlineto{\pgfqpoint{1.100775in}{0.947276in}}%
\pgfpathlineto{\pgfqpoint{1.100775in}{0.870733in}}%
\pgfpathlineto{\pgfqpoint{1.113565in}{0.870733in}}%
\pgfpathlineto{\pgfqpoint{1.113565in}{0.857975in}}%
\pgfpathlineto{\pgfqpoint{1.126356in}{0.857975in}}%
\pgfpathlineto{\pgfqpoint{1.126356in}{0.819704in}}%
\pgfpathlineto{\pgfqpoint{1.139147in}{0.819704in}}%
\pgfpathlineto{\pgfqpoint{1.139147in}{0.785685in}}%
\pgfpathlineto{\pgfqpoint{1.151938in}{0.785685in}}%
\pgfpathlineto{\pgfqpoint{1.151938in}{0.768675in}}%
\pgfpathlineto{\pgfqpoint{1.164729in}{0.768675in}}%
\pgfpathlineto{\pgfqpoint{1.164729in}{0.785685in}}%
\pgfpathlineto{\pgfqpoint{1.177520in}{0.785685in}}%
\pgfpathlineto{\pgfqpoint{1.177520in}{0.696385in}}%
\pgfpathlineto{\pgfqpoint{1.190311in}{0.696385in}}%
\pgfpathlineto{\pgfqpoint{1.190311in}{0.675123in}}%
\pgfpathlineto{\pgfqpoint{1.203102in}{0.675123in}}%
\pgfpathlineto{\pgfqpoint{1.203102in}{0.683627in}}%
\pgfpathlineto{\pgfqpoint{1.215892in}{0.683627in}}%
\pgfpathlineto{\pgfqpoint{1.215892in}{0.679375in}}%
\pgfpathlineto{\pgfqpoint{1.228683in}{0.679375in}}%
\pgfpathlineto{\pgfqpoint{1.228683in}{0.628346in}}%
\pgfpathlineto{\pgfqpoint{1.241474in}{0.628346in}}%
\pgfpathlineto{\pgfqpoint{1.241474in}{0.615589in}}%
\pgfpathlineto{\pgfqpoint{1.254265in}{0.615589in}}%
\pgfpathlineto{\pgfqpoint{1.254265in}{0.590075in}}%
\pgfpathlineto{\pgfqpoint{1.267056in}{0.590075in}}%
\pgfpathlineto{\pgfqpoint{1.267056in}{0.607084in}}%
\pgfpathlineto{\pgfqpoint{1.279847in}{0.607084in}}%
\pgfpathlineto{\pgfqpoint{1.279847in}{0.577318in}}%
\pgfpathlineto{\pgfqpoint{1.292638in}{0.577318in}}%
\pgfpathlineto{\pgfqpoint{1.292638in}{0.534794in}}%
\pgfpathlineto{\pgfqpoint{1.305429in}{0.534794in}}%
\pgfpathlineto{\pgfqpoint{1.305429in}{0.492270in}}%
\pgfpathlineto{\pgfqpoint{1.318219in}{0.492270in}}%
\pgfpathlineto{\pgfqpoint{1.318219in}{0.598580in}}%
\pgfpathlineto{\pgfqpoint{1.331010in}{0.598580in}}%
\pgfpathlineto{\pgfqpoint{1.331010in}{0.471008in}}%
\pgfpathlineto{\pgfqpoint{1.343801in}{0.471008in}}%
\pgfpathlineto{\pgfqpoint{1.343801in}{0.458251in}}%
\pgfpathlineto{\pgfqpoint{1.356592in}{0.458251in}}%
\pgfpathlineto{\pgfqpoint{1.356592in}{0.436989in}}%
\pgfpathlineto{\pgfqpoint{1.369383in}{0.436989in}}%
\pgfpathlineto{\pgfqpoint{1.369383in}{0.445494in}}%
\pgfpathlineto{\pgfqpoint{1.382174in}{0.445494in}}%
\pgfpathlineto{\pgfqpoint{1.382174in}{0.449746in}}%
\pgfpathlineto{\pgfqpoint{1.394965in}{0.449746in}}%
\pgfpathlineto{\pgfqpoint{1.394965in}{0.560308in}}%
\pgfpathlineto{\pgfqpoint{1.407755in}{0.560308in}}%
\pgfpathlineto{\pgfqpoint{1.407755in}{0.411475in}}%
\pgfpathlineto{\pgfqpoint{1.420546in}{0.411475in}}%
\pgfpathlineto{\pgfqpoint{1.420546in}{0.428484in}}%
\pgfpathlineto{\pgfqpoint{1.433337in}{0.428484in}}%
\pgfpathlineto{\pgfqpoint{1.433337in}{0.432736in}}%
\pgfpathlineto{\pgfqpoint{1.446128in}{0.432736in}}%
\pgfpathlineto{\pgfqpoint{1.446128in}{0.390213in}}%
\pgfpathlineto{\pgfqpoint{1.458919in}{0.390213in}}%
\pgfpathlineto{\pgfqpoint{1.458919in}{0.419979in}}%
\pgfpathlineto{\pgfqpoint{1.471710in}{0.419979in}}%
\pgfpathlineto{\pgfqpoint{1.471710in}{0.381708in}}%
\pgfpathlineto{\pgfqpoint{1.484501in}{0.381708in}}%
\pgfpathlineto{\pgfqpoint{1.484501in}{0.394465in}}%
\pgfpathlineto{\pgfqpoint{1.497292in}{0.394465in}}%
\pgfpathlineto{\pgfqpoint{1.497292in}{0.398717in}}%
\pgfpathlineto{\pgfqpoint{1.510082in}{0.398717in}}%
\pgfpathlineto{\pgfqpoint{1.510082in}{0.394465in}}%
\pgfpathlineto{\pgfqpoint{1.535664in}{0.394465in}}%
\pgfpathlineto{\pgfqpoint{1.535664in}{0.351941in}}%
\pgfpathlineto{\pgfqpoint{1.548455in}{0.351941in}}%
\pgfpathlineto{\pgfqpoint{1.548455in}{0.356193in}}%
\pgfpathlineto{\pgfqpoint{1.561246in}{0.356193in}}%
\pgfpathlineto{\pgfqpoint{1.561246in}{0.381708in}}%
\pgfpathlineto{\pgfqpoint{1.574037in}{0.381708in}}%
\pgfpathlineto{\pgfqpoint{1.574037in}{0.330679in}}%
\pgfpathlineto{\pgfqpoint{1.586828in}{0.330679in}}%
\pgfpathlineto{\pgfqpoint{1.586828in}{0.322174in}}%
\pgfpathlineto{\pgfqpoint{1.599619in}{0.322174in}}%
\pgfpathlineto{\pgfqpoint{1.599619in}{0.360446in}}%
\pgfpathlineto{\pgfqpoint{1.612409in}{0.360446in}}%
\pgfpathlineto{\pgfqpoint{1.612409in}{0.339184in}}%
\pgfpathlineto{\pgfqpoint{1.625200in}{0.339184in}}%
\pgfpathlineto{\pgfqpoint{1.625200in}{0.347689in}}%
\pgfpathlineto{\pgfqpoint{1.637991in}{0.347689in}}%
\pgfpathlineto{\pgfqpoint{1.637991in}{0.356193in}}%
\pgfpathlineto{\pgfqpoint{1.650782in}{0.356193in}}%
\pgfpathlineto{\pgfqpoint{1.650782in}{0.317922in}}%
\pgfpathlineto{\pgfqpoint{1.663573in}{0.317922in}}%
\pgfpathlineto{\pgfqpoint{1.663573in}{0.356193in}}%
\pgfpathlineto{\pgfqpoint{1.676364in}{0.356193in}}%
\pgfpathlineto{\pgfqpoint{1.676364in}{0.317922in}}%
\pgfpathlineto{\pgfqpoint{1.689155in}{0.317922in}}%
\pgfpathlineto{\pgfqpoint{1.689155in}{0.313670in}}%
\pgfpathlineto{\pgfqpoint{1.701945in}{0.313670in}}%
\pgfpathlineto{\pgfqpoint{1.701945in}{0.322174in}}%
\pgfpathlineto{\pgfqpoint{1.714736in}{0.322174in}}%
\pgfpathlineto{\pgfqpoint{1.714736in}{0.313670in}}%
\pgfpathlineto{\pgfqpoint{1.727527in}{0.313670in}}%
\pgfpathlineto{\pgfqpoint{1.727527in}{0.351941in}}%
\pgfpathlineto{\pgfqpoint{1.740318in}{0.351941in}}%
\pgfpathlineto{\pgfqpoint{1.740318in}{0.300912in}}%
\pgfpathlineto{\pgfqpoint{1.753109in}{0.300912in}}%
\pgfpathlineto{\pgfqpoint{1.753109in}{0.317922in}}%
\pgfpathlineto{\pgfqpoint{1.765900in}{0.317922in}}%
\pgfpathlineto{\pgfqpoint{1.765900in}{0.313670in}}%
\pgfpathlineto{\pgfqpoint{1.778691in}{0.313670in}}%
\pgfpathlineto{\pgfqpoint{1.778691in}{0.322174in}}%
\pgfpathlineto{\pgfqpoint{1.791482in}{0.322174in}}%
\pgfpathlineto{\pgfqpoint{1.791482in}{0.305165in}}%
\pgfpathlineto{\pgfqpoint{1.817063in}{0.305165in}}%
\pgfpathlineto{\pgfqpoint{1.817063in}{0.313670in}}%
\pgfpathlineto{\pgfqpoint{1.829854in}{0.313670in}}%
\pgfpathlineto{\pgfqpoint{1.829854in}{0.305165in}}%
\pgfpathlineto{\pgfqpoint{1.842645in}{0.305165in}}%
\pgfpathlineto{\pgfqpoint{1.842645in}{0.292408in}}%
\pgfpathlineto{\pgfqpoint{1.855436in}{0.292408in}}%
\pgfpathlineto{\pgfqpoint{1.855436in}{0.283903in}}%
\pgfpathlineto{\pgfqpoint{1.868227in}{0.283903in}}%
\pgfpathlineto{\pgfqpoint{1.868227in}{0.300912in}}%
\pgfpathlineto{\pgfqpoint{1.881018in}{0.300912in}}%
\pgfpathlineto{\pgfqpoint{1.881018in}{0.305165in}}%
\pgfpathlineto{\pgfqpoint{1.893809in}{0.305165in}}%
\pgfpathlineto{\pgfqpoint{1.893809in}{0.317922in}}%
\pgfpathlineto{\pgfqpoint{1.906599in}{0.317922in}}%
\pgfpathlineto{\pgfqpoint{1.906599in}{0.309417in}}%
\pgfpathlineto{\pgfqpoint{1.919390in}{0.309417in}}%
\pgfpathlineto{\pgfqpoint{1.919390in}{0.292408in}}%
\pgfpathlineto{\pgfqpoint{1.932181in}{0.292408in}}%
\pgfpathlineto{\pgfqpoint{1.932181in}{0.296660in}}%
\pgfpathlineto{\pgfqpoint{1.944972in}{0.296660in}}%
\pgfpathlineto{\pgfqpoint{1.944972in}{0.300912in}}%
\pgfpathlineto{\pgfqpoint{1.957763in}{0.300912in}}%
\pgfpathlineto{\pgfqpoint{1.957763in}{0.283903in}}%
\pgfpathlineto{\pgfqpoint{1.970554in}{0.283903in}}%
\pgfpathlineto{\pgfqpoint{1.970554in}{0.288155in}}%
\pgfpathlineto{\pgfqpoint{1.983345in}{0.288155in}}%
\pgfpathlineto{\pgfqpoint{1.983345in}{0.292408in}}%
\pgfpathlineto{\pgfqpoint{1.996135in}{0.292408in}}%
\pgfpathlineto{\pgfqpoint{1.996135in}{0.300912in}}%
\pgfpathlineto{\pgfqpoint{2.008926in}{0.300912in}}%
\pgfpathlineto{\pgfqpoint{2.008926in}{0.288155in}}%
\pgfpathlineto{\pgfqpoint{2.021717in}{0.288155in}}%
\pgfpathlineto{\pgfqpoint{2.021717in}{0.300912in}}%
\pgfpathlineto{\pgfqpoint{2.034508in}{0.300912in}}%
\pgfpathlineto{\pgfqpoint{2.034508in}{0.292408in}}%
\pgfpathlineto{\pgfqpoint{2.047299in}{0.292408in}}%
\pgfpathlineto{\pgfqpoint{2.047299in}{0.288155in}}%
\pgfpathlineto{\pgfqpoint{2.072881in}{0.288155in}}%
\pgfpathlineto{\pgfqpoint{2.072881in}{0.309417in}}%
\pgfpathlineto{\pgfqpoint{2.085672in}{0.309417in}}%
\pgfpathlineto{\pgfqpoint{2.085672in}{0.288155in}}%
\pgfpathlineto{\pgfqpoint{2.149626in}{0.288155in}}%
\pgfpathlineto{\pgfqpoint{2.149626in}{0.296660in}}%
\pgfpathlineto{\pgfqpoint{2.162417in}{0.296660in}}%
\pgfpathlineto{\pgfqpoint{2.162417in}{0.288155in}}%
\pgfpathlineto{\pgfqpoint{2.213580in}{0.288155in}}%
\pgfpathlineto{\pgfqpoint{2.213580in}{0.283903in}}%
\pgfpathlineto{\pgfqpoint{2.290325in}{0.283903in}}%
\pgfpathlineto{\pgfqpoint{2.290325in}{0.292408in}}%
\pgfpathlineto{\pgfqpoint{2.303116in}{0.292408in}}%
\pgfpathlineto{\pgfqpoint{2.303116in}{0.288155in}}%
\pgfpathlineto{\pgfqpoint{2.328698in}{0.288155in}}%
\pgfpathlineto{\pgfqpoint{2.328698in}{0.283903in}}%
\pgfpathlineto{\pgfqpoint{2.341489in}{0.283903in}}%
\pgfpathlineto{\pgfqpoint{2.341489in}{0.296660in}}%
\pgfpathlineto{\pgfqpoint{2.354280in}{0.296660in}}%
\pgfpathlineto{\pgfqpoint{2.354280in}{0.283903in}}%
\pgfpathlineto{\pgfqpoint{2.367071in}{0.283903in}}%
\pgfpathlineto{\pgfqpoint{2.367071in}{0.288155in}}%
\pgfpathlineto{\pgfqpoint{2.379862in}{0.288155in}}%
\pgfpathlineto{\pgfqpoint{2.379862in}{0.283903in}}%
\pgfpathlineto{\pgfqpoint{2.431025in}{0.283903in}}%
\pgfpathlineto{\pgfqpoint{2.431025in}{0.288155in}}%
\pgfpathlineto{\pgfqpoint{2.443816in}{0.288155in}}%
\pgfpathlineto{\pgfqpoint{2.443816in}{0.283903in}}%
\pgfpathlineto{\pgfqpoint{2.469398in}{0.283903in}}%
\pgfpathlineto{\pgfqpoint{2.469398in}{0.292408in}}%
\pgfpathlineto{\pgfqpoint{2.482189in}{0.292408in}}%
\pgfpathlineto{\pgfqpoint{2.482189in}{0.283903in}}%
\pgfpathlineto{\pgfqpoint{2.546143in}{0.283903in}}%
\pgfpathlineto{\pgfqpoint{2.546143in}{0.288155in}}%
\pgfpathlineto{\pgfqpoint{2.558934in}{0.288155in}}%
\pgfpathlineto{\pgfqpoint{2.558934in}{0.283903in}}%
\pgfpathlineto{\pgfqpoint{2.640906in}{0.283903in}}%
\pgfpathlineto{\pgfqpoint{2.640906in}{0.283903in}}%
\pgfusepath{stroke}%
\end{pgfscope}%
\end{pgfpicture}%
\makeatother%
\endgroup%

%% file: manuscript.bbl
\begin{thebibliography}{10}
\providecommand{\url}[1]{{#1}}
\providecommand{\urlprefix}{URL }
\expandafter\ifx\csname urlstyle\endcsname\relax
  \providecommand{\doi}[1]{DOI~\discretionary{}{}{}#1}\else
  \providecommand{\doi}{DOI~\discretionary{}{}{}\begingroup
  \urlstyle{rm}\Url}\fi

\bibitem{AizenmanDuminilCopinSidoravicius15}
Aizenman, M., Duminil-Copin, H., Sidoravicius, V.: {Random Currents and
  Continuity of Ising Model{\textquoteright}s Spontaneous Magnetization}.
\newblock Communications in Mathematical Physics \textbf{334}, 719--742 (2015)

\bibitem{Baxter78}
Baxter, R.J.: {Solvable eight-vertex model on an arbitrary planar lattice}.
\newblock Philosophical Transactions of the Royal Society A \textbf{289}, 315--346 (1978)

\bibitem{BeffaraDuminilCopin12}
Beffara, V., Duminil-Copin, H.: {The self-dual point of the two-dimensional
  random-cluster model is critical for q~$\ge$ 1}.
\newblock Probability Theory and Related Fields \textbf{153}, 511--542 (2012)

\bibitem{Billingsley94}
Billingsley, P.: {Probability and Measure, 3rd ed. (Wiley Series in Probability
  and Statistics)}, 3 edn.
\newblock Wiley, New York (1994)

\bibitem{CesiGuadagniMartinelliSchonmann96}
Cesi, F., Guadagni, G., Martinelli, F., Schonmann, R.H.: {On the
  Two-Dimensional Stochastic Ising Model in the Phase Coexistence Region Near
  the Critical Point}.
\newblock Journal of Statistical Physics \textbf{85}, 55--102 (1996)

\bibitem{ChayesMachta98}
Chayes, L., Machta, J.: {Graphical representations and cluster algorithms II}.
\newblock Physica A \textbf{254}, 477--516 (1998)

\bibitem{ChowTeicher78}
Chow, Y.S., Teicher, H.: {Probability Theory: Independence, Interchangeability,
  Martingales}.
\newblock Springer, New York (1978)

\bibitem{DengBlote03}
Deng, Y., Bl{\"o}te, H.: {Simultaneous analysis of several models in the
  three-dimensional Ising universality class}.
\newblock Physical Review E \textbf{68}, 036,125 (2003)

\bibitem{DengGaroniMachtaOssolaPolinSokal07}
Deng, Y., Garoni, T., Machta, J., Ossola, G., Polin, M., Sokal, A.: {Critical
  Behavior of the
  Chayes{\textendash}Machta{\textendash}Swendsen{\textendash}Wang Dynamics}.
\newblock Physical Review Letters \textbf{99}, 055,701 (2007)

\bibitem{DengGaroniSokal07_sweeny}
Deng, Y., Garoni, T.M., Sokal, A.D.: {Critical speeding-up in the local
  dynamics of the random-cluster model}.
\newblock Physical Review Letters \textbf{98}, 230,602 (2007)

\bibitem{DuminilCopinGagnebinHarelManolescuTassion16}
Duminil-Copin, H., Gagnebin, M., Harel, M., Manolescu, I., Tassion, V.:
  {Discontinuity of the phase transition for the planar random-cluster and
  Potts models with $q > 4$}.
\newblock arXiv:1611.09877  (2016)

\bibitem{DuminilCopinSidoraviciusTassion15}
Duminil-Copin, H., Sidoravicius, V., Tassion, V.: {Continuity of the phase
  transition for planar random-cluster and Potts models with $1 \le q \le 4$}.
\newblock arXiv:1602.05677  (2016)

\bibitem{DyerGreenhillUllrich14}
Dyer, M., Greenhill, C., Ullrich, M.: {Structure and eigenvalues of heat-bath
  Markov chains}.
\newblock Linear Algebra and its Applications \textbf{454}, 57--71 (2014)

\bibitem{Elci15_thesis}
Elci, E.: {Algorithmic and geometric aspects of the random-cluster model}.
\newblock Ph.D. thesis (2015)

\bibitem{ElciWeigel13}
El{\c c}i, E.M., Weigel, M.: {Efficient simulation of the random-cluster
  model}.
\newblock Physical Review E \textbf{88}, 033,303 (2013)

\bibitem{ElciWeigel14}
El{\c c}i, E.M., Weigel, M.: {Dynamic connectivity algorithms for Monte Carlo
  simulations of the random-cluster model}.
\newblock Journal of Physics: Conference Series \textbf{510}, 012,013--10
  (2014)

\bibitem{ErdosRenyi61}
Erdos, P., Renyi, A.: {On a classical problem of probability theory}.
\newblock Publ. Math. Inst. Hung. Acad. Sci., Ser. A \textbf{6}, 215--219
  (1961)

\bibitem{Feller68}
Feller, W.: {An Introduction to Probability Theory and Its Applications},
  vol.~1, 3 edn.
\newblock John Wiley {\&} Sons Inc, New York (1968)

\bibitem{FriedliVelenik16}
Friedli, S., Velenik, Y.: {Statistical Mechanics of Lattice Systems}.
\newblock A Concrete Mathematical Introduction. Cambridge University Press,
  Cambridge (2016)

\bibitem{GheissariLubetzky16}
Gheissari, R., Lubetzky, E.: {Mixing Times Of Critical 2D Potts Models}.
\newblock arXiv:1607.02182  (2016)

\bibitem{Gliozzi02}
Gliozzi, F.: {Simulation of Potts models with real q and no critical slowing
  down}.
\newblock Physical Review E \textbf{66}, 016,115 (2002)

\bibitem{GrahamKnuthPatashnik94}
Graham, R.L., Knuth, D.E., Patashnik, O.: {CONCRETE MATHEMATICS}, 2 edn.
\newblock A Foundation for Computer Science. Addisen-Wesley Publishing Company
  (1994)

\bibitem{Grassberger95}
Grassberger, P.: {Damage spreading and critical exponents for
  {\textquotedblleft}model A{\textquotedblright} Ising dynamics}.
\newblock Physica A \textbf{214}, 547--559 (1995)

\bibitem{Grimmett06}
Grimmett, G.: {The Random-Cluster Model}.
\newblock Springer, New York (2006)

\bibitem{Grimmett10}
Grimmett, G.: {Probability on Graphs}.
\newblock Cambridge University Press, Cambridge (2010)

\bibitem{GuoJerrum16}
Guo, H., Jerrum, M.: {Random cluster dynamics for the Ising model is rapidly
  mixing}.
\newblock arXiv:1605.00139 pp. 1--15 (2016)

\bibitem{Haggstrom03}
H{\"a}ggstr{\"o}m, O.: {Finite Markov Chains and Algorithmic Applications}.
\newblock Cambridge University Press, Cambridge (2003)

\bibitem{Hartmann05}
Hartmann, A.: {Calculation of Partition Functions by Measuring Component
  Distributions}.
\newblock Physical Review Letters \textbf{94}, 050,601 (2005)

\bibitem{HolmDeLichtenbergThorup01}
Holm, J., de~Lichtenberg, K., Thorup, M.: {Poly-logarithmic deterministic
  fully-dynamic algorithms for connectivity, minimum spanning tree, 2-edge, and
  biconnectivity}.
\newblock Journal of the ACM (JACM) \textbf{48}, 723--760 (2001)

\bibitem{JaegerVertiganWelsh90}
Jaeger, F., Vertigan, D.L., Welsh, D.J.A.: {On the computational complexity of
  the Jones and Tutte polynomials}.
\newblock Mathematical Proceedings of the Cambridge Philosophical Society
  \textbf{108}, 35--53 (1990)

\bibitem{Janson14}
Janson, S.: {Tail bounds for sums of geometric and exponential random
  variables} (2014).
\newblock \urlprefix\url{http://www2.math.uu.se/~svante/papers/sjN14.pdf}

\bibitem{Jerrum98}
Jerrum, M.: {Mathematical Foundations of the Markov Chain Monte Carlo Method}.
\newblock In: Probabilistic Methods for Algorithmic Discrete Mathematics, pp.
  116--165. Springer, New York (1998)

\bibitem{LaanaitMessagerMiracleSoleRuizShlosman91}
Laanait, L., Messager, A., Miracle-Sol{\'e}, S., Ruiz, J., Shlosman, S.:
  {Interfaces in the Potts model I: Pirogov-Sinai theory of the
  Fortuin-Kasteleyn representation}.
\newblock Communications in Mathematical Physics \textbf{140}(1), 81--91 (1991)

\bibitem{LeadbetterLindgrenRootzen83}
Leadbetter, M.R., Lindgren, G., Rootzen, H.: {Extremes and related properties
  of random sequences and processes}.
\newblock Springer, New York (1983)

\bibitem{LevinPeresWilmer09}
Levin, D.A., Peres, Y., Wilmer, E.L.: {Markov Chains and Mixing Times}.
\newblock American Mathematical Society, Providence (2009)

\bibitem{MadrasSlade96}
Madras, N., Slade, G.: {The Self-Avoiding Walk}.
\newblock Birkhauser, Boston (1996)

\bibitem{McCoyWu73}
McCoy, B.M., Wu, T.T.: {The Two-Dimensional Ising Model}.
\newblock Harvard University Press, Cambridge (1973)

\bibitem{MitzenmacherUpfal05}
Mitzenmacher, M., Upfal, E.: {Probability and Computing}.
\newblock Randomized Algorithms and Probabilistic Analysis. Cambridge
  University Press, Cambridge (2005)

\bibitem{Nacu03}
Nacu, S.: {Glauber dynamics on the cycle is monotone}.
\newblock Probability Theory and Related Fields \textbf{127}, 177--185 (2003)

\bibitem{NienhuisJSP84}
Nienhuis, B.: {Critical behavior of two-dimensional spin models and charge
  asymmetry in the Coulomb gas}.
\newblock Journal of Statistical Physics \textbf{34}, 731--761 (1984)

\bibitem{NightingaleBloete96}
Nightingale, M.P., Bloete, H.W.J.: {Dynamic Exponent of the Two-Dimensional
  Ising Model and Monte Carlo Computation of the Subdominant Eigenvalue of the
  Stochastic Matrix }.
\newblock Physical Review Letters \textbf{76}, 4548--4551 (1996)

\bibitem{Posfai10}
Posfai, A.: {Approximation Theorems Related to the Coupon
  Collector{\textquoteright}s Problem }.
\newblock Ph.D. thesis (2010)

\bibitem{ProppWilson96}
Propp, J., Wilson, D.: {Exact sampling with coupled Markov chains and
  applications to statistical mechanics}.
\newblock Random Structures and Algorithms \textbf{9}, 223--252 (1996)

\bibitem{BlancaSinclair16}
Sinclair, A.B., Alistair, Sinclair, A.: {Random-Cluster Dynamics in
  $\mathbb{Z}^2$}.
\newblock In: Twenty-Seventh Annual ACM-SIAM Symposium on Discrete Algorithms,
  pp. 498--513 (2016)

\bibitem{Sokal97}
Sokal, A.D.: {Monte Carlo Methods in Statistical Mechanics: Foundations and New
  Algorithms}.
\newblock In: C.~DeWitt-Morette, P.~Cartier, A.~Folacci (eds.) Functional
  Integration: Basics and Applications (1996 Carg{\`e}se summer school), pp.
  131--192. Plenum, New York (1997)

\bibitem{Sweeny83}
Sweeny, M.: {Monte Carlo study of weighted percolation clusters relevant to the
  Potts models}.
\newblock Physical Review B \textbf{27}, 4445--4455 (1983)

\bibitem{SwendsenWang87}
Swendsen, R.H., Wang, J.S.: {Nonuniversal critical dynamics in Monte Carlo
  simulations}.
\newblock Physical Review Letters \textbf{58}, 86--88 (1987)

\bibitem{WangKozanSwendsen02}
Wang, J.S., Kozan, O., Swendsen, R.: {Sweeny and Gliozzi dynamics for
  simulations of Potts models in the Fortuin-Kasteleyn representation}.
\newblock Physical Review E \textbf{66} (2002)

\bibitem{Welsh93}
Welsh, D.J.A.: {Complexity: Knots, Colourings and Counting}, \emph{London
  Mathematical Society Lecture Note Series}, vol. 186.
\newblock Cambridge University Press (1993)

\bibitem{DengGaroniSokalZhou}
{Youjin Deng, Timothy M. Garoni, Alan Sokal and ZongzhengZhou}: {Dynamic
  critical behavior of the Chayes-Machta random-cluster algorithm II:
  Three-dimensions}.
\newblock In preparation

\bibitem{Young15}
Young, P.: {Everything you wanted to know about Data Analysis and Fitting but
  were afraid to ask }.
\newblock Springer, New York (2015)

\end{thebibliography}
